\documentclass[12pt]{article}

\usepackage{verbatim,color,amssymb,bbm,bm}
\usepackage{amsmath}					
\usepackage{mathabx}
\usepackage{amsthm}
\usepackage{array}
\usepackage{natbib}
\usepackage{multirow}
\usepackage{breqn}
\usepackage{setspace}
\usepackage[mathscr]{euscript}
\usepackage{fancyhdr}
\usepackage{enumitem}
\usepackage{graphicx,grffile}
\usepackage[letterpaper, left=1.15in,right=1.05in,top=1in, bottom=1in]{geometry}
\usepackage[dvipsnames]{xcolor}
\usepackage{physics,xfrac}
\usepackage[colorlinks,linkcolor=violet,citecolor=blue]{hyperref}
\usepackage{caption,subcaption}
\usepackage[linesnumbered,ruled,vlined]{algorithm2e}
\usepackage{algpseudocode}

\graphicspath{{Figures/}}

\usepackage{multibib}
\newcites{latex}{References}

\usepackage{tabularx, booktabs}
\newcolumntype{Y}{>{\centering\arraybackslash}X}
\usepackage{lscape}

\usepackage{float}
\usepackage{lineno}

\newtheorem{lemma}{Lemma}
\newtheorem{theorem}{Theorem}

\newtheorem{corollary}{Corollary}
\newtheorem*{theorem*}{Theorem}

\def\pP{\mathbb{P}}  
\def\nN{\mathbb{N}}
\def\rR{\mathbb{R}}
\def\eE{\mathbb{E}}

\newcommand{\cspace}[1]{\mathbb{C}\left( #1\right)}
\newcommand{\nullity}[1]{\mathbb{N}\left( #1\right)}

\def\L{{\cal L}}

\renewcommand{\O}{\mathcal{O}}

\def\diag{\mathrm{diag}}

\def\wh{\widehat}
\def\wt{\widetilde}

\def\var{\mathrm{var}}
\def\cov{\mathrm{cov}}

\def\trace{\mathrm{trace}}

\def\vect{\mathrm{vec}}
\def\vmax{\mathrm{VARIMAX}}

\def\DE{\mathrm{DE}}
\def\Dir{\hbox{Dir}}

\def\Exp{\hbox{Exp}}

\def\Ga{\hbox{Ga}}

\newcommand{\mn}{\mathrm{N}}
\def\Poisson{\mathrm{Poisson}}

\def\Unif{\mathrm{Unif}}

\newcommand{\DL}{\mathrm{DL}}

\def\P_25_ICML{{\it Proceedings of the 25th international conference on Machine learning}}

\def\bse{\begin{eqnarray*}}
	\def\ese{\end{eqnarray*}}
\def\be{\begin{eqnarray}}
	\def\ee{\end{eqnarray}}
\def\bq{\begin{equation}}
	\def\eq{\end{equation}}

\def\wh{\widehat}

\def\trans{^{\rm T}}
\def\inv{^{-1}}

\def\sgn{{\mathrm{sign}}}

\def\th{^{th}}

\def\bA{{\mathbf A}}

\def\bB{{\mathbf B}}

\def\bC{{\mathbf C}}

\def\bD{{\mathbf D}}
\def\b1e{{\mathbf e}}
\def\bE{{\mathbf E}}
\def\b1f{{\mathbf f}}
\def\bF{{\mathbf F}}
\def\bG{{\mathbf G}}

\def\bh{{\mathbf h}}
\def\bH{{\mathbf H}}
\def\bI{{\mathbf I}}

\def\bM{{\mathbf M}}

\def\bp{{\mathbf p}}

\def\bQ{{\mathbf Q}}

\def\bR{{\mathbf R}}

\def\bW{{\mathbf W}}
\def\bx{{\mathbf x}}

\def\bY{{\mathbf Y}}
\def\bz{{\mathbf z}}

\def\bzero{{\mathbf 0}}

\def\simind{\stackrel{\mathrm{\scriptsize{ind}}}{\sim}}
\def\simiid{\stackrel{\mathrm{\scriptsize{iid}}}{\sim}}

\newcommand{\boldeta}{ \bm{\eta}} 
\newcommand{\bmu}{\mbox{\boldmath $\mu$}}

\newcommand{\bdelta}{\bm{\delta}}
\newcommand{\bDelta}{\bm{\Delta}}
\newcommand{\bphi}{\bm{\phi}}
\newcommand{\bPhi}{\bm{\Phi}}

\newcommand{\bUpsilon}{\mbox{\boldmath $\Upsilon$}}
\newcommand{\bepsilon}{\mbox{\boldmath $\epsilon$}}
\newcommand{\btheta}{\mbox{\boldmath $\theta$}}
\newcommand{\bTheta}{\bm{\Theta}}

\newcommand{\bzeta}{\mbox{\boldmath $\zeta$}}

\newcommand{\bSigma}{\bm{\Sigma}}

\newcommand{\blambda}{\bm{\lambda}}

\newcommand{\bLambda}{\bm{\Lambda}}
\newcommand{\bLambdas}{\bLambda_{s}}

\newcommand{\bLambdan}{\bLambda_{n}}
\newcommand{\bLambdann}{\bLambda_{0n}}
\newcommand{\bLambdasnn}{\bLambda_{0sn}}
\newcommand{\tLambdas}{\widetilde{\bLambda}_{s}}
\newcommand{\tLambdasnn}{\widetilde{\bLambda}_{0sn}}
\newcommand{\bLambdasn}{\bLambda_{sn}}
\newcommand{\tLambdann}{\widetilde {\bLambda}_{0n}}
\newcommand{\bchisnn}{\mbox{\boldmath \large$\chi$}_{0sn}}
\newcommand{\chisrnn}{{\chi}_{0s,rn}}
\newcommand{\bFsnn}{\bF_{0sn}}
\newcommand{\bFsn}{\bF_{sn}}
\newcommand{\bDeltann}{\bDelta_{0n}}
\newcommand{\bDeletasnn}{\bDelta_{\eta,0sn}}
\newcommand{\bSigetasnn}{\bSigma_{\eta,0sn}}
\newcommand{\bSigetasn}{\bSigma_{\eta,sn}}
\newcommand{\bSigmas}{\bSigma_{s}}
\newcommand{\bSigmasnn}{\bSigma_{0sn}}
\newcommand{\bSigmasn}{\bSigma_{sn}}
\newcommand{\bSigmann}{\bSigma_{0n}}
\newcommand{\bSigman}{\bSigma_{n}}

\newcommand{\Dn}{\mathcal{D}_{n}}
\newcommand{\Pnn}{\mathbb{P}_{0n}}
\newcommand{\Pn}{\mathbb{P}_{n}}
\newcommand{\PPn}{\mathcal{P}_{n}}
\newcommand{\Pnon}{\mathcal{P}_{1,n}}
\newcommand{\Pntw}{\mathcal{P}_{2,n}}

\newcommand{\sigman}{ \sigma_{\delta} }
\newcommand{\btdeltan}{\wt{\boldsymbol{\delta}}_{n} }
\newcommand{\btdeltann}{\wt{\boldsymbol{\delta}}_{0n} }
\newcommand{\bDeltan}{\mbox{\boldmath $\Delta$}_{n}}

\newcommand{\bXisn}{\mbox{\boldmath $\Xi$}_{sn}}

\newcommand{\bpsi}{\mbox{\boldmath $\psi$}}
\newcommand{\bGamma}{\mbox{\boldmath $\Gamma$}}

\newcommand{\half}{\frac{1}{2}}

\newcommand{\de}{\mathrm{d}}

\newcommand{\sn}{s_{n}}
\newcommand{\qn}{q_{n}}
\newcommand{\qnn}{q_{0n}}
\newcommand{\qsn}{q_{sn}}
\newcommand{\qsnn}{q_{0sn}}
\newcommand{\tilqs}{\wt{q}_{s}}
\newcommand{\tilqsnn}{\wt{q}_{0,sn}}
\newcommand{\varphisn}{\varphi_{sn}}
\newcommand{\varphin}{\varphi_{n}}

\newcommand{\dn}{d_{n}}
\newcommand{\cn}{c_{n}}
\newcommand{\en}{e_{n}}
\newcommand{\ten}{{\wt{e}}^{(s)}_{n}}
\newcommand{\an}{a_{n}}

\newcommand{\tn}{t_{n}}
\newcommand{\taun}{\tau_{n}}
\newcommand{\varepsilonn}{\varepsilon_{n}}
\newcommand{\epsilonn}{\epsilon_{n}}
\newcommand{\pin}{\Pi_{n}}
\newcommand{\ns}{n_{s}}
\newcommand{\qs}{q_{s}}

\newcommand{\kl}[2] {\mathbb{KL} \left(#1\parallel #2\right)  }
\newcommand{\lnt}[2] {\ell_{0} \left[#1, #2\right]  }
\newcommand{\klv}[2] {\mathbb{V} \left(#1 \parallel #2\right)  }

\newcommand{\bThetan}{\bTheta_{n}}
\newcommand{\bThetann}{\bTheta_{0n}}

\newcommand{\delmin}{{\delta}_{\min}}

\newcommand{\compstore}{\textbf{Compute} \& \textbf{store}}

\newcommand{\baA}{b_{\bA}}
\newcommand{\bAs}{\bA_{s}}
\newcommand{\bPhis}{\bPhi_{s}}

\newcommand{\bAsn}{\bA_{sn}}
\newcommand{\bAsnn}{\bA_{0sn}}
\newcommand{\tAsnn}{\wt{\bA}_{0sn}}
\newcommand{\bGs}{\bG_{s}}
\newcommand{\bCs}{\bC_{s}}

\newcommand{\fnorm}[1] {  \norm{#1}_F}

\newcommand{\specnorm}[1] {  \norm{#1}_2}

\newcommand{\smin}[1] { s_{\min}\left(#1\right) }

\renewcommand\footnoterule{\kern-3pt \hrule \textwidth 2in \kern 2.6pt}

\usepackage[normalem]{ulem}

\usepackage{tikz}
\newcommand{\changen}[1]{#1}

\def\colred#1{\textcolor{red}{ #1}}
\def\colredb#1{\textcolor{red}{\bf #1}}

\def\boxit#1{\vbox{\hrule\hbox{\vrule\kern6pt \vbox{\kern6pt \textcolor{blue}{#1}\kern6pt}\kern6pt\vrule}\hrule}}

\def\authorfootnote#1{{\let\thefootnote\relax\footnotetext{#1}}}

\newcommand\undermat[2]{%
	\makebox[0pt][l]{$\smash{\underbrace{\phantom{%
					\begin{matrix}#2\end{matrix}}}_{\text{$#1$}}}$}#2}

\newcommand{\cut}{\colred{cut}}

\usepackage{titlesec}
\titleformat*{\section}{\normalsize\bfseries}
\titleformat*{\subsection}{\normalsize\bfseries}
\titleformat*{\subsubsection}{\normalsize\bfseries}
\titlespacing*{\section} {0pt}{2ex}{2ex}
\allowdisplaybreaks

\begin{document}
	\thispagestyle{empty}
	\baselineskip=28pt
	
	\newcommand{\papertitle}{{\Large \bf Inferring Covariance Structure from
			Multiple Data Sources via  
			Subspace Factor Analysis }}
	\newcommand{\authors}{
		Noirrit Kiran Chandra$^{\dagger}$,  \quad 
		David B. Dunson$^\ddagger$, \quad  
		Jason Xu$^{\ddagger}$  \\ 
		\vskip 7mm
		$^{\dagger}$Department of Mathematical Sciences, \\
		The University of Texas at Dallas, 
		Richardson, TX \\ Email: noirrit.chandra@utdallas.edu
		\vskip 9pt
		$^{\ddagger}$Department of Statistical Science, Duke University, Durham, NC
	}

\begin{center}
\papertitle
\vskip10mm
\baselineskip=12pt

\authors
\vskip 20pt 
\end{center}



\vskip 20pt 
\begin{center}
{\Large{\bf Abstract}} 
\end{center}
\baselineskip=12pt

Factor analysis provides a canonical framework for imposing lower-dimensional structure such as sparse covariance in high-dimensional data. High-dimensional data on the same set of variables are often collected under different conditions, for instance in reproducing studies across research groups. In such cases, it is natural to seek to learn the shared versus condition-specific structure. 
Existing hierarchical extensions of factor analysis have been proposed, but face practical issues including identifiability problems. To address these shortcomings, we propose a class of SUbspace Factor Analysis (SUFA) models, which characterize variation across groups at the level of a lower-dimensional subspace. We prove that the proposed class of SUFA models lead to identifiability of the shared versus group-specific components of the covariance, and study their posterior contraction properties. Taking a Bayesian approach, these contributions are developed alongside efficient posterior computation algorithms. 
Our sampler fully integrates out latent variables, is easily parallelizable and has complexity that does not depend on sample size.
We illustrate the methods through application to integration of multiple gene expression datasets relevant to immunology.


\baselineskip=12pt




\pagenumbering{arabic}
\setcounter{page}{0}
\newlength{\gnat}
\setlength{\gnat}{25.5pt} 
\baselineskip=\gnat

\section{Introduction}
\label{sec:intro}
With increasing calls for reproducibility and transparency in science, it has become standard  to make datasets widely available to the scientific community.  
This motivates interest in aggregating different but related datasets together. 
A prominent example arises in gene network analyses where it is of interest to merge datasets from studies that consider a common set of genes. 
Combining datasets increases sample size in studying covariance structure among the genes. To conduct principled inference, it is important to account for differences across studies, and not simply pool the data. 
This article focuses on developing statistical methods for inferring common versus study-specific covariance structures, with a particular motivation to multi-study applications in the high-dimensional setting.
These methods immediately apply to data from a single study featuring hierarchical, multi-group structure as well.

Consider the analysis of  gene expression levels in immune cells.
Integrating data from multiple studies serves to (1) increase statistical precision in making inferences on covariance structure between genes; 
(2) yield results that are more robust to study-to-study variability and hence more generalizable; and 
(3) obtain insight into shared versus study-specific contributors to the covariance.
We focus on datasets from the Immunological Genome Project  
\citep{immgen_paper},   
comprising microarray assays as well as bulk RNA sequencing data as
shown in Figure \ref{fig:heatmap_originaldata} in the supplementary materials.
Substantial similarity between the datasets can be observed albeit with significant amount of heterogeneity.
Such heterogeneity is typical and arises due to  differences between the subject populations,
and variation in data collection technologies inducing platform-specific effects in the respective datasets.

Given our interest in the covariance structure of high-dimensional data, it is natural to consider Bayesian sparse factor analysis (FA) models, which achieve state-of-the-art performance in the single-study case \citep{FAN_factor,knowles2011}. 
Even when different studies collect the same variables for closely-related populations, applying separate factor analyses to each dataset can lead to very different inferred structures. 
There are recent approaches extending factor analysis to the multi-study setting---notably, 
\citet{devito2019multi,devito2018bayesian} proposed a conceptually simple multi-study FA model.
This appealing model includes shared and study-specific components via an additive expansion but is known to  face identifiability issues, discussed further below.  
 \cite{roy_pfa2020} proposed a multiplicative `perturbed' FA model that focuses on inferring the shared structure while making use of subject-specific perturbations to resolve identifiability issues in post-processing steps.

Motivated by identifiability issues arising in multi-study factor analysis in the high-dimensional setting, we propose  a novel shared \textit{SUbspace Factor Analysis (SUFA)} model.
In particular, SUFA assumes that there is a common lower-dimensional subspace shared across the {multiple} studies under consideration,
providing a mechanism for borrowing of information. 
Exploiting this subspace structure, we propose \textit{almost sure} solutions to identify the study-specific variations from the shared covariance structure across all studies.\label{pg:intro_mod}
In related work, \citep{subspace_hoff} focus on identifying the \textit{best} shared subspace. 
In contrast, we focus on learning  shared versus study-specific contributions to the covariance structure.  
Our factor analytic approach allows one to infer latent factors jointly with the their loadings, providing valuable interpretations in  applications including gene expression studies \citep{factor_analysis_gene}.

We focus on identifiability up to the usual rotational ambiguity encountered in factor modeling. Following standard practice in the Bayesian factor modeling literature, we rotationally align the factors across MCMC iterations in a post-processing stage.
By using a fast algorithm recently introduced by  \citet{poworoznek2021}, we maintain sparsity in the post-processed samples of the loadings by leveraging a key sparsity-inducing property of the varimax rotation 
\citep{kaiser_varimax,rohe2020vintage}, which may improve interpretability. For example, a subset of genes that loads onto a particular factor is commonly interpreted as belonging to the same pathway. 

Moreover, we derive favorable posterior contraction rates for subspace factor models in a high-dimensional setting.
Our analysis shows that the shared study-specific covariance structures can be separately recovered with added precision for each additional study,
even when the marginal distributions of data differ substantially across studies.
While the classical factor analysis literature relies heavily on choosing the number of latent factors correctly,
we show that our covariance matrix estimation is robust  to this choice.
These contributions are developed  alongside computational innovations to enable efficient posterior inference. Latent variable-based Gibbs sampling \citep{bhattacharya2011sparse,sabnis2016divide} or expectation-maximization (EM) algorithms \citep{rovckova2016fast} set the standards for Bayesian factor models.
However, the conditional updating can become a computational bottleneck for massive sample sizes, further exacerbated in multi-study settings. 
\citet{dutta2020fa} proposed a fast matrix-free approach for exploratory factor analysis which does not carry over to Bayesian contexts. 
We develop a scalable  Hamiltonian Monte Carlo \citep[HMC,][]{neal_hmc} sampler that jointly updates parameters while marginalizing out latent factors. 
This leads to significant improvements in mixing relative to Gibbs sampling. 
The proposed sampler depends only on the study-specific sample covariance matrices, which can be computed and cached prior to running MCMC.
The computational complexity is  essentially invariant of the sample size, and amenable to a distributed computing framework parallelizing each MCMC step. 

The remainder of the article is organized as follows: 
Section \ref{sec:msfa} describes the SUFA model, and 
\ref{subsec:identifiability_issues} discusses and resolves the potential identifiability issues common to multi-study factor models.
Section \ref{subsec:specify_q} discusses our choice of prior and other important specifications. 
Section \ref{sec:theory} establishes posterior contraction rates in high-dimensional settings. 
These \changen{methodological developments} are made practical in
Section \ref{subsec:posterior_computation}, where we propose a scalable HMC algorithm for posterior sampling.
Section \ref{sec:simstudies} compares SUFA with existing approaches via a suite of simulation studies, and the methods are applied in an integrative gene network analysis case study in Section \ref{sec:application}.
We close by discussing our contributions and future directions in Section \ref{sec: discussion}. 

\section{Bayesian SUbspace Factor Analysis (SUFA)}
\label{sec:msfa}

In this article, we consider the setting where data consist of $S$ studies each comprising $d$-variate observations $\bY_{s,i}=(Y_{s,i,1},\dots, Y_{s,i,d})\trans$, $i=1,\dots,\ns$, $s=1,\dots,S$ measured on the same set of features with
$\bY_{s,i}\simiid \mn(\bmu_{s},\bSigma_{s})$.
Without loss of generality, we assume $\bmu_{s}=\bzero$ following standard practice to center the data. 
In our motivating application, we seek to jointly learn from  two microarray datasets and a bulkRNASeq dataset: we have $S=3$ studies each analyzing $d=474$ genes. In studying the correlation structure between genes,
we expect a common fundamental association between features across the studies, along with study-specific variations.
Hence, we let $\bSigma_{s}= \bSigma+ \bGamma_{s}$ where 
$\bSigma$ is a positive definite matrix quantifying the shared structure, and $\bGamma_{s}$ accounts for the respective study-specific dependencies. 

Bayesian factor analysis often yields state-of-the-art performance in single-study high-dimensional covariance estimation \citep{bhattacharya2011sparse, pati2014, rovckova2016fast}. Adopting a usual factor-analytic covariance factorization as the sum of a low rank and diagonal matrix, 
we let 
$\bSigma=\bLambda\bLambda\trans+\bDelta$,
where $\bLambda$ is a $d\times q$ factor loading matrix, with $q\ll d$ the number of latent factors,
and $\bDelta=\diag(\delta_{1}^{2},\dots, \delta_{d}^{2})$ is a diagonal matrix of idiosyncratic variances.
Evidence in the literature suggests that  expression-level dependent error variances tend to produce better results \citep{kepler2002normalization}, so we do not assume $\bDelta = \sigma^2 {\bf I}_d$.

To model the study-specific deviations, 
we let 
$\bGamma_s = \bLambda\bAs\bAs\trans\bLambda\trans$, where $\bAs$ is a $q\times \qs$  matrix.
This leads to the following expression for the covariance specific to study $s$:
$\bSigma_{s}=\bLambda\bLambda\trans+ \bLambda\bAs\bAs\trans\bLambda\trans+\bDelta.$
This can equivalently be written 
\vskip-8ex
\begin{equation}
\bY_{s,i}= \bLambda \boldeta_{s,i}+ \bLambda \bA_{s} \bzeta_{s,i}+ \bepsilon_{s,i},
~ \boldeta_{s,i}\simiid \mn_{q}(\bzero,\bI_{q}), 
~ \bzeta_{s,i}\simiid \mn_{\qs}(\bzero,\bI_{\qs}),
~ \bepsilon_{s,i} \simiid \mn_{d}(\bzero,\bDelta),	\label{eq:ourmodel}
\end{equation} 
\vskip-3ex
\noindent where $\boldeta_{s,i}$ is a $q$ dimensional latent factor in the shared subspace, 
$\bzeta_{s,i}$ are study-specific latent factors of dimension $q_s$,
and $\bepsilon_{s,i}$ are mean zero Gaussian error terms. 
Because $\boldeta_{s,i}$ are supported on the same subspace, our hierarchical model allows borrowing of information across studies in
estimating $\bSigma=\bLambda \bLambda\trans +\bDelta$. This is formalized in Section \ref{sec:theory}. 
Model \eqref{eq:ourmodel} implies the following marginal distributions 
\vskip-4ex
\begin{equation}
\bY_{s,i}\simiid \mn_{d}(\bzero, \bLambda\bLambda\trans+ \bLambda\bA_{s}\bA_{s}\trans\bLambda\trans+\bDelta)\text{ for }s=1,\dots,S.\label{eq:ourmodel_marginal}
\end{equation}
\vskip-2ex
Typically we take $\qs <q$, as
integrative analyses of multiple studies 
are meaningful when the signals are mostly shared across studies.
Imposing that $\qs <q$ allows $\bLambda$ to be the dominant term explaining the variation across the studies. 
The following sections make this intuition rigorous, and show how it ensures identifiability.

\vspace*{-3ex}
\paragraph*{Related works:}
\label{subsec:comparison_with_lit}
Under the same data organization,   \citet{devito2019multi, devito2018bayesian} define the multi-study factor analysis (MSFA) model:
\vskip-8ex
\begin{equation}
\bY_{s,i}= \bLambda \boldeta_{s,i}+ \bPhi_{s}\bzeta_{s,i}+ \bepsilon_{s,i},
~~ \boldeta_{s,i}\simiid \mn_{q}(\bzero,\bI_{q}), 
~~ \bzeta_{s,i}\simiid \mn_{\qs}(\bzero,\bI_{\qs}),
~~ \bepsilon_{s,i} \simiid \mn_{d}(\bzero,\bDelta_{s}),	\label{eq:MSFA_model}
\end{equation} 
\vskip-3ex
\noindent \label{pg:MSFA_diffs} where $\bLambda$ is the $d\times q$ shared factor loading matrix and $\bPhi_{s}$'s are the $d\times \qs$ study-specific loading matrices, with the remaining notation matching ours above. 
Their formulation is clearly related to ours---in fact, they are strictly more general in that \eqref{eq:ourmodel} can be viewed as the case where the study-specific loadings take the form $\bPhis = \bLambda \bAs$.
Indeed, our work is motivated by observing that \eqref{eq:MSFA_model} is \textit{too flexible}; in particular, allowing  $\bPhi_{s}$ to be arbitrary produces critical identifiability issues down the road.
To see this, consider its  implied marginal distribution $\bY_{s,i}\simiid \mn_{d}(\bzero, \bLambda\bLambda\trans + \bPhis\bPhis\trans +\bDelta_{s})$ \changen{which is equivalent to $\bY_{s,i}\simiid \mn_{d}(\bzero, \bzero\bzero\trans+ \wt{\bPhi}_{s}\wt{\bPhi}_{s}\trans +\bDelta_{s})$, where $\wt{\bPhi}_{s}=\begin{bmatrix}
		\bLambda & \bPhis
	\end{bmatrix}$.
	That is, the posterior distribution of the overly flexible Bayesian MSFA with multiplicative gamma process priors on $\bLambda$ and $\bPhi_{s}$'s can concentrate on the latter model with positive probability raising 
	critical identifiability issues between  $\bLambda$ and $\bPhis$'s.
}\label{pg:msfa_identifiable}  
In contrast, \eqref{eq:ourmodel_marginal} is more restrictive in not allowing $\bDelta$ to vary per study, and with $\bLambda$ parametrizing both the shared and study-specific terms. 
Our  contributions in Section \ref{subsec:identifiability_issues} make precise how our proposed, more parsimonious approach avoids these pitfalls. 
The combinatorial MSFA model  \citep{grabski2020bayesian} introduces notions of \textit{shared}, \textit{partially shared} and \textit{study-specific} factors by identifying latent factors. 

Perturbed factor analysis \citep[PFA,][] {roy_pfa2020} takes an altogether different approach with the same aim of learning shared covariance structure: 
\vskip-6ex
\begin{equation}
\bQ_{s}\bY_{s,i}= \bLambda \boldeta_{s,i}+ \bepsilon_{s,i},\qquad \boldeta_{s,i}\simiid \mn(\bzero,\bI_{q}), \qquad \bepsilon_{s,i} \simiid \mn(\bzero,\bDelta),	\label{eq:PFA_model}
\end{equation}
\vskip-2ex
\noindent  where $\bLambda$ is the $d\times q$ common factor loading matrix, $\bQ_{s}$ is a $d \times d$ perturbation matrix, and the remaining notation as above. 
Though it overcomes some of the pitfalls of MSFA, the study-specific effects cannot be recovered under PFA, and the introduction of $\bQ_{s}$ makes it difficult to scale beyond a few hundred dimensions.
\citet{alejandra2022} proposed the Bayesian factor regression (BFR) model with additional hierarchies to handle different types of batch effects but do not consider separate covariance structures across studies.

\subsection{Model Identifiablity Guarantees}
\label{subsec:identifiability_issues}
Factor analytic models such as SUFA are prone to two key identifiability issues: 
\begin{enumerate}[label=\textbf{(\Roman*)},leftmargin=14pt,topsep=0pt, partopsep=0pt,itemsep=2pt]
\item\label{identifiability_info} \textbf{Information Switching:} 
If there exists $q\times q$ non-null symmetric matrices $\bC_{1}$ and $\bC_{2}$ 
such that $\bI_{q}-\bC_{1} \succ \bzero$, 
$\bDelta- \bLambda \bC_{2} \bLambda\trans$ is a diagonal matrix with positive entries and
$\bAs \bAs\trans+\bC_{1}+\bC_{2}=\wt{\bA}_{s} \wt{\bA}_{s}\trans$ for $q\times \qs$ matrices $\wt{\bA}_{s}$'s for all $s=1,\dots,S$, then 
substituting $\wt{\bLambda}=\bLambda(\bI_{q}-\bC_{1})^{\half}$, $\wt{\bDelta}=\bDelta- \bLambda \bC_{2} \bLambda\trans$ and $\wt{\bA}_{s}$
yield the same marginal distribution in \eqref{eq:ourmodel_marginal},
leading to an identifiability issue.

\item \label{identifiability_rotation} \textbf{Rotational Ambiguity: }
Let $\wt{\bLambda}=\bLambda \bH$ and $\wt{\bA}_{s}=\bH\trans \bA_{s} \bH_{s}$ where $\bH$ and $\bH_{s}$'s are orthogonal matrices of order $q \times q$ and $\qs \times \qs$ respectively.
Then substituting $\bLambda$ and $\bA_{s}$'s by $\wt{\bLambda}$ and $\wt{\bA}_{s}$'s respectively in \eqref{eq:ourmodel_marginal} yields identical marginal distributions.

\end{enumerate}
This \textit{information switching} {allows $\bLambda\bAs$'s to have shared column(s) across all studies posing} a critical issue in interpreting $\bSigma=\bLambda\bLambda\trans+\bDelta$ as the shared covariance term, since $\bSigma$ will not be identifiable.
To establish some intuition, we consider an example where $d=5$, $S=2$, $q=3$ with $q_{1}=q_{2}=2$, 
\vskip-.5ex
{\scriptsize
\begin{equation}
\bLambda=\begin{bmatrix}
	7  &  5 &   6\\
	6  &  6 &   7\\
	6  &  9 &   4\\
	5  & 5 &   6\\
	4 &   6 &   6
\end{bmatrix}, \quad
\bA_{1}=\begin{bmatrix}
	3 & \colred{0}\\
	0 & \colred{2}\\
	0 & \colred{0}
\end{bmatrix} \quad \text{ and }
\bA_{2}=\begin{bmatrix}
	\colred{0} & 0\\
	\colred{2} & 0\\
	\colred{0} & 4
\end{bmatrix}.\label{eq:example}
\end{equation}}
\vskip-1ex
\noindent 
Equation \eqref{eq:ourmodel} implies that for each of the two ``studies", $\bY_{s,i}$ is given by
\vskip.75ex
{\scriptsize
\begin{equation}
\bY_{1,i}= \left [
\begin{array}{rrr|rrr}
	7  &  5 &   6 & 21 & \colred{10}\\
	6  &  6 &   7 & 18 & \colred{12}\\
	6  &  9 &   4 & 18 & \colred{18}\\
	5  & 5 &   6 & 15 & \colred{10}\\
	\undermat{\bLambda}{4 &   6 &   6}& \undermat{\bLambda\bA_{1}}{12 & \colred{12}}\\
\end{array}
\right ] \begin{bmatrix} \boldeta_{1,i}\\ \bzeta_{1,i} \end{bmatrix} +\bepsilon_{1,i}
= \left [
\begin{array}{rrrr|rr}
	7  &  5 &   6 & \colred{10} & 21 \\
	6  &  6 &   7& \colred{12} & 18 \\
	6  &  9 &   4& \colred{18} & 18 \\
	5  & 5 &   6 & \colred{10} & 15 \\
	\undermat{\wt{\bLambda}}{4 &   6 &   6& \colred{12}}& \undermat{\wt{\bLambda}\wt{\bA}_{1}}{12 }\\
\end{array}
\right ] \begin{bmatrix} \wt{\boldeta}_{1,i}\\ \wt{\bzeta}_{1,i} \end{bmatrix} +\bepsilon_{1,i},\label{eq:example1}
\end{equation}
\vskip3.6ex
\begin{equation}
\bY_{2,i}= \left [
\begin{array}{rrr|rrr}
	7  &  5 &   6 &  \colred{10} & 24\\
	6  &  6 &   7 &  \colred{12} & 28\\
	6  &  9 &   4 &  \colred{18} & 16\\
	5  & 5 &   6 &  \colred{10} & 24\\
	\undermat{\bLambda}{4 &   6 &   6}& \undermat{\bLambda\bA_{2}}{\colred{12} & 24}\\
\end{array}
\right ] \begin{bmatrix} \boldeta_{2,i}\\ \bzeta_{2,i} \end{bmatrix} +\bepsilon_{2,i}
= \left [
\begin{array}{rrrr|rr}
	7  &  5 &   6 & \colred{10} & 24 \\
	6  &  6 &   7& \colred{12} & 28 \\
	6  &  9 &   4& \colred{18} & 16 \\
	5  & 5 &   6 & \colred{10} & 24 \\
	\undermat{\wt{\bLambda}}{4 &   6 &   6& \colred{12}}& \undermat{\wt{\bLambda}\wt{\bA}_{2}}{24 }\\
\end{array}
\right ] \begin{bmatrix} \wt{\boldeta}_{2,i}\\ \wt{\bzeta}_{2,i} \end{bmatrix} +\bepsilon_{2,i}.\label{eq:example2}
\end{equation}}
\vskip2.5ex

That the above equations can each be written in two equivalent decompositions illustrates the \textit{information switching} problem: we cannot distinguish whether the column $\bLambda \times {\scriptsize \begin{bmatrix}
\colred{0} & \colred{2} & \colred{0}
\end{bmatrix}\trans}= {\scriptsize\begin{bmatrix}
\colred{10} & \colred{12} & \colred{18} & \colred{10} & \colred{12}
\end{bmatrix}\trans}$ is a part of the shared effect or the study-specific effect.
Although the $\bAs$ matrices should take into account the study-specific variations only, the individual effect of one study in this example can be at least ``partially explained" by the others.
From the perspective of the statistical model, $\bAs$'s  are no longer study-specific, but absorb a part of the shared variation which should be captured entirely by $\bSigma$.
In Section \ref{subsec:soln_info_switch} we identify the necessary and sufficient condition causing this issue and propose an \textit{almost sure} solution.

Rotational ambiguity has been well documented in the prior literature, affecting both MSFA and PFA. 
In addition, MSFA can suffer from information switching---briefly, for a $d \times d$ order symmetric matrix $\wt{\bC}$ such that $\bLambda\bLambda\trans-\wt{\bC}\succ \bzero$, $\bPhis \bPhis\trans+\wt{\bC}=\wt{\bPhi}_{s} \wt{\bPhi}_{s}\trans$ for some $d\times \qs$ order matrix $\wt{\bPhi}_{s}$. \label{pg:msga_ident2}
\citet{grabski2020bayesian}  use a posterior credible ball-based post-processing algorithm to identify the shared and study-specific factors.
To resolve \ref{identifiability_info} in a model-based manner, \citet{devito2019multi} restrict $\bLambda$ and $\bPhi_{s}$'s to be lower-triangular matrices while enforcing the augmented matrix $\begin{bmatrix}
\bLambda & \bPhi_{1}& \cdots & \bPhi_{S}
\end{bmatrix} $ to be of full column-rank while obtaining the MLE.
\label{pg:fullrank}
Lower-triangular structures incur the tradeoff of introducing an order dependence in the variables \citep{fruhwirth2018sparse}; when no natural ordering exists, results can be very sensitive to permutations of the variables \citep{Carvalho2008orderdependence}.
These structural constraints also can affect mixing and yield inconsistent estimates \citep{millsap2001,erosheva2017dealing}.
In high-dimensional applications often sparsity is assumed on the loading matrices \citep{daniele2019penalizingfactor}; with exact zero entries in $\bLambda$ and $\bPhi_{s}$'s, ensuring full column-rank of $\begin{bmatrix}
\bLambda & \bPhi_{1}& \cdots & \bPhi_{S}
\end{bmatrix}$ is non-trivial.
The Bayesian MSFA by \citet{devito2018bayesian} assumes multiplicative gamma process priors on the loading matrices without any restrictions on the column dimensions and thus can be susceptible to both \ref{identifiability_info} and \ref{identifiability_rotation}.
On the other hand, \citet{roy_pfa2020} set $\bQ_{1}=\bI_{d}$ to impose identifiability in PFA, which can again be sensitive to the choice of the ``first" study.
In the following section, we provide \textit{almost sure} solutions to the identifiability issues 
under much milder conditions, avoiding any such structural assumptions.

\subsubsection{Resolving Information Switching}
\label{subsec:soln_info_switch}
We now derive the necessary and sufficient conditions 
for information switching.  

\begin{lemma}
\label{lemma:colspace}
Assume that the data admits the marginal distribution in \eqref{eq:ourmodel_marginal} for all the studies. 
Then information switching occurs \emph{if and only if} there exists $\bLambda$, $\bDelta$ and $\bAs$'s such that $\bY_{s,i}\simiid \mn_{d}(\bzero, \bLambda\bLambda\trans+ \bLambda\bAs\bAs\trans\bLambda\trans+\bDelta)$ and $\bigcap_{s=1}^S \cspace{\bAs}$ is non-null or $\mathrm{rank}(\bAs)<\qs$ for all $s$, where $\cspace{\bA}$ denotes the column space of a matrix $\bA$.	
\end{lemma}
\vspace*{-2ex}

Intuitively, a non-null $\bigcap_{s=1}^{S} \cspace{\bAs}$ implies that the $\bAs$ matrices are no longer study-specific. 
As we saw in example \eqref{eq:example}, $\bAs$ can absorb part of the shared variation which is meant to be captured by $\bSigma$.
Rank deficiencies in the columns leave space for the $\bAs$'s to partially absorb the shared information.
The Lemma provides requisite insights to avoid this with a very simple condition to ensure $\bSigma$ fully explains the shared covariance. 
\begin{theorem}
\label{th:no_prior_support}
If \changen{$q\geq\sum_{s=1}^{S}\qs >0$} then {information switching} has zero support under any non-degenerate continuous prior on $\bAs$.
\end{theorem}
\vspace*{-2ex}
Intuitively, aggregating multiple studies is a fruitful pursuit only when the datasets share enough similarity or structure; it is natural to expect that our subspace factor model is well-posed when there are fewer study-specific latent factors than shared factors. 
\changen{As $S$ grows, the condition in Theorem \ref{th:no_prior_support} entails $q$ to tend towards larger values compared to study-specific values $\qs$, while also resolving the information switching issue \textit{almost surely}.
	That is, \textit{generic identifiability} of $\bSigma$ \citep{allman2009identifiability} is ensured under any continuous prior on $\bAs$, as long as $\sum_{s=1}^{S} \qs\leq q$.}
The assumption is therefore a natural one, and $\bSigma$ is completely identifiable as a result, allowing for us to study and interpret the interaction between variables. 
To illustrate concretely in the context of the above example, consider setting the shared $\wt{\bLambda}$ to be a $5 \times 4$ matrix and $q_{1}=q_{2}=1$, i.e., right-hand side of \eqref{eq:example1}-\eqref{eq:example2}.
Now, $\wt{\bLambda}$ accounts for the variability in the marginal distribution \eqref{eq:ourmodel_marginal}, and in turn
the parts explained by $\bA_{1}$ and $\bA_{2}$ no longer have intersection. 
\changen{When $S$ is very large, it may not be possible to satisfy $q\geq\sum_{s=1}^{S}\qs$. 
	To see this, consider an extreme case where $S=d+1$ with $\qs=1$ for all $s$, which implies $\sum_{s=1}^{S} \qs=d+1>q$. 
	Even in such cases when identifiability of shared versus study-specific components is compromised, SUFA can nonetheless fit the data well in terms of efficiently learning the marginal covariance structures by borrowing information across studies.	
	Further detailed discussions and simulations can be found in Section \ref{sm subsec:many_studies} of the supplementary materials.
	Notably, requiring that the augmented matrix
 $\begin{bmatrix}
		\bLambda & \bPhi_{1}& \cdots & \bPhi_{S}
	\end{bmatrix} $ is of full column-rank  to ensure identifiability in \citet{devito2019multi} imposes a similar condition $dq-q(q-1)/2 +\sum_{s=1}^{S}\{d\qs-\qs(\qs-1)/2\}\leq Sd(d-1)/2$, which similarly becomes unattainable when $S$ grows large.}

\subsubsection{Resolving Rotational Ambiguity}
\label{subsec:rot_ambig}
Identifiability with respect to  orthogonal rotation \ref{identifiability_rotation} is not essential for estimating the shared covariance $\bSigma$. 
However, rotational ambiguity does create obstacles to inferring interpretable lower-dimensional factors \citep{Psychology_fa, factor_analysis_gene}.
Notably $\bLambda \bLambda\trans$ is identifiable from the diagonal $\bDelta$ if $q<(2d-\sqrt{8d+1})/2$ \citep{bekker1997variance_identification} and with this condition
it is thus common to address \ref{identifiability_rotation} via post-processing the $\bLambda$ samples \citep[and others]{roy_pfa2020,devito2018bayesian,papastamoulis2020}.
Such approaches avoid order dependence arising from structural assumptions.

Interpretability of the factor loadings is greatly enhanced by sparsity---for example, the sparse subset of genes with nonzero loadings onto a common factor can be associated with a common biological pathway. 
Even when one places a shrinkage prior on the loadings to favor (near) sparsity, this desired structure is potentially destroyed after applying post-processing algorithms to the MCMC samples of the loadings. 
To avoid this, we adopt the approach of 
\citet{poworoznek2021}. 
The method first tackles the generic  rotational invariance  across the MCMC samples 
using an \textit{orthogonalization step} leveraging the varimax transformation \citep{kaiser_varimax} to resolve ambiguity up to switching of the column labels and signs.
These ambiguities are both resolved in the next step by \textit{matching} each MCMC sample to a reference matrix called a \textit{pivot}, aligning samples via a greedy maximization scheme.

The post-processed MCMC samples can thus be considered \textit{matched and aligned} with respect to a common orthogonal transformation for inference downstream.
Critically, the varimax rotation in the \textit{orthogonalization step} implicitly induces sparsity \citep{rohe2020vintage}. 
In more detail, the varimax criterion is given by
\vskip-8ex
\begin{equation*}
\textstyle{\bH_{\vmax}= {\arg\max}_{\bH} \left[\frac{1}{d} \sum_{h=1}^{q} \sum_{j=1}^{d} (\bLambda\bH)_{j,h}^{4} - \sum_{h=1}^{q}\left\{\frac{1}{d}\sum_{j=1}^{d}(\bLambda\bH)_{j,h}^{2}\right\}^{2}\right]},
\end{equation*}
\vskip-2ex
\noindent and describes the optimal rotation maximizing the sum of the variances of squared loadings. 
Intuitively, this is achieved if (i) any given variable has a high loading on a single factor but near-zero loadings on the others, and (ii) any given factor is constituted by a few variables with very high loadings but near-zero support from the others. 
As a result, the summary matrix obtained from the posterior MCMC samples of $\bLambda$ is sparse upon applying varimax rotation, while the marginal distribution of the data is unaffected.
Section \ref{sm sec:additional_sim} of the supplementary materials contains thorough simulations showing see that this strategy accurately recovers $\bLambda$ consistently across several realistic scenarios. 
As $\bLambda^{(s)}=\bLambda\bAs$ can be defined to denote the study-specific loading matrices, 
the above prescription applies similarly to obtaining study-specific loadings from samples of $\bLambda^{(s)}$.

\subsection{Latent Dimensions and Prior Specification}
\vspace*{-1.5ex}
\paragraph{Latent dimensions:}
\label{subsec:specify_q}
In most practical applications, we do not know the column dimensions of $\bLambda$ and $\bAs$, denoted by $q$ and $\qs$, respectively. 
Although one may choose priors on $q$ and $\qs$ and implement reversible-jump algorithms \citep{rjmcmc}, this will often lead to inefficient computation in large $d$ settings. Instead it has become common practice in the single-study literature to use overfitted factor models, first fixing $q$ at some upper bound and then leveraging appropriate priors to shrink the extra columns \citep[and related others]{legramanti2019bayesian, schiavon2021}.
This strategy substantially simplifies MCMC implementations. To choose upper bounds on the number of factors in practice, 
we employ augmented implicitly restarted Lanczos bidiagonalization  \citep{irlba} to obtain approximate singular values and eigenvectors of the pooled dataset, choosing the smallest $\wh{q}$ that explains at least $95\%$ of the variability in the data. 
To ensure identifiablity of $\bLambda \bLambda\trans$ from $\bDelta$,  we restrict $\wh{q}<(2d-\sqrt{8d+1})/2$ \citep{bekker1997variance_identification}.
After doing so, the simple choice $\wh{\qs}=\wh{q}/S$ satisfies the conditions in Theorem \ref{th:no_prior_support}.
We later show in Section \ref{sec:theory} that recovering the marginal covariance is asymptotically robust with respect to the choice of $q$ and $\qs$ under appropriate priors on $\bLambda$.

\vspace*{-2.5ex}
\paragraph{Prior specifications:}
\label{subsec:prior_specifications}
In high-dimensional applications, it is important to reduce the number of parameters in the factor loadings matrices. 
Among a wide variety of appropriate shrinkage priors, we use the Dirichlet-Laplace \citep[DL,][] {dir_laplace} for its near-minmax optimal contraction rates in the single-study context \citep{pati2014} and computational simplicity. 
We let $\vect(\bLambda)\sim \DL(a)$ where $a$ is a suitably chosen hyperparameter; details appear in Section \ref{sec: sm hyperparameters} of the supplementary materials. 
For the study-specific terms $\bAs$, 
we let its entries $a_{s,i,j}\simiid \mn(0, \baA)$, a choice that 
avoids information switching  \textit{almost surely}.  
We avoid shrinkage priors on $\bAs$'s as in practice it 
can potentially lead to concentration close to parameter values that violate the identifiability conditions of Lemma \ref{lemma:colspace}. 
\label{pg:sparsity_avoid}
Regarding the idiosyncratic variances, we consider 
$\log \delta_{j}^{2} \simiid \mn(\mu_{\delta},\sigma^{2}_{\delta})$
as log-normal distributions have modes bounded away from zero and tend to produce more numerically stable estimates compared to commonly used inverse-gamma and half-Cauchy priors \citep{gemlan_noninfo_prior}.

\section{Posterior Contraction Rates}
\label{sec:theory}
We now analyze the posterior contraction rates for recovering the shared and study-specific covariance matrices.
Before detailing the assumptions on the true data-generating mechanism and the postulated SUFA model, we fix notational conventions. 
We let $\pin(\cdot)$ denote the prior and $\pin(\cdot\mid \Dn )$ the corresponding posterior given data $\Dn=\{\bY_{1,1},\dots,\bY_{1,n_{1}},\dots, \bY_{S,1},\dots,\bY_{S,n_{S}}\}$. 
Throughout, $\fnorm{\bA}$  and $\specnorm{\bA}$ denote the Frobenius and spectral norms of a matrix $\bA$, respectively. 
For real sequences $\{\an\}$, $\{b_{n}\}$,
$\an=o(b_{n})$ implies that $\lim \abs{{\an}/{b_{n}}}=0$
and $\an\succ b_{n}$ implies that $\liminf {\an}/{b_{n}}>0$.
A $d$-dimensional vector $\btheta$ is said to be $s$-sparse if it has only $s$ nonzero elements, and
we denote the set of all $s$-sparse vectors in $\rR^{d}$ by $\lnt{s}{d}$. 

\vspace*{-3.5ex}
\paragraph*{Data-generating mechanism:} 
We assume that the data are generated according to $\bY_{s,1:\ns}\simiid \mn_{\dn}(\bzero, \bSigmasnn)$ for each study $s=1,\dots,S$. 
Here,   $\bSigmasnn=\bLambdann\bLambdann\trans+ \bLambdann\bAsnn\bAsnn\trans\bLambdann\trans+\bDeltann$ 
with
$\bLambdann$ a $\dn\times\qnn$ sparse matrix,
$\bAsnn$ a $\qnn\times\qsnn$ matrix with real entries,
and $\bDeltann:=\diag(\delta^{2}_{01},\dots,\delta^{2}_{0\dn})$. Compared to the setup in prior work by \cite{pati2014}, we consider multiple studies and allow heterogeneous idiosyncratic errors, which will provide more realism in modeling the data as well as require a more nuanced analysis. 
Denoting $n=\sum_{s=1}^{S}\ns$ the combined sample size across studies, 
we will allow the model parameters to increase in dimension with $n$ as suggested by the subscripts in our notations.  
Under this setup, we now state the required sparsity conditions on the true parameters in order to recover the shared 
as well as the study specific covariance structures in high-dimensional settings. 

\begin{enumerate}[label=(C\arabic*),topsep=0pt, partopsep=0pt,itemsep=-5pt]
\item \label{ass1} For $s=1,\dots,S$,  $\liminf_{n\to\infty}\frac{\ns}{n}>0$.
\item \label{ass2} 
Let $\{\qnn\}$, $\{\qsnn\}$ for $s=1,\dots,S$, $\{\sn\}$  and $\{\dn\}$ be increasing sequences of positive integers 
such that $\sum_{s=1}^{S}\qsnn\leq\qnn<\dn$.

\item \label{ass3} 
$\bLambdann$ is a $\dn\times\qnn$ full rank matrix such that each column of $\bLambdann$ belongs to $\lnt{\sn}{\dn}$, 
$\specnorm{\bLambdann}^{2}=o(\frac{\cn}{\qnn})$ and 
$\max_{1\leq s\leq S} \specnorm{\bAsnn}^{2}=o(\qnn)$ where $\{\cn\}$ is an increasing sequence of positive real numbers.

\item \label{ass4} $\max_{j}\delta^{2}_{0\dn}=o(\cn)$ and $\min_{j}\delta^{2}_{0\dn}>\delmin^{2}$ 
where $\delmin^{2}$ is a positive constant.

\item\label{ass5}  
$\frac{\dn\max\{ (\log\cn)^{2},\log\dn\} }{\cn^{2}\sn\qnn\log(\dn \qnn)}=o(1)$.
\end{enumerate}
In less technical terms, \ref{ass1} ensures that none of the studies has a negligible proportion of the data.
\ref{ass2} 	  and \ref{ass3}
specify the requisite sparsity conditions. 
Together, \ref{ass3} and \ref{ass4} imply that the marginal variances grow in $o(\cn)$ while ensuring that  
they are also bounded away from 0.
The technical condition \ref{ass5} specifies the relative rate of increment of the parameters.

\vspace*{-3.5ex}
\paragraph*{Specifics of the posited SUFA model:} 
In practice, the latent dimensions are unknown, and it is desirable to establish theory for the strategy of overspecifying a model and then shrinking redundant parameters via shrinkage priors.
This section identifies conditions under which we may derive guarantees in this overparametrized setting. To this end,  let $\bY_{s,1:\ns}\simiid \mn_{\dn}(\bzero, \bLambdan\bLambdan\trans+ \bLambdan\bAsn\bAsn\trans\bLambdan\trans+\bDeltan)$ where 
$\bLambdan$ and $\bAsn$ are $\dn\times \qn$ and $\qn\times \qsn$ matrices, respectively, 
and $\bDeltan=\diag(\delta^{2}_{1},\dots,\delta^{2}_{\dn})$.
The parameters in the posited (possibly misspecified) model do not have $0$ subscripts, to distinguish from their ground truth counterparts.

We let $\{\qn\}$, $\{\qsn\}$ be increasing sequences of positive integers such that $\qnn\leq\qn$, $\qsnn\leq\qsn$ so that the chosen number of shared and study-specific latent factors upper bound the respective true (also unknown) numbers resulting in a misspecified and over-parametrized factor model in the postulation.\label{pg:overspecification}
We additionally assume $\sum_{s=1}^{S}\qsn\leq \qn$  to ensure the identifiability condition from Theorem  \ref{th:no_prior_support}.
As discussed in Section \ref{subsec:prior_specifications} 
we let $\vect(\bLambdan)\sim \DL(\an)$ 
with $\an={1}/{\dn\qn}$ to impose sparsity on $\bLambdan$.
Below, we provide sufficient conditions on the overspecified dimensions to ensure consistency of the class of postulated models even when they are misspecified.
\begin{enumerate}[label=(D\arabic*),topsep=0pt, partopsep=0pt,itemsep=-3pt]
\item \label{condd2} $\qn^{2}=  o(\cn^{2}\sn\qnn)$ and 
$\lim_{n\to\infty} {\frac{\cn^{12}} {n} } { \left\{\sn\qnn\log  (\dn\qn) \right\}^{ 3 }}  = 0$.

\item \label{condd1}  $\log(\dn\qn)\min\left\{\frac{\cn}{\sn\qnn^{3}} (\dn\qn)^{\cn^{2}}, \sn(\dn\qn)^{\frac{\cn^{2}}{\qn^{2}}\sn\qnn } 
, \frac{\sn\qnn}{(\log\cn)^{2}} (\dn\qn)^{\frac{\cn^{2}}{\dn } \sn\qnn }
\right\} \succ n$.
\end{enumerate}  
The first condition  \ref{condd2} makes explicit how much overparametrization can be tolerated,
while the second condition specifies the signal-to-noise ratio with respect to the true column dimension $\qnn$, upper bound $\cn$ of the marginal variances and  number of non-zero elements $\sn$ in each column of $\bLambda$.
The second condition in \ref{condd2} and \ref{ass5} jointly imply that $\dn\log\dn$ grows in $o(n^{\alpha})$,  $0<\alpha<1$.	
\ref{condd1} can be interpreted as a lower bound on the signal strength in order to recover the covariance structures. 

\vspace*{-3.5ex}
\paragraph*{Contraction results:}
Let
$\bThetann=\{\bLambdann,\bDeltann, \bA_{01n},\dots,\bA_{0Sn}\}$ 
be the true data-generating parameters 
and $\Pnn$ denote the joint distribution of $\Dn$ under $\bThetann$.
We define $\bSigmann=\bLambdann\bLambdann\trans+\bDeltann$ as the true shared covariance structure and $\bSigman=\bLambdan\bLambdan\trans+\bDeltan$ to be the shared covariance in the postulated model.
The following result formally characterizes the posterior contraction properties around $\bSigmann$.

\begin{theorem}[Contraction rate]
\label{thm: posterior concentration}
For any $\bThetann$ satisfying \ref{ass1}-\ref{ass5} and 
priors satisfying \ref{condd2}-\ref{condd1}, we have
$\lim_{n\to \infty} \eE_{\Pnn}   \pin \left(\specnorm{\bSigman-\bSigmann}>M\varepsilonn\mid \Dn \right) = 0$,  
where $\varepsilonn =   \cn^{6} \sqrt {\frac{ \{ \sn\qnn\log(\dn\qn)\}^{3}} { n_{1}+\dots+n_{S}} } $ 
and $M>0$ is a large enough constant.	  	
\end{theorem}
Theorem \ref{thm: posterior concentration} has several important implications for modeling and computation. First, the combined sample size appears in the denominator of the contraction rate $\varepsilon_{n}$, which is consistent with the intuition that learning the shared covariance structure improves inference by borrowing strength across multiple studies or data views.
Because the true values of $\qnn$ and $\qsnn$ are unknown, 
we made the weak assumption that the specified column dimensions are larger than the ground truth dimensions.
Although the best contraction rate is unsurprisingly achieved when correctly setting $\qn=\qnn$,
Theorem \ref{thm: posterior concentration} validates the practice of beginning with overparametrized models, showing that it is still possible to recover $\bSigmann$.
This result provides an important theoretical foundation to the common heuristic strategy of setting the column dimension of $\bLambda$ to some upper bound in FA model implementations.

Note that $\bLambdann\bAsnn\bAsnn\trans\bLambdann\trans$s are the study-specific covariance structures.
The following result shows that these study-specific terms can also be recovered for all $s$.
\begin{corollary}[Recovering individual structures]
\label{cor: posterior concentration_indiv}
Under the same conditions as Theorem \ref{thm: posterior concentration},
$\lim_{n\to \infty}\eE_{\Pnn}  \pin \left(\specnorm{\bLambdan\bAsn\bAsn\trans\bLambdan\trans-\bLambdann\bAsnn\bAsnn\trans\bLambdann\trans}>M\varepsilon_{n}\mid \Dn \right) = 0$.
\end{corollary}

\section{Gradient-based Posterior Computation}
\label{subsec:posterior_computation}
Gibbs sampling is the most popular approach for posterior inference in Bayesian factor models.
While
such an approach is straightforward for our model \eqref{eq:ourmodel}, 
it is well-documented that alternately updating latent factors and covariance structure parameters can lead to slow mixing. Instead,  marginalizing out auxiliary parameters can often dramatically improve MCMC performance \citep{robert2021rao}, and avoids instantiating the latent factors
$\boldeta_{s,i}$ and $\bzeta_{s,i}$ which quickly increases in computational cost for large $n$.
In view of this, we develop a Hamiltonian Monte Carlo-within-Gibbs sampler based on the marginal posterior, after integrating out all latent factors.
The HMC sampler makes informed proposals via the gradient of the $\log$-target density \citep[HMC,][] {neal_hmc}.
We will see that our proposed algorithm confers substantial gains in real computing time compared to Gibbs sampling. 

Denote the log-likelihood of the marginal SUFA model \eqref{eq:ourmodel_marginal} by $\L := K -\half \sum_{s=1}^{S} \{ \ns \log \abs{\bSigma_{s}} +\tr(\bSigma_{s}^{-1}\bW_{s})\}$
with marginal covariance matrix $\bSigma_{s}=\bLambda(\bI_{q} +\bAs\bAs\trans) \bLambda\trans + \bDelta$, sample sum of squares matrix $\bW_{s}=\sum_{i=1}^{\ns}\bY_{s,i}\bY_{s,i}\trans$ for study $s$
and a constant $K$.
The DL prior on $\bLambda$ admits a Gaussian distribution conditionally on the hyperparameters $\tau$, $\bpsi$ and $\bphi$ (see \eqref{eq:dir_laplace_prior} in the supplementary materials for details). 
Thus from the prior specifications outlined in Section \ref{subsec:prior_specifications},
the posterior density of $\bTheta:=(\bLambda, \bDelta, \bA_{1},\dots, \bA_{S})$ given $\bpsi, \bphi, \tau$ and the observed data is given by
\vskip-5ex
\begin{equation}
\textstyle{\Pi(\bTheta \mid -)= \exp(\L) \times \Pi(\bLambda \mid \tau, \bphi, \bpsi) \times \Pi(\bDelta) \times \prod_{s=1}^{S} \Pi(\bAs).}
\label{eq:full_posterior}
\end{equation}
\vskip-2ex
\noindent 
To implement HMC, gradients of $\log \Pi(\bTheta \mid -)$ with respect to $\bTheta$ are required (see Section \ref{subsec:hmc_brief} in the supplementary materials for a general discussion on HMC algorithms). 
In Appendix \ref{sec:distributed_gradient} we show that the gradients can be analytically expressed and executable via distributed parallel operations.
Thus we avoid expensive automatic differentiation algorithms \citep{autodiff} used by probabilistic samplers like \texttt{Stan}.
Equipped with these gradients, we propose an HMC-within-Gibbs sampler where 
$\bTheta$ is updated given $(\tau,\bphi,\bpsi)$ using an HMC step, 
and then the hyperparameters $(\tau,\bphi,\bpsi)$ are updated conditionally on $\bLambda$ using a Gibbs step following \citet{dir_laplace}.
\begin{algorithm}[!b]
\footnotesize
\caption{HMC-within-Gibbs Sampler for SUFA}
\label{algo:hmc_msfa}
\SetAlgoLined
\DontPrintSemicolon
Initialize $\bTheta$, $\bpsi$, $\bphi$ and $\tau$.\;
\For{$i=1,\dots,N$}{
\begin{description}[leftmargin=0pt]
\item[HMC Update $\bTheta$:]  
Draw $\bp\sim\mn(\bzero, \bM)$ and initiate 
$\{\bTheta(0), \bp(0)\}=(\bTheta_{i-1},\bp)$.
Repeat the following $L$ leapfrog steps to numerically solve Hamiltonian equations
:\;
\For{$\ell=0,\dots,L-1$}{
\vskip-4ex
\begin{align*}
&\textstyle{\bTheta\{t+ (\ell+1)\delta t\}= 	\bTheta(t+ \ell\delta t) + \delta t \bM^{-1} \left[ \bp(t+ \ell\delta t)-\frac{\delta t}{2} \nabla V\left\{\bTheta(t+ \ell\delta t)\right\} \right]},\\
&\textstyle{\bp\{t+ (\ell+1)\delta t\}= 	\bp(t+ \ell\delta t) - \frac{\delta t }{2}  \left[  \nabla V\left\{\bTheta(t+ \ell\delta t)\right\} + \nabla V\left\{\bTheta(t+ (\ell+1)\delta t)\right\}\right]}.
\end{align*}
\vskip-2ex}
Letting $(\bTheta', \bp')=(\bTheta(t+ L\delta t),\bp(t+ L\delta t))$, 
set $\bTheta_{i}= \bTheta'$ with probability $\min \left[1, e^{-\left\{ H\left(\bTheta', \bp'\right)- H\left(\bTheta_{i-1},\bp\right) \right\}} \right]$,  where $H(\bTheta,\bp)=-\log \Pi(\bTheta \mid -) + \bp\trans \bM\inv \bp/2$.		

\item[Gibbs Update $\bpsi$, $\bphi$ and $\tau$:] \hspace{-.1in}
\begin{enumerate}[label=({\roman*})]
\item For $j=1, \dots, d$ and $h=1,\dots q$ sample $\wt{\psi}_{j,h}$ independently from inverse-Gaussian distributions $\mathrm{iG}\left(\tau_{i-1}{\phi_{i-1,j,h}}/{\abs{\lambda_{i,j,h}}},1\right)$ and set $\psi_{i,j,h}=1/\wt{\psi}_{j,h}$.			
\item Draw $T_{j,h}$ independently from from a generalized inverse Gaussian distributions $\mathrm{giG}(a - 1, 1, 2\abs{\lambda_{i,j,h}} )$ and set $\phi_{i,j,h}  = T_{j,h}/T$ with $T=\sum_{j,h} T_{j,h}$.
\item Sample $\tau_{i}\sim\mbox{giG}\left\{dq(1-a),1, 2\sum_{j,h} \frac{\abs{\lambda_{i,j,h}}}{\phi_{i,j,h} }\right\}$.
\end{enumerate}
\end{description}\vskip-3ex	}
\end{algorithm}
The vanilla HMC algorithm requires the following tuning parameters:  a positive real number $\delta t$,
a positive integer $L$
and a positive definite matrix $\bM$ 
of the same order as $\bTheta$.
Letting $\nabla V(\bTheta)=-\pdv {}{\bTheta} \log \Pi(\bTheta\mid -)$ be the gradient of the negative $\log$-posterior,
we outline the HMC-within-Gibbs sampler in Algorithm \ref{algo:hmc_msfa} to obtain $N$ MCMC samples from the joint posterior distribution of 
$(\bTheta,\tau,\bpsi,\bphi)$. 
Implementing Algorithm \ref{algo:hmc_msfa} requires computing $\log \Pi(\bTheta\mid -)$ as well as $\pdv{}{\bTheta} \log \Pi(\bTheta\mid -)$ at each MCMC step.
In Appendix \ref{sec:distributed_gradient}, we detail how these quantities are amenable to a parallelized implementation. 

Before validating theory and methods in empirical studies, we highlight two advantages of our sampler.
First, it can be adapted to any prior with (conditionally) differentiable density, including the spike-and-slab \citep{Ishwaran2005spikeslab}, spike-and-slab lasso \citep{rovckova2016fast}, horseshoe \citep{carvalho2009handling}, multiplicative gamma \citep{bhattacharya2011sparse}, generalized double Pareto \citep{armagan2013generalized}, cumulative shrinkage \citep{legramanti2019bayesian}, and others.
Almost all of these admit a conditionally Gaussian prior on $\bLambda$, so that adapting Algorithm \ref{algo:hmc_msfa} requires only modifying the Gibbs step. \label{pg:adapt_HMC}
Second, a crucial advantage of Algorithm \ref{algo:hmc_msfa} is its scalability with respect to the sample size.
In particular,
the log-likelihood $\L$ and its gradient depend on the data only through the study-specific sample sum of squares $\bW_{s}$, which just needs to be computed and cached once.
Choices of the tuning parameters are discussed in Section \ref{sm subsec:HMC_tuning} of the supplementary materials.

\section{Empirical Study}

\vspace*{-1ex}
\label{sec:simstudies}
To assess the proposed methodology, we generate synthetic data according to the true model \eqref{eq:MSFA_model} under three scenarios induced by different $\bLambda$ matrices, which we refer to as FM1, FM2 and FM3. 
The ground truth shared covariance structures are displayed in the first column of Figure \ref{fig:cormats}; complete details of the simulation setup and further figures appear in Section \ref{sm subsec:sim_truths} of the supplementary materials. 
For each scenario, we also examine performance under two misspecified cases:
(i) \textit{slight misspecification}, where $\bPhi_{s}=\bLambda \bAs+ \bE_{s}$ and $\bE_{s}$ is a matrix of randomly generated errors, and 
(ii) \textit{complete misspecification}, where $\bPhi_{s}$'s are in the null space of $\cspace{\bLambda}$ 
for all $s=1,\dots,S=5$. 
We vary $d=50, 200$ and $450$ and set $n_{1},\dots,n_{S} \simiid \max\{\Poisson(\frac{d}{S}),\frac{d}{S}\}$  
for each $d$.
We generate $\qs\sim \Poisson(\frac{q}{S})$ where $q=10$ for the $d=50$ and $q=20$ otherwise.

\vspace*{-2ex}
\paragraph*{Peer methods:} We compare SUFA to B-MSFA \citep{devito2018bayesian}, BFR \citep{alejandra2022} and PFA \citep{roy_pfa2020}.
However, PFA did not scale beyond $d>50$ and therefore we could only apply it for $d=50$.
Implementations of all these methods
require the user to input an upper bound on the number of latent factors; we choose the bound following the strategy in Section \ref{subsec:specify_q}.
We repeat BFR and SUFA over 100 replicates.
PFA and B-MSFA codes did not scale across independent parallel replications and we repeat them over 30 replications.
\begin{figure}[!h]
\centering
\includegraphics[trim={0cm .2cm .1cm 0cm}, clip,width=\linewidth]{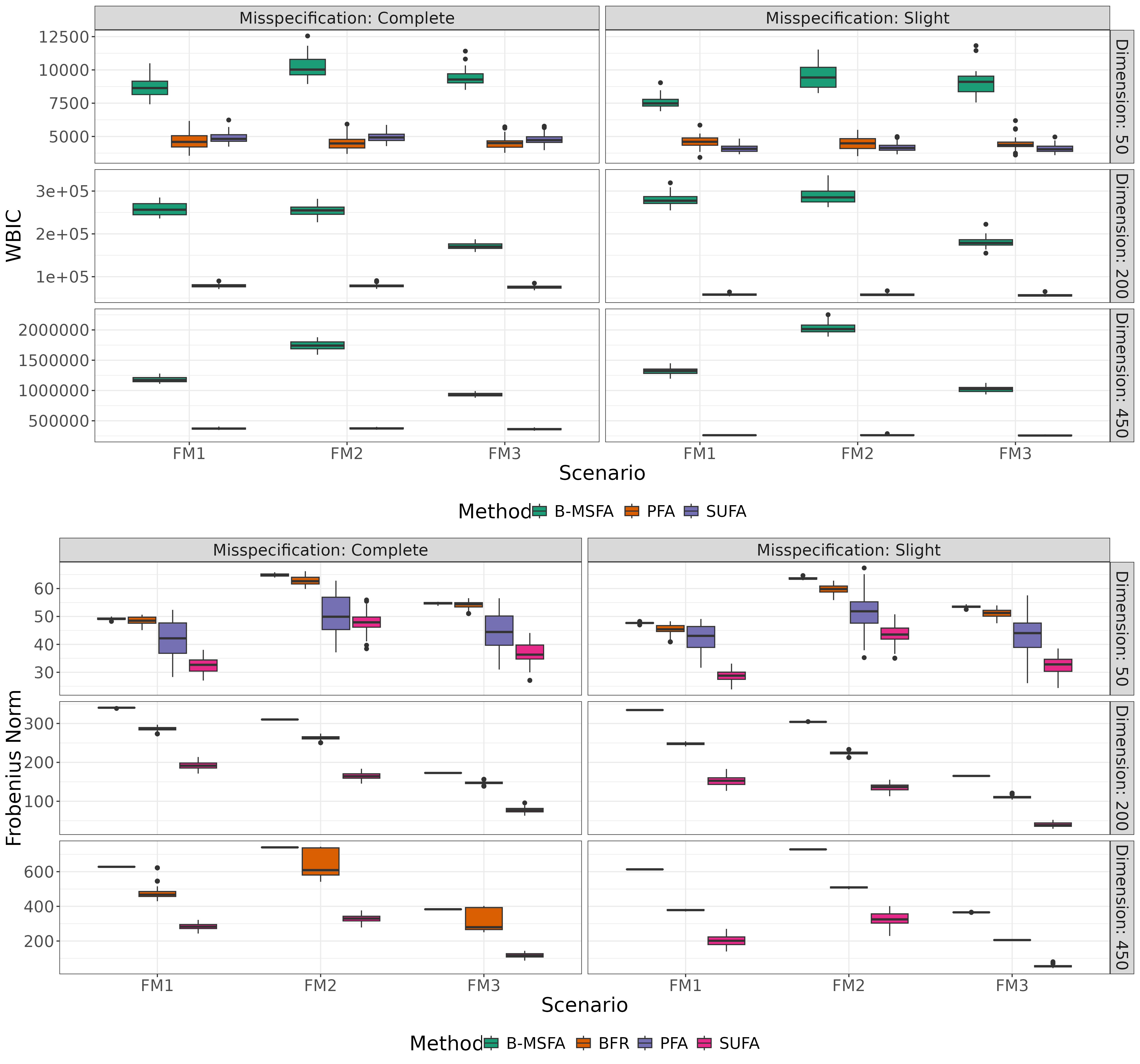}
\vskip-1ex
\caption{Comparing our proposed SUFA with B-MSFA, BFR and PFA across different simulation scenarios (FM1, FM2 and FM3) and model misspecifications: The top and bottom panels show the boxplots of the WBIC values and the boxplots of Frobenius norms between the true and estimated shared covariance $\bSigma$, respectively.}
\label{fig:main_sim}
\end{figure}   
To assess model fit, we consider the \textit{widely applicable Bayesian information criterion} \citep[WBIC,][]{wbic2013}, with lower values implying better fit.
Boxplots of the WBIC values across independent replicates are shown  in the top panel of Figure \ref{fig:main_sim}.
Notably, posterior samples are required to obtain WBIC values.
The BFR implementation provides the MAP estimate only and we can not report WBIC values for BFR.
For $d=50$, PFA performs slightly better than SUFA only in the completely misspecifed case.
As dimension grows, performance improvements under SUFA become more evident.
Identifiability issues with B-MSFA may degrade parameter estimation performance but are not expected to adversely impact goodness-of-fit measures such as WBIC.
As SUFA was motivated by improving identifiability, the better model fit is surprising, and may be explained by 
(i) improved learning of the model parameters in high-dimensional low sample size cases by SUFA due to parsimony
and
(ii) more efficient posterior sampling using the proposed HMC-within-Gibbs sampler (itself in part due to  improved parameter identifiability). 
We compute the Frobenius norms between the simulation truth and estimated values of the shared covariance matrix $\bSigma$. 
Boxplots of the norms across independent replicates are shown in the bottom panel of Figure \ref{fig:main_sim}. 
In almost every case, SUFA yields the lowest estimation errors.
\begin{figure}[!h]
\centering
\includegraphics[width=\linewidth]{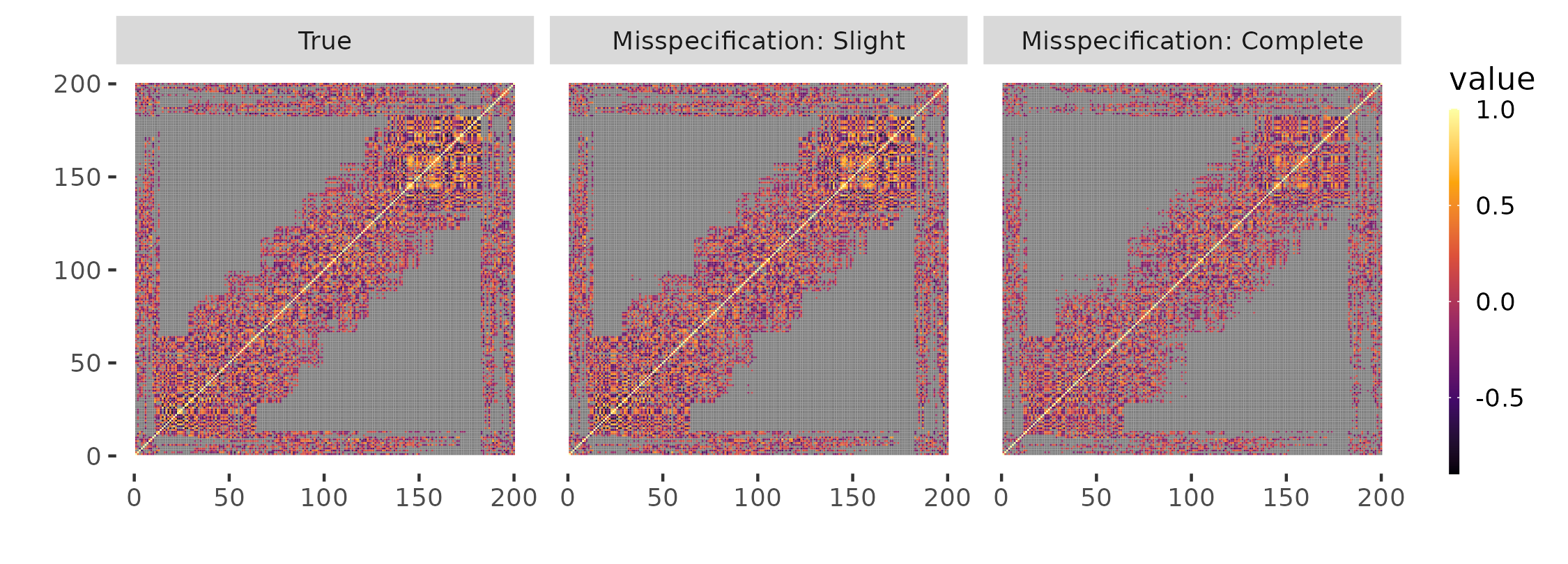}
\vskip-3.25ex
\includegraphics[width=\linewidth]{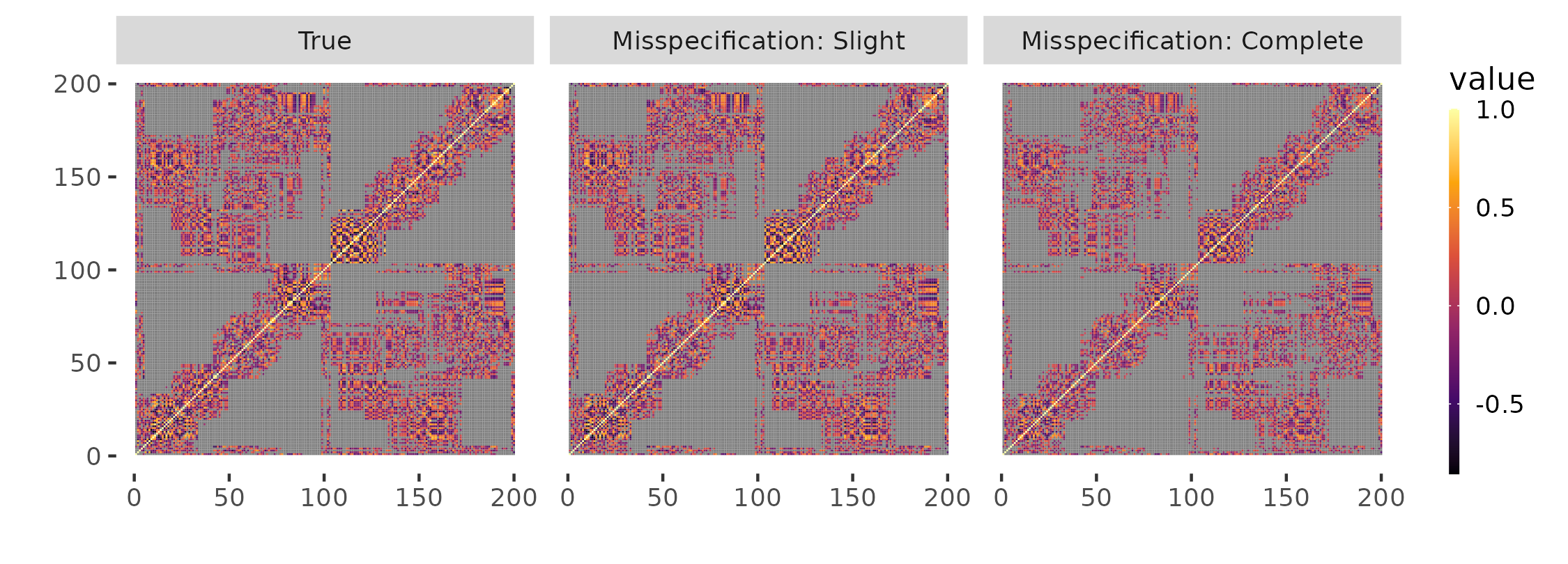}
\vskip-3.25ex
\includegraphics[width=\linewidth]{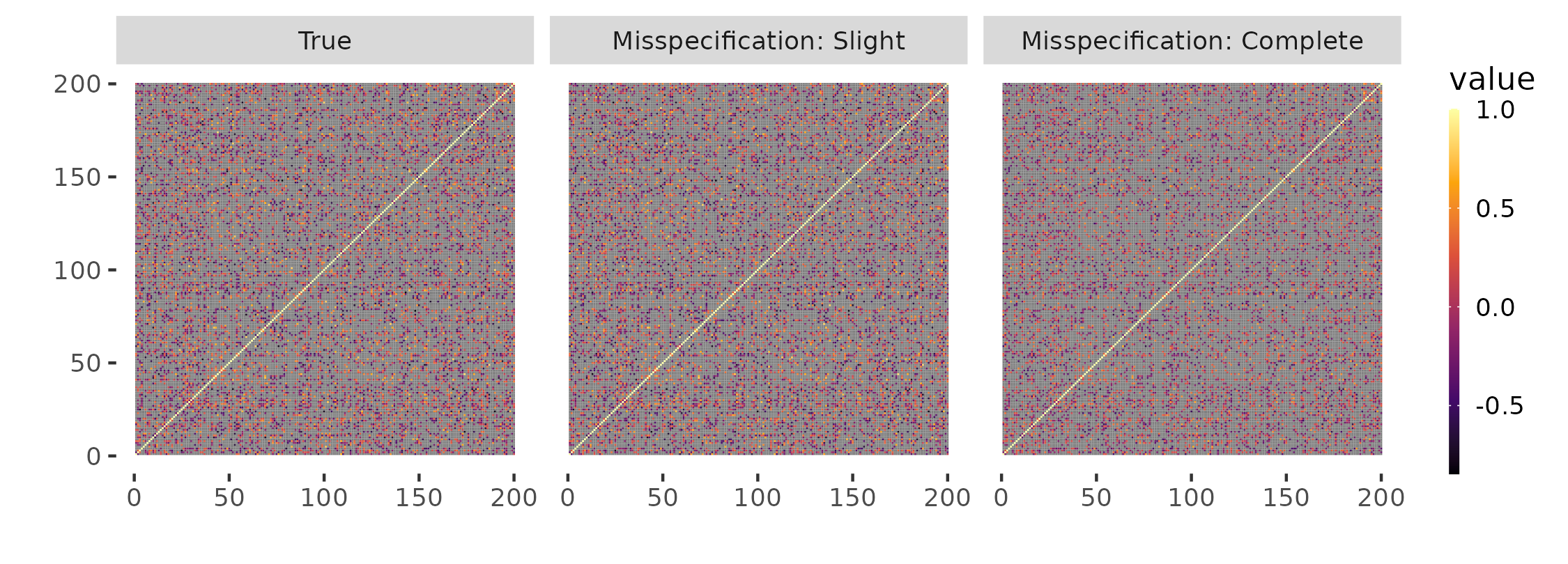}
\vskip-3.75ex
\caption{The true and estimated shared correlation matrices: The rows correspond to the three different scenarios considered.
In each row, the true correlation matrix is shown in the leftmost plot; the middle and right plots are the point estimates in the slight and complete misspecified cases, respectively.}
\label{fig:cormats}
\end{figure}

In Figure \ref{fig:cormats}, we plot the shared correlation matrices $\bR= \diag(\bSigma)^{-\half}\bSigma \diag(\bSigma)^{-\half}$ 
recovered by SUFA. 
In each setting, we plot the instance with median Frobenius norm from the ground truth out of the  the 100 independent replicates, and
we use the posterior mean of the MCMC samples as the point estimate.
Due to page limits, we only show results for $d=200$ in Figure \ref{fig:cormats}, and we use gray to depict correlations having absolute value less then $0.10$ to simplify visualization.
The heatplots under the ``\textit{True}" panel of Figure \ref{fig:cormats} 
show the simulation truths of the three different correlation structures;  
the ``\textit{Slight}" and ``\textit{Complete}" \textit{misspecification} panels show the recovered $\bR$ by the SUFA model for the two model misspecification types under consideration.
The figure clearly indicates that SUFA accurately recovers the shared correlation structures under misspecified scenarios as well.

We conduct a similar study on inferring the shared loadings. 
Here, we quantify the accuracy of the estimated $\wh{\bLambda}$ by regressing $\blambda_{j\cdot}$ on $\wh{\bLambda}$ for each $j=1,\dots,q$, where  $\blambda_{j\cdot}$ denotes the $j\th$ column of the ground truth. 
Treating $\wh{\bLambda}$ as the ``predictor matrix" in this way,
a coefficient of determination  value $R^{2}$ close to 1 suggests that the true $\bLambda$ is well-estimated up to orthogonal rotation, since $R^2$ is invariant to non-singular matrix multiplication.
Results under peer methods are summarized in Figure \ref{fig:adj_R}. 
The same post-processing strategy in Section \ref{subsec:rot_ambig} is used on all peer methods; note it does not depend on prior knowledge of $\bLambda$. 
\begin{figure}[h]
\includegraphics[trim={0cm .1cm .1cm 0cm}, clip,width=\linewidth] {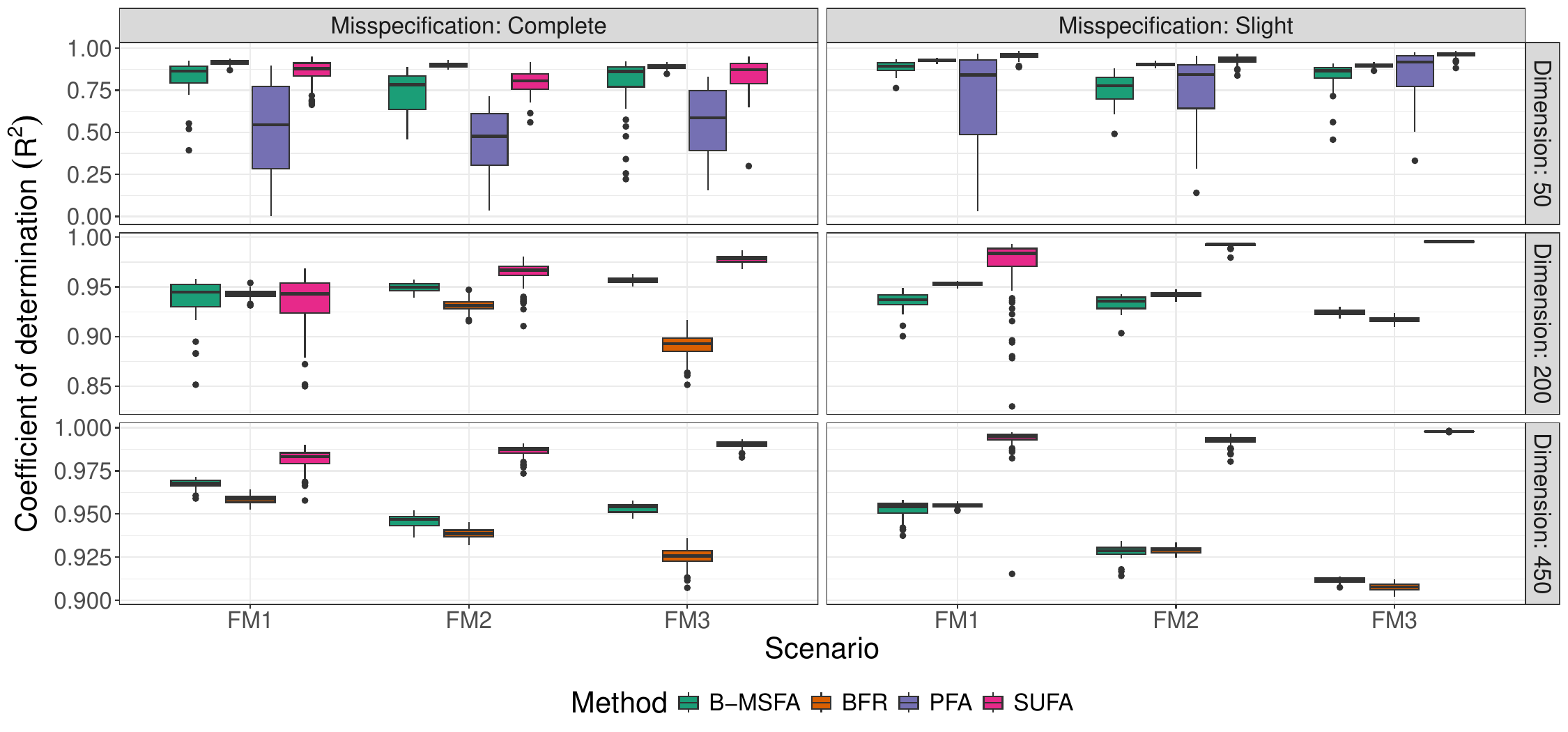}
\caption{Boxplots of median $R^2$ across columns of $\bLambda$ regressed on $\wh{\bLambda}$, across all simulation settings and alternate approaches under consideration.}
\label{fig:adj_R}
\end{figure}

Figure \ref{fig:adj_R} reveals that
even in the completely misspecified case, performance of SUFA is comparable with the other approaches,
while PFA exhibits a lot of variability in recovering $\bLambda$.
B-MSFA, BFR and SUFA all recover $\bLambda$ quite well, though SUFA exhibits better performance in high-dimensions. This is likely owing to the ability of SUFA to address the information switching issue, leading to improved learning of the shared versus study-specific structures.

While results remain favorable for inferring shared loadings, the proposed method exhibits subpar performance in recovering study-specific loadings under the complete misspecification scenario.
This may be expected as it corresponds essentially to a worst case for the assumptions behind SUFA.
Due to space limitations, we defer the complete details, figures, and discussion to Section \ref{sm subsec:additional_sim_studyspecific} in the Supplement.


Finally, we compare the total runtimes for $7,500$ MCMC iterations in  Figure \ref{fig:runtime}. 
BFR is fastest since its implementation provides the MAP estimate only without any uncertainty quantification. 
PFA is slowest due to requiring updates of $S$ dense $d\times d$  matrices.
However, owing to the distributed computing implementation of our HMC-within-Gibbs sampler, we observe dramatic improvements in computational efficiency using SUFA. 
Complete details of the parallelized implementation and a complexity analysis showing that each HMC step is of order $\O(Lqd^{2})$ appear in Appendix \ref{sec:distributed_gradient}.
\begin{figure}[!h]
\centering
\includegraphics[width=\linewidth]{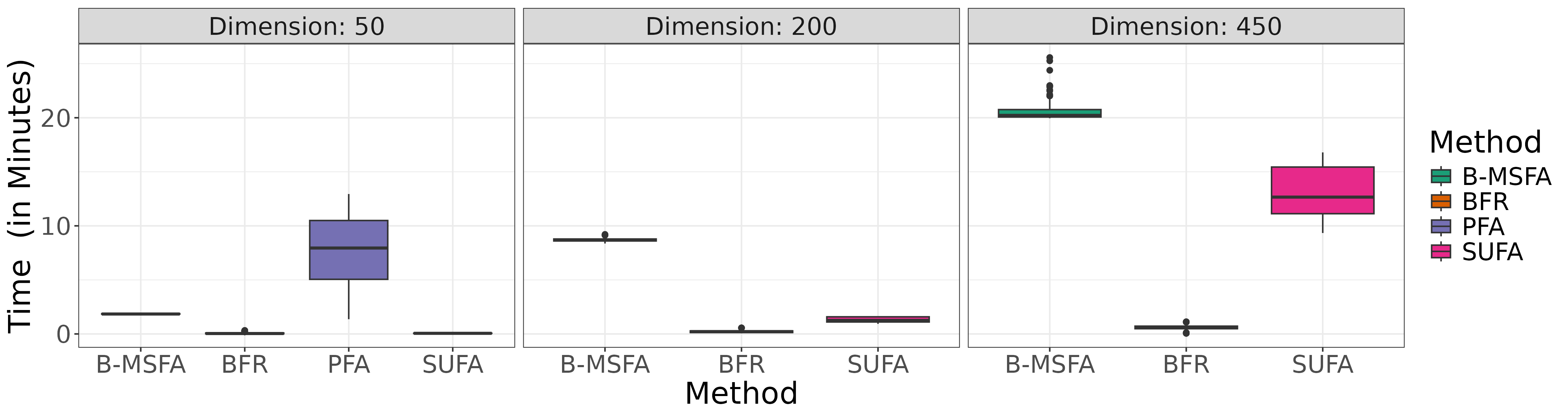}
\vskip-2ex
\caption{Comparing execution times (in minutes) of B-MSFA, BFR, PFA and SUFA across increasing dimensions.}
\label{fig:runtime}
\end{figure}

\subsection{Application to Gene Expression Data}
\label{sec:application}

We now turn to a case study on gene associations among immune cells, which serve specialized roles in innate and adaptive immune responses that function to eliminate antigens \citep{immune_cells}.  
Understanding gene associations in immune cells is of particular current interest as an essential step in developing cancer therapeutics in immunotheraphy \citep{immunotherapy}.
Here, we integrate data from three studies analyzing gene expressions. 
The first is the GSE109125 bulkRNAseq dataset, collected from 103 highly purified immunocyte populations representing all lineages and several differentiation cascades and profiled using the ImmGen  pipeline \citep{bulk_train}. 
The second study is  a microarray dataset GSE15907 \citep{micro1_train1,micro1_train2}, measured on multiple \textit{ex-vivo} immune lineages, primarily from adult B6 male mice. 
Finally, we include the GSE37448 \citep{micro2_train} microarray dataset, also  part of the Immgen project. 
After standard pre-processing (detailed in Section \ref{sm subsec:data_preprocess} of the supplementary material), we work with 474 common genes measured on 156, 146 and 628 cells from the  respective datasets. 

\begin{figure}[!b]
\centering
\begin{subfigure}{.45\linewidth}
\centering
\includegraphics[trim={.1cm .1cm .1cm .3cm}, clip,height=.275\textheight]{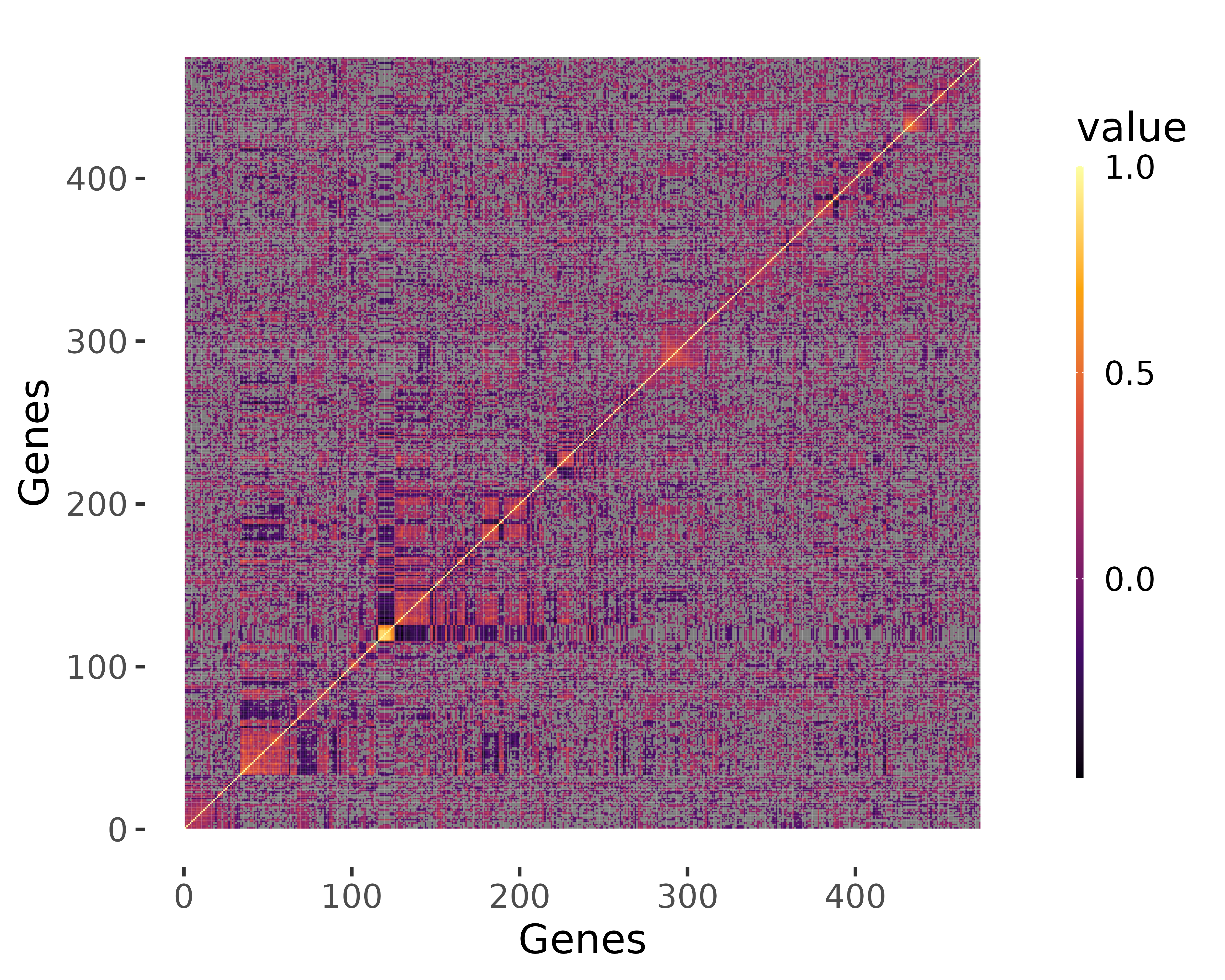}
\caption{Shared correlation matrix.}
\label{subfig:sharedmats_application}
\end{subfigure}
\begin{subfigure}{.54\linewidth}
\centering
\includegraphics[trim={1.25in 1.4in 1in 1.2in}, clip,height=.275\textheight] {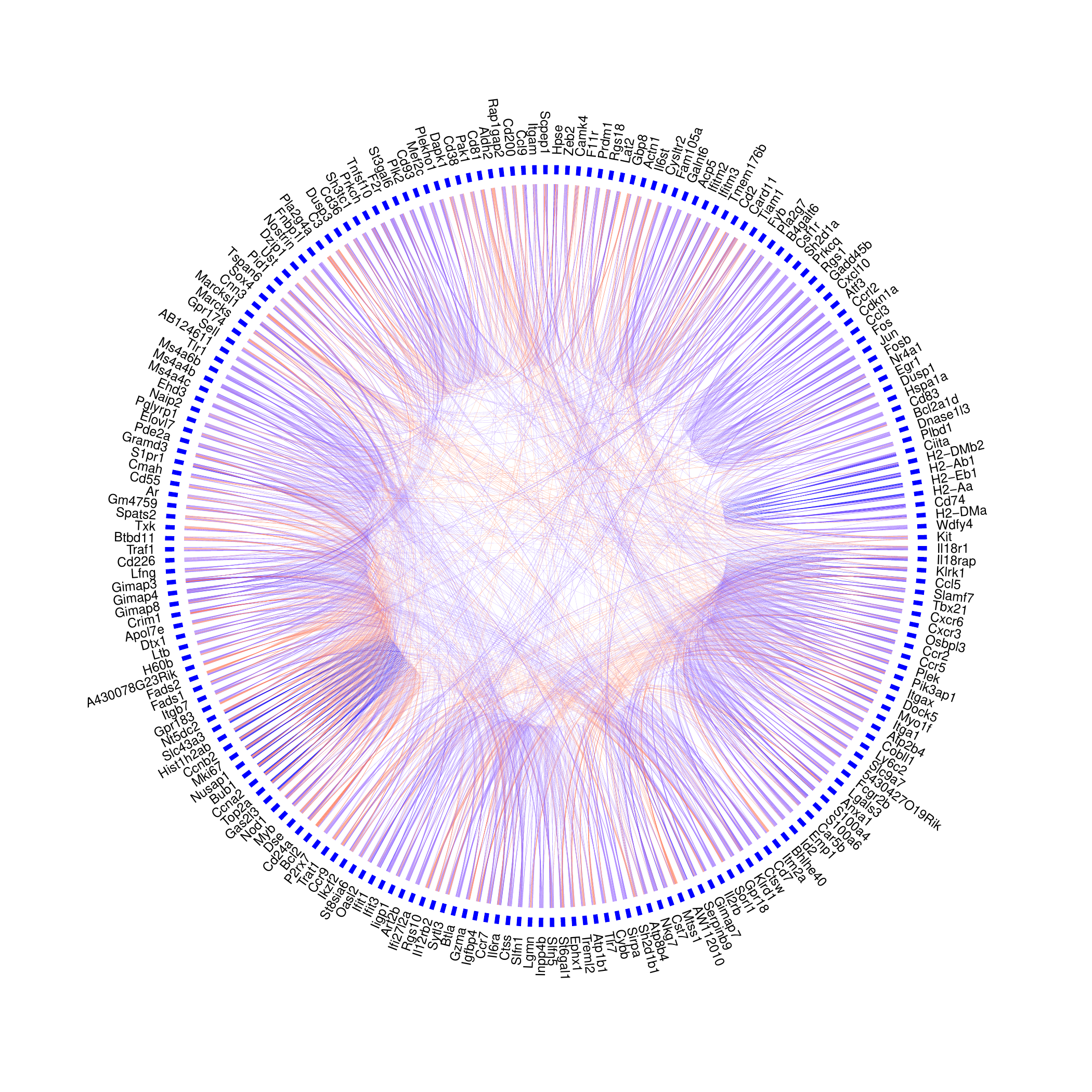}
\caption{Shared correlation network between genes.}
\label{subfig:circosplot_application}
\end{subfigure}
\end{figure}

\begin{figure}[!h]
\ContinuedFloat
\begin{subfigure}{\linewidth}
\centering
\includegraphics[trim={.18cm .42cm .1cm .18cm}, clip,width=\linewidth]{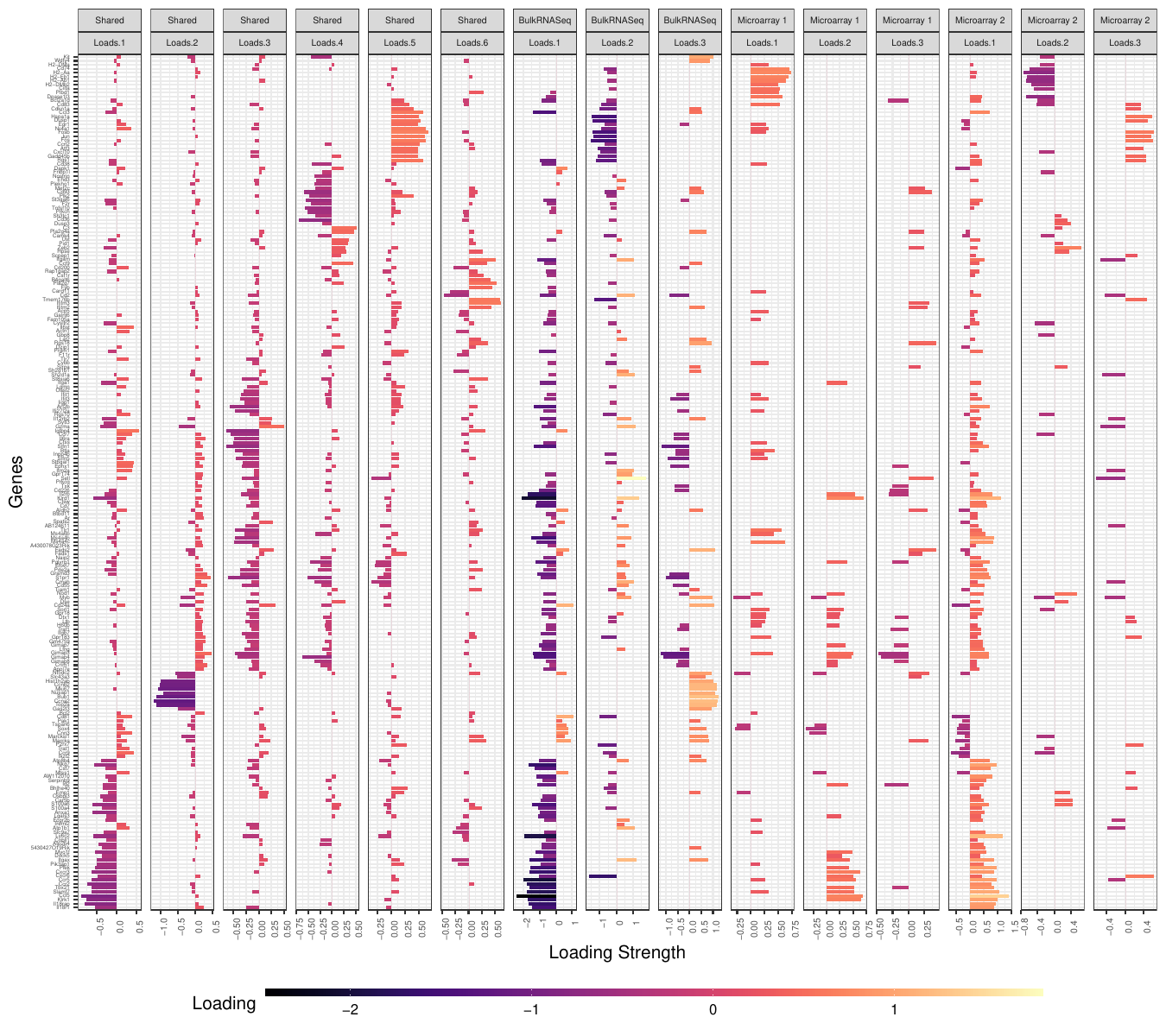}
\caption{Shared and study-specific loading matrices.}
\label{subfig:all_loads}
\end{subfigure}
\vskip-1.5ex
\caption{Results on ImmGen data: Panel (a) shows the transpose of the shared loading matrix $\bLambda$ and correlation matrix derived from $\bSigma$ from left to right; 
panel (b) is the circos plot of the dependency structure between genes, blue (red) connection implies a positive (negative), their opacities being proportional to the corresponding association strengths. 
Associations are only plotted for absolute correlation $\geq 0.25$;
panel (c) plots the top six columns with highest sum-of-squares values from the shared $\bLambda$, and repeat for the top three columns of each study-specific loading matrix.}
\label{fig:application}
\end{figure}

We apply SUFA and B-MSFA to the integrated datasets, resulting in WBIC values of $424278$ and $445916$, respectively.
\label{pg:Bmsfa_fit} 
This suggests that SUFA provides a better fit, and we hence focus on interpreting the SUFA results.
First, we visualize the estimated shared correlation structure between genes in Figure \ref{subfig:sharedmats_application}. 
We use the uncertainty quantified by the posterior samples to infer whether correlations are zero. Following the Bayesian literature on inference for sparsity patterns in covariance matrices \citep{KsheeraSagar2021precision},
we encode the off-diagonals of the correlation matrix as zero if the respective 95\% posterior credible interval contains zero.

Next, we focus on identifying important hub genes that have absolute correlation $\geq 0.25$ with at least 10 other genes. 
The resulting dependency network is summarized as a circos plot \citep{gu2014circlize} in Figure \ref{subfig:circosplot_application}.
Our findings are quite consistent with results from the prior literature: 
we observe strong positive correlation between the genes Ly6a, Ly6c1 and Ly6c2 identified within the murine Ly6 complex on chromosome 15 \citep{ly6};
similar behavior is observed within the membrane-spanning 4A class of genes, namely Ms4a4c, Ms4a4b and Ms4a6b \citep{ms4a}.
Validation of our findings via a thorough literature review is discussed in Section \ref{sm subsec:networks} of the supplementary materials.
The joint analysis of multiple datasets using our SUFA framework thus reveals interesting findings adding to the relevant immunology literature, discovering potential relationships that warrant further scientific study.

\label{pg:interpretable_factors} Even though gene expression datasets are high-dimensional, often the variability in the data is governed by relatively few latent factors often interpreted as \textit{meta-genes} or pathways 
\citep{brunet2004metagene}.
It is thus possible to potentially reduce the dimension of expression data from hundreds of genes to a  handful  of  meta-genes to  identify  distinct molecular  patterns. 
Accordingly the shared $\bLambda$  can be interpreted as the common activation profiles of the meta-genes \citep{pournara2007factor} across studies.
Differences in gene-networks across studies are captured by the study-specific latent factors ($\bzeta_{s,i}$'s in equation \eqref{eq:ourmodel}); while they are 
also loaded via the $\bLambda$ matrix, 
the respective $\bAs$ matrices allow 
varying profiles.

We visualize the six columns having largest sum-of-squares from the shared loading matrix, and repeat for the top three columns of each  study-specific loading matrix. 
These appear in Figure \ref{subfig:all_loads}.
We observe large loadings for the  structurally related genes Il18rap and Il18r1 within the Il18 receptor complex \citep{il18} in the first columns consistently across the shared and study-specific loading matrices.
The membrane-spanning 4A class of genes (Ms4a4c, Ms4a4b and Ms4a6b) exhibit similar behavior; see column 3 of the shared loading in Figure \ref{subfig:all_loads}.
These findings indicate the potential existence of meta-genes controlling  co-expression profiles along a pathway.
Additional interpretation of factor loadings are provided in Section \ref{sm subsec:loading_interpretation} of the supplementary materials.
\label{pg: loading_interpret}
Coherent with empirical studies \citep{bulk_vs_marray}, the strongest signals appear in the bulkRNAseq data despite having much smaller sample size than the microarrays.
The loadings of Microarray 2 show much stronger signals than those in Microarray 1,  possibly due to the larger sample size of the former.
These intuitive findings, corroborated by results from the prior literature, illustrate the interpretability of the inferred study-specific loadings under SUFA.

\section{Discussion} \label{sec: discussion}
This article proposes a new factor analytic approach for {covariance learning} by integrating multiple related studies in high-dimensional settings. 
We provide both practical solutions and theoretical guarantees for learning the shared and {study-specific} covariance structures, {ensuring \textit{generic identifiability}.}
Moreover, we quantify the utility of data integration under the proposed model by way of improved learning rates with each additional study in our analysis of posterior contraction properties.
The theoretical investigations on posterior concentration provide support for the study-specific inferences as well.
Overall model fit and inference on the shared covariance structure are robust to model misspecifications.
However, when model assumptions are strongly violated, we  begin to see bias in inferring study-specific terms.
These contributions are developed alongside a scalable HMC approach to posterior inference capable of handing high-dimensional datasets with massive sample size.
In realistic simulation studies, we show that our proposed method outperforms several existing approaches in multiple aspects.
Our approach yields new insights when applied to a gene network estimation problem, integrating 
immune cell gene expression datasets.

Leveraging low-rank matrix decompositions, factor models have widespread scope in contemporary data analysis regimes beyond our covariance study.
These extensions include  measurement error problems \citep{sarkar2021me}, model-based clustering \citep{chandra_lamb}, conditional graph estimation \citep{chandra2021precfact}, among many others.
Our proposed SUFA approach can be readily adopted in these domains. 
Integrating data from multiple sources is also of interest in many applications involving non-Gaussian data.
Building on copula-based factor models \citep{murray2013_copula_fa}, our approach can be hierarchically adopted to handle such data.
{When additional information is available regarding \textit{groups of features}, it can be of interest to use group sparsity priors, potentially building on the approach of \citet{schiavon2021} for more refined inference on covariance structures.}
These extensions represent several examples of promising future directions building upon the ideas in this article.

\baselineskip=14pt
\paragraph*{Supplementary materials:}
Complete prior specifications,  proofs and details of the theoretical results and technical lemmas, 
extended simulation results,
details on the gene expression datasets
and codes to reproduce the results. 

\paragraph*{Acknowledgments:}
We thank the Editor, Dr. Dylan Small, an anonymous Associate Editor and two anonymous referees for comments that led to significant improvements in the clarity and presentation of the paper.
The authors are grateful for partial funding support from NIH grants R01-ES027498 and R01-ES028804, NSF grants DMS-2230074 and PIPP-2200047, and ERC Horizon 2020 grant agreement No. 856506.

\newcommand{\Appendix}{\appendix\def\thesection{\Alph{section}}\def\thesubsection{\Alph{section}.\arabic{subsection}}}
\begin{appendix}
\Appendix
\section*{Appendix}
\appendix
\renewcommand{\thesubsection}{A.\arabic{subsection}}
\baselineskip=17pt

\subsection{Proofs of Main Results}
\label{app:proofs}
\begin{proof}[\bf Proof of Lemma \ref{lemma:colspace}]
We first show that information switching or the existence of matrices $\bC_{1}$ and $\bC_{2}$ as mentioned in \ref{identifiability_info} implies $\bigcap_{s=1}^{S}\cspace{\bAs}$ is non-null.
We let $\bC=\bC_{1}+\bC_{2}$.
Since $\bC$ is symmetric, $\bC=\bH \bD \bH\trans$ where $\bH=\begin{bmatrix}
\bh_{1}& \cdots &\bh_{q}
\end{bmatrix}$ is an orthogonal matrix and $\bD=\diag(d_{1},\dots,d_{q})$ is a diagonal matrix with real entries. 
Clearly $\bC=\sum_{i=1}^{q} d_{i}\bh_{i}\bh_{i}\trans$.
Letting $\bC^{+}=\sum_{i=1}^q d_{i}\mathbbm{1}_{d_{i}> 0}\bh_{i} \bh_{i}\trans$ and $\bC^{-}=-\sum_{i=1}^q d_{i}\mathbbm{1}_{d_{i}< 0}\bh_{i} \bh_{i}\trans$, we have $\bC=\bC^{+}- \bC^{-}$.
We establish the result in a case-by-case manner:
\begin{description}[leftmargin=3pt,itemsep=-1pt,topsep=-2ex]
    \item[Case 1: ${\bC}^{+}=\bzero$] Let $\bx$ be a non-zero vector such that $\bx\trans \bAs \bAs\trans\bx=0$.
Since by assumption $\bAs \bAs\trans +\bC$ is positive semi-definite (psd), 
$\bx\trans (\bAs \bAs\trans +\bC) \bx\geq 0$.
Now $\bx\trans (\bAs \bAs\trans +\bC) \bx= - \bx\trans \bC^{-}\bx \geq 0$ implies that $\bx\trans \bC^{-}\bx=0$ as $\bC^{-}$ is a psd matrix by construction. 
Now for any psd matrix $\bB$ and non-zero vector $\bx$, $\bx\trans\bB\bx=0$ implies that $\bx\in\nullity{\bB}$ where $\nullity{\bB}$ denotes the null space of $\cspace{\bB}$.
Hence $\nullity{ \bAs}\subset\nullity{{\bC^{-}}}$ for all $s=1,\dots,S$ implying that $\cspace{\bC^{-}}\subseteq \bigcap_{s=1}^{S} \cspace{\bAs}$.
\item[Case 2: ${\bC}^{-}=\bzero$] Due to information switching $\bAs\bAs\trans+\bC=\wt{\bA}_{s}\wt{\bA}_{s}\trans$ for $q\times\qs$ matrix $\wt{\bA}_{s}$ for all $s$.
Now $\bC=\bC^{+}\Rightarrow
\cspace{\bAs}\subset \cspace{\wt{\bA}_{s}}$.
Note that $\mathrm{rank}(\wt{\bA}_{s})\leq \qs$.
Hence
\vspace*{-3ex}
\begin{equation*}
    \mathrm{rank}(\bAs)=\qs\Rightarrow \cspace{\bAs}=\cspace{\wt{\bA}_{s}} \Rightarrow \cspace{\bC^{+}}\subset \cspace{\bAs}\Rightarrow \cap_{s=1}^{s}\cspace{\bAs} \supset \cspace{\bC^{+}} \neq \phi.
\end{equation*}
\vspace*{-6.5ex}

    \item[Case 3: ${\bC}^{+},{\bC}^{-}\neq\bzero$] If possible let $\cap_{s=1}^{S}\cspace{\bAs}=\phi$.
    This assumption implies that $\cup_{s=1}^{S}\nullity{\bAs}=\rR^{q}$.
    Note that $\bAs\bAs\trans+\bC= \wt{\bA}_{s} \wt{\bA}_{s}\trans + \bC^{-}$ where $\wt{\bA}_{s}=\begin{bmatrix}
        \bAs & \bC^{+\half}
    \end{bmatrix}$.  
    Hence $\cspace{\bAs}\subset \cspace{\wt{\bA}_{s}}$. 
    Using the argument from case 1 we have $\cspace{\bC^{-}}\subset \cspace{\wt{\bA}_{s}}$.
    Since $\bC^{-}$ is non-null, we can find a non-zero vector $\bx \in \cspace{\bC^{-}}$.
    Note that by construction $\cspace{\bC^{-}}$ and $\cspace{\bC^{+}}$ are orthogonal and hence $\bx\trans \bC^{+}\bx=0$.
    Also $\bx \in \nullity{\bA_{\wt{s}}}$ for some $\wt{s}$ as $\cup_{s=1}^{S}\nullity{\bAs}=\rR^{q}$ implying that $\bx\trans\bA_{\wt{s}}\bA_{\wt{s}}\trans\bx=0$.
    Now $\bA_{\wt{s}}\bA_{\wt{s}}\trans+\bC$ is a psd matrix implies that $\bx\trans\left(\bA_{\wt{s}}\bA_{\wt{s}}\trans+\bC\right)\bx\geq 0$. 
    By definition of $\bx$
    \vspace*{-1.5ex}
    \begin{equation*}
       \bx\trans\left(\bA_{\wt{s}}\bA_{\wt{s}}\trans+\bC\right)\bx= \bx\trans\bA_{\wt{s}}\bA_{\wt{s}}\trans \bx+ \bx\trans\bC^{+} \bx -\bx\trans\bC^{-} \bx= -\bx\trans\bC^{-} \bx<0.\label{eq:contradiction}
    \end{equation*}
    \vskip-2.5ex
    Since $\bx \in \cspace{\bC^{-}}$, $\bx\trans\bC^{-} \bx>0$ leading to the contradiction in the above equation. 
    Hence $\cap_{s=1}^{S}\cspace{\bAs}$ can not be null.
\end{description}
\vskip2ex

Next we show that a non-null $\bigcap_{s=1}^{S}\cspace{\bAs}$ guarantees the existence of a (not necessarily unique) psd matrix $\bC$ such that $\bAs \bAs\trans -\bC\succeq \bzero$ for all $s$ resulting in information switching. 
Let $\bH=\begin{bmatrix}
\bh_{1}& \cdots & \bh_{r}
\end{bmatrix} $ be a $q\times r$ matrix with orthonormal columns such that the columns form a basis of $\bigcap_{s=1}^{S} \cspace{ \bAs}$. 
Let $\bAs\bAs\trans=\bH_{s} \bD_{s}\bH_{s}\trans$ be the spectral decomposition of $\bAs\bAs\trans$. 
Clearly $\bD_{s}$ are diagonal matrices with non-negative entries. 
Let $\varsigma$ be the minimum among all positive entries of $\bD_{s}$ across all $s=1,\dots,S$. 
Define $\bC=\varsigma \bH\bI_{r}\bH\trans$. 
By construction $\bx\trans\bAs \bAs \trans \bx=0$ implies that $ \bx\trans\bC\bx=0$. 
Also for $\bx\in \cspace{\bC}$, we have $0<\bx\trans\bC\bx\leq \varsigma \norm{\bx}^2\leq \bx\trans\bAs \bAs\trans\bx$. 
Therefore, $\bAs \bAs\trans-\bC\succeq \bzero$ for all $s=1,\dots,S$, establishing the result.
\end{proof}

\begin{proof}[\bf Proof of Theorem \ref{th:no_prior_support}]
when information switching occurs,  Lemma \ref{lemma:colspace} guarantees that $\bigcap_{s=1}^{S} \cspace{\bAs}$ is non-null, so there exist vectors $\bx_{1},\dots,\bx_{S}$ such that $\bA_{1}\bx_{1}=\dots =\bA_{S} \bx_{S}\neq \bzero$. 
Now, for any non-degenerate continuous prior on $\bAs,s=1,\dots,S$, the augmented matrix 
$\begin{bmatrix}
\bA_{1}&\cdots & \bA_{S}
\end{bmatrix}$
is full rank if and only if $\sum_{s=1}^{S} \qs\leq q$.
Hence, the necessary and sufficient condition in Lemma \ref{lemma:colspace} has zero prior support.
\end{proof}

\begin{proof}[\bf Proof of Theorem \ref{thm: posterior concentration}]
For two sequences $\an,b_{n} \geq 0$ 
we let $\an\asymp b_{n}$ imply $0< \liminf \abs{{\an}/{b_{n}}} \leq \limsup \abs{{\an}/{b_{n}}}<\infty$.
We define $\bThetan=\{\bLambdan,\bDeltan,\bA_{1n},\dots,\bA_{Sn}\} $ 
as the set of all parameters, and
let $\Pn$ and $\Pnn$ denote the joint distributions of $\Dn$ under $\bThetan$ and the true value $\bThetann$, respectively. 
to be the Kullback-Leibler (KL) divergence between $\Pnn$ and $\Pn$.

Let $\PPn=\Pnon\cup \Pntw$ be a partition of the parameter space $\PPn$ where $\Pnon= \{\bThetan: \specnorm{\bAsn}^{2}\leq C n \taun^{2} \text{ for all } s=1,\dots,S \} $ for some large enough constant $C>0$ with  $n \taun^{2}=\cn^{2}\qnn\sn\log(\dn\qn)$ and $\Pntw=\Pnon^{\complement}$.
Note that $\pin \left(\specnorm{\bSigman-\bSigmann}>M\varepsilonn\mid \Dn \right)\leq \pin (\bThetan\in \Pnon: \specnorm{\bSigman-\bSigmann}>M\varepsilonn\mid \Dn )+\pin \left(\Pntw \mid \Dn \right)$.

To prove the theorem, we show that, as $n \to \infty$, the posterior distribution on $\Pnon$ concentrates around $\bSigmann$ while the remaining mass assigned to $\Pntw$ diminishes to zero. 
We define $B_{n,0} (\bThetann , \epsilon) = \left\{\bThetan \in \PPn : \kl{\Pnn}{\Pn} \leq n\epsilon^{2}\right\}$ as the set of probability measures characterized by $\bThetan$ within $\epsilon$-radius of $\Pnn$ in KL divergence. 
We state the following lemma.
\begin{lemma}
\label{th:ghosal}
Let	 $\en$ be a semimetric on $\PPn$ such that
$\en(\bThetan,\bThetann)= \frac{1}{n\taun^{2}\cn^{3} } (\specnorm{\bSigman-\bSigmann}-\cn^{2}\epsilonn )$ where  $\epsilonn\asymp \cn n\taun^{3}$.   
Then $\pin \left\{\bThetan\in\Pnon:  \en\left(\bThetan, \bThetann\right) >M\taun\mid \Dn \right\} \to 0$ in $\Pnn$-probability for a sufficiently large $M$ as long as 
the following conditions hold for sufficiently large $j\in \nN$:
\begin{enumerate}[label=({{\Roman*}})]
\item\label{cond1} For some $C>0$, $\pin\left\{B_{n,0}(\bThetann,\taun)\right\}\geq e^{-Cn\taun^{2}}$.
\item\label{cond2} Define the set $G_{j,n}=\left\{\bThetan\in \Pnon: j\taun< \en \left( \bThetan,\bThetann\right)\leq 2j\taun \right\}$. 
There exist test functions $ \varphin$ such that for some $K>0$,
$\lim_{n\to \infty} \eE_{\bThetann} \varphin= 0$ and $\sup_{\bThetan\in G_{j,n}} \eE_{\bThetan} (1-\varphin)\leq \exp(- Knj^{2}\taun^{2} )$.			
\end{enumerate}
\end{lemma}

The remaining technical arguments are included in full in the Supplemental Materials. 
We prove Lemma \ref{th:ghosal}	in Section \ref{sm sec:more_proofs}, and	
verify condition \ref{cond1} in Theorem \ref{lemma:kl_support} 
and show the existence of a sequence of test functions satisfying condition \ref{cond2} in Theorem \ref{th:testfn} there. Note that $\en\left(\bThetan, \bThetann\right) >M\taun \Leftrightarrow \specnorm{\bSigman-\bSigmann} >M\varepsilonn$. 
This implies that $\pin (\bThetan\in \Pnon: \specnorm{\bSigman-\bSigmann}>M\varepsilonn\mid \Dn )\to 0$ in $\Pnn$-probability.
Subsequently applying the dominated convergence theorem (DCT)  establishes that $\lim_{n\to\infty} \eE_{\Pnn} \pin (\bThetan\in \Pnon: \specnorm{\bSigman-\bSigmann}>M\varepsilonn\mid \Dn )= 0$.
It remains to show that the remaining mass assigned to $\Pntw$ goes to 0, which is established in Theorem \ref{th:remainingmass}. 
\end{proof}

\begin{proof}[\bf Proof of Corollary \ref{cor: posterior concentration_indiv}]
Let $\bFsn=\bLambdan\bAsn\bAsn\trans\bLambdan\trans$ and $\bFsnn=\bLambdann\bAsnn\bAsnn\trans\bLambdann\trans$.
Using the notations defined in the proof of Theorem \ref{thm: posterior concentration}, we have
$\pin (\specnorm{\bFsn-\bFsnn}>M\varepsilonn\mid \Dn )\leq \pin (\bThetan\in \Pnon: \specnorm{\bFsn-\bFsnn}>M\varepsilonn\mid \Dn )+\pin \left(\Pntw \mid \Dn \right)$.    
As Theorem \ref{th:remainingmass} in the supplementary materials shows that $\pin \left(\Pntw \mid \Dn \right)\to 0$ in $\Pnn$-probability we focus on the first expression in the previous display in this proof. 

Note that $\bSigmasnn=\bSigmann+ \bFsnn$ and $\bSigmasn=\bSigman+ \bFsn$,
and therefore $\pin(\bThetan\in \Pnon: \specnorm{\bFsn-\bFsnn} >M\varepsilonn\mid \Dn  )
\leq 
\pin(\bThetan\in \Pnon: \specnorm{\bSigman-\bSigmann} >\frac{M}{2}\varepsilonn\mid \Dn  )+
\pin(\bThetan\in \Pnon: \specnorm{\bSigmasn-\bSigmasnn} >\frac{M}{2} \varepsilonn\mid \Dn  )$.    
In Theorem \ref{thm: posterior concentration} we established that 
the first term in the previous display diminishes to 0  
in $\Pnn$-probability;
in Theorem \ref{sm thm:margibal_var} of the supplementary materials we show that $\pin(\bThetan\in \Pnon: \specnorm{\bSigmasn-\bSigmasnn} >\frac{M}{2} \varepsilonn\mid \Dn  )\to0$ in $\Pnn$-probability.
Hence $\pin (\specnorm{\bFsn-\bFsnn}>M\varepsilonn\mid \Dn )\to 0$ in $\Pnn$-probability.
Subsequently applying DCT we conclude the proof.
\end{proof}

\subsection{Sampler Details and Distributed Computation}
\label{sec:distributed_gradient}
In this section, instead of $\bDelta$ we use its one-to-one transformation  $\wt{\bdelta}=(\wt{\delta}_{1},\dots, \wt{\delta}_{d})\trans =\log \diag(\bDelta) $. 
It allows updating the $\wt{\bdelta}$ vector in the unrestricted space $\rR^{d}$ resulting in simplified numerical operations.
Accordingly, we redefine $\bTheta:=(\bLambda, \wt{\bdelta}, \bA_{1},\dots, \bA_{S})$.

\vspace*{-2ex}
\paragraph*{Gradients of the $\log$-posterior:}
Following equation \eqref{eq:full_posterior}, we have 
\vskip-3ex
\begin{equation*}
\textstyle{\log \Pi(\bTheta\mid -) = \L+ \log\Pi(\bLambda \mid \bpsi, \bphi, \tau)+ \log \Pi(\wt{\bdelta})+ \sum_{s=1}^{S} \log\Pi(\bAs)},    
\end{equation*}
\vskip-1.5ex
\noindent where 
$\log \Pi(\bLambda\mid \bpsi, \bphi, \tau)= K_{\bLambda}-  \frac{1}{2\tau^{2}} \sum_{j=1}^{d} \sum_{h=1}^{q} \frac{\lambda_{j,h}^{2}}{ \psi_{j,h} \phi_{j,h}^{2} } $, 
$\log\Pi(\wt{\bdelta})=K_{\wt{\bdelta}} - \frac{1}{2\sigma^{2}_{\delta}} \sum_{j=1}^{d}(\wt{\delta}_{j}-\mu_{\delta})^{2} $ and 
$\log\Pi(\bAs)=K_{\bAs}- \frac{1}{2\baA}\sum_{j=1}^{q} \sum_{h=1}^{\qs}a_{s,j,h}^{2}$.
The constants $K_{\bLambda}$, $K_{\wt{\bdelta}}$ and $K_{\bAs}$'s are never required as they cancel out in the Metropolis-Hastings ratio.
Denoting the Hadamard product by $\odot$, the partials of $\log\Pi(\bLambda \mid -)$ 
are given below:
\vskip-5ex
\begin{align}
&\textstyle{\pdv{ 	} { \bLambda} \log \Pi(\bTheta \mid -)	=- \sum_{s=1}^{S} \bGs\bLambda \bCs -  \left(\frac{\lambda_{j,h}}{\psi_{j,h} \phi_{j,h}^{2}  \tau^{2}}\right)_{d \times q}} , \label{eq:grad_lambda}\\
&\textstyle{\pdv {} { \wt{\bdelta}} \log \Pi(\bTheta \mid -) =-\half\sum_{s=1}^{S} \diag(\bGs) \odot \exp(\wt{\bdelta})   - \left(\frac{\wt{\delta}_{1}-\mu_{\delta}} {\sigma^{2}_{\delta}},\dots, \frac{\wt{\delta}_{d}-\mu_{\delta}} {\sigma^{2}_{\delta}}\right)\trans},\label{eq:grad_delta}\\
&\textstyle{ \pdv{ 	} { \bAs} \log \Pi(\bTheta \mid -)	=- \bLambda\trans \bGs\bLambda \bAs - \frac{1}{\baA}\bAs},\label{eq:grad_A} 
\end{align}
\vskip-2ex
\noindent where $\bGs=\ns \bSigma_{s}^{-1}- \bSigma_{s}^{-1}\bW_{s}\bSigma_{s}^{-1}$ and  $\bCs=\bI_{q}+ \bAs\bAs\trans$. 

\vspace*{-2ex}
\paragraph*{Distributed gradient and likelihood computation:}
In each HMC step of Algorithm \ref{algo:hmc_msfa}, calculating $\log \Pi(\bTheta \mid -)$ and its gradients $\pdv{}{\bTheta} \log \Pi(\bTheta \mid -)$  are required for $L$ many leapfrog operations.
Note that the first terms in the RHS of \eqref{eq:grad_lambda}-\eqref{eq:grad_A} are gradients of $\L$ with respect to parameters, while the second terms are gradients of the prior densities.
Although computing the gradients of $\L$ are seemingly intensive,
these can be obtained very efficiently by dividing the calculations into parallel sub-operations across different studies.
These sub-operations can be distributed across parallel processes in multi-thread computing systems.
We elaborate the procedure in Algorithm \ref{algo:distributed_computing}, also summarized as a schematic in Figure \ref{fig:distributed_computing}. 
\begin{figure}[h]
\centering
\includegraphics[width=\linewidth]{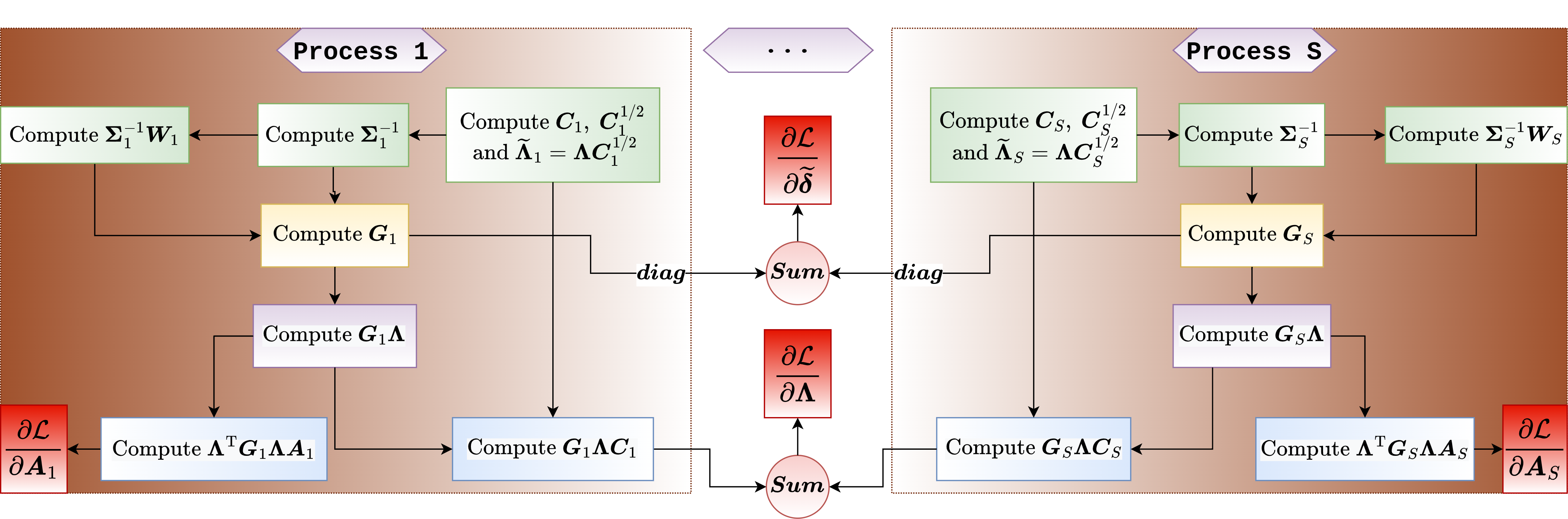}
\caption{Schematic diagram of the distributed gradient computation. The figure sketches out a flowchart to avoid repetitive calculations by storing variables and subsequently using them in following steps.}
\label{fig:distributed_computing}
\end{figure}

Lemma \ref{lem:numerical_complexity} analyzes its complexity.

\begin{algorithm}[h]
	\footnotesize
\caption{Distributed Gradient Computation of the $\log$-likelihood $\L$}
\label{algo:distributed_computing}
\SetAlgoLined
\DontPrintSemicolon
\textbf{parallel} \For{$s=1,\dots,S$}{
\compstore $~\bCs$, $\bCs^{\half}$, $\wt{\bLambda}_{s}=\bLambda\bCs^{\half}$ where $\bCs^{\half}=\texttt{Cholesky}(\bCs)$ and $\bSigma_{s}\inv=(\wt{\bLambda}_{s} \wt{\bLambda}_{s} \trans+\bDelta)^{-1}=\bDelta\inv$ $- \bDelta\inv \wt{\bLambda}_{s} (\bI_{q}+\wt{\bLambda}_{s}\trans\bDelta^{-1}\wt{\bLambda}_{s})	^{-1} \wt{\bLambda}_{s}\trans\bDelta^{-1}$.\label{mat_inv}\;
\compstore $~\bSigma_{s}^\inv \bW_{s}$ and subsequently $\bGs$.\label{step:bigmats}\;
\compstore $~\bGs \bLambda$ and $\bGs \bLambda \bCs$.\label{step:gradlambda}\;
\textbf{Compute} $\bLambda\trans \bGs\bLambda \bAs$ to obtain $\pdv{ \L	} {\bAs} \label{step5}$.
}
\textbf{Sum} the vectors $\diag(\bGs)$ from Step \ref{step:bigmats} across $s=1,\dots, S$ to obtain $\pdv{ \L} { \wt{\bdelta}}$. \label{sum_del}\;
\textbf{Sum} the matrices $\bGs \bLambda \bCs$ from Step \ref{step:gradlambda} across $s=1,\dots, S$ to obtain $\pdv{ \L	} {\bLambda} $. \label{sum_lam}	
\end{algorithm}

\begin{lemma}
\label{lem:numerical_complexity}
The runtime complexity of each HMC step is  $ \mathcal{O}(Lqd^{2})$ with $L$ being the number of leapfrog steps.
\end{lemma}

\end{appendix}



{\small 
\baselineskip=10pt
\bibliographystyle{natbib}
\bibliography{{main}}


\clearpage\pagebreak\newpage
\newgeometry{textheight=9in, textwidth=6.5in,}
\pagestyle{fancy}
\fancyhf{}
\rhead{\bfseries\thepage}
\lhead{\bfseries SUPPLEMENTARY MATERIALS}

\baselineskip 20pt
\begin{center}
{\LARGE{Supplementary Materials for\\}} 

\papertitle
\vskip10mm
\baselineskip=12pt

\authors
\vskip 20pt 
\end{center}

\setcounter{equation}{0}
\setcounter{page}{1}
\setcounter{table}{1}
\setcounter{figure}{0}
\setcounter{section}{0}
\numberwithin{table}{section}
\renewcommand{\theequation}{S.\arabic{equation}}
\renewcommand{\thesubsection}{S.\arabic{section}.\arabic{subsection}}
\renewcommand{\thesection}{S.\arabic{section}}
\renewcommand{\thepage}{S.\arabic{page}}
\renewcommand{\thetable}{S.\arabic{table}}
\renewcommand{\thefigure}{S.\arabic{figure}}
\baselineskip=15pt

\vskip 10mm
Supplementary materials provide further details on Hamiltonian Monte Carlo,
complete prior specifications,  
proofs and details of the theoretical results and technical lemmas, 
extended simulation results,
and details on the gene expression datasets. 

\baselineskip=16pt
\newpage

\section{Details on Prior Specifications} \label{sec: sm hyperparameters}
\paragraph*{Specifics of the Dirichlet-Laplace (DL) prior:} On a $d$-dimensional vector $\btheta$, the DL prior with parameter $a$, denoted by $\DL(a)$, can be specified hierarchically as
\begin{equation}
\textstyle{\theta_{j} \mid \bphi,\tau \simind  \DE(\phi_{j} \tau ),~
\bphi\sim \Dir(a,\ldots,a),~
\tau \sim \Ga\left(da,\half\right)}, \label{eq:DL_original}    
\end{equation}
where $\theta_{j}$ is the $j\th$ element of $\btheta$, 
$\tau \in \rR$, $ \bphi \in \rR^{d}$,
$\DE(b)$ is a Laplace distribution with mean 0 and variance $2b^{2}$,
$\Dir(a_{1}, \dots, a_{d} )$ is the $d$-dimensional Dirichlet distribution, 
and $\Ga(a,b)$ is the gamma distribution with mean $a/b$ and variance $a/b^{2}$.
For straightforward Gibbs sampling, often the following equivalent representation of model \eqref{eq:DL_original} is used
\begin{equation}
\textstyle{\theta_{j} \mid \bpsi,\bphi,\tau \simind \mn(0,\psi_{j}\phi_{j}^{2}\tau^{2}),~
\psi_{j} \simiid \Exp\left(\half\right),~
\bphi\sim \Dir(a,\ldots,a),~
\tau \sim \Ga\left(da,\half\right)}, 
\label{eq:dir_laplace_prior}
\end{equation}
\noindent where $\bpsi \in \rR^{d}$ 
and $\Exp(a)$ is an exponential distribution with mean $1/a$.

\paragraph*{Choice of hyperparameters:} Following the suggestions by \citetlatex{dir_laplace} we set $a=\half$ as the DL hyperparameter.
Recall that we have assumed $a_{s,j,h}\simiid \mn(0,\baA)$ and 
$\log \delta_{j}^{2}\simiid \mn(\mu_{\delta}, \sigma_{\delta}^{2})$ where $\bAs=((a_{s,j,h}))$ and $\bDelta=\diag(\delta_{1}^{2},\dots,\delta_{d}^{2})$.
To specify weakly informative priors,
we set $\baA=1$ and choose $\mu_{\delta}, \sigma_{\delta}^{2}$ such that $\eE(\delta_{j}^{2})=1$ and $\var(\delta_{j}^{2})=7$ a priori for all $j=1,\dots,d$.

\section{Proofs of Theoretical Results} \label{sec: sm proofs}
\paragraph*{Notations:} 
For two sequences $a_{n},b_{n} \geq 0$, $a_{n}\precsim b_{n}$ implies that $a_{n}\leq C b_{n}$ for some constant $C>0$;
$\an\asymp b_{n}$ implies that $0< \liminf \abs{{\an}/{b_{n}}} \leq \limsup \abs{{\an}/{b_{n}}}<\infty$.
$\abs{\bA}$ denotes the determinant of the square matrix $\bA$.
For a set $S$, $\abs{S}$ denotes its cardinality.
Let $\norm{\bx}$ be the Euclidean norm of a vector $\bx$.
We define $\bThetan=\{\bLambdan,\bDeltan,\bA_{1n},\dots,\bA_{Sn}\} $ 
as the set of all parameters and
$\bThetann=\{\bLambdann,\bDeltann, \bA_{01n},\dots,\bA_{0S}\}$ as the true data-generating values, and
let $\Pn$ and $\Pnn$ denote the joint distributions of $\Dn=\{\bY_{1,1},\dots,\bY_{1,n_{1}},\dots, \bY_{S,1},\dots,\bY_{S,n_{S}}\}$ under $\bThetan$ and the true value $\bThetann$, respectively. 
We define $\kl{\Pnn}{\Pn}$ 
to be the Kullback-Leibler (KL) divergence between $\Pnn$ and $\Pn$.

{For brevity of notations, we reuse $C$, $\wt{C}$, $C''$, etc. in the proofs to denote constants whose values may not be the same throughout the same proof.
Nevertheless, we are careful to make sure that these quantities are indeed constants.}

We define the following quantities extensively used in the proofs\vskip-6ex
\begin{align*}
\tLambdann&=\begin{bmatrix} \bLambdann & \bzero_{\dn\times {(\qn-\qnn)}} \end{bmatrix}:   \bLambdann\text{ padded with $\bzero$ columns to have the same order as $\bLambdan$},\\
\tAsnn&=\begin{bmatrix} \bAsnn & \bzero_{\qsnn  \times {(\qsn-\qsnn)}} \end{bmatrix}:   \bAsnn\text{ padded with $\bzero$ columns to have the same order as $\bAsn$},\\
\tLambdasnn&=\tLambdann\tAsnn:  \text{product of the padded matrices to have the same order as $\bLambdan\bAsn$}.
\end{align*}

\begin{theorem}[KL support]
$\pin \left\{B_{n,0}(\Pnn,\taun)\right\} \geq e^{-Cn\taun^{2}}$ for $n\taun^{2}\asymp\cn^{2} \sn\qnn\log(\dn\qn)$.
\label{lemma:kl_support}
\end{theorem}
\begin{proof}
Since the priors on $\bLambdan$, $\bAsn$ and $\bDeltan$ are independent,
from Lemma \ref{sm lemma:kl_bound}, we have
\begin{align}
&\pin\left\{B_{n,0}(\Pnn,\taun)\right\} \geq \pin\left\{\sum_{s=1}^{S}  \frac{\ns\fnorm{\bSigmasnn - \bSigmasn}^{2} } { n\delmin^{2} \smin{\bSigmasn}}  \leq \taun^{2}\right\}\notag \\
\geq& \pin\left( \fnorm{\bLambdan-\tLambdann}<\frac{C\taun} {\qnn\sqrt{\cn}} \right) \times   \pin\left( \max_{1\leq s \leq S} \fnorm{\bAsn-\tAsnn}<\frac{C\taun}{\cn\sqrt{\qnn}} \right)\notag\\ 
&\qquad\qquad\qquad\qquad\qquad\qquad\qquad
\times \pin\left\{\fnorm{\bDeltann - \bDeltan} \leq C\taun , \smin{\bDeltan}\geq \nu \right\}
\label{eq:productprior}.
\end{align}
\noindent We handle the 
$\bAs$, $\bLambda$ 
and $\bDelta$ parts in \eqref{eq:productprior} separately. 
We then conclude the proof by showing that each part 
individually exceeds $e^{-Cn\taun^{2}}$ for some constant $C>0$.

\begin{description}[leftmargin=7.5pt]
\item[The $\bAsn$ term in \eqref{eq:productprior}:]
Note that the priors on $\bAsn$ are independent.
Also from \ref{ass3} $\fnorm{\tAsnn}< \sqrt{\qsnn} \specnorm{\bAsnn}=o(\sqrt{\qsnn\qnn})$.
Further using Lemma \ref{sm lemma:gaussianbdd} and the conditions 
$\sum_{s=1}^{S}\qsnn\leq \qnn$ and $\sum_{s=1}^{S}\qsn\leq \qn$,
we obtain
\begin{align}
\pin \left(\max_{1\leq s \leq S} \fnorm{\bAsn-\tAsnn} \leq \frac{C\taun}{\cn\sqrt{\qnn}} \right)
&\geq e^{-C'\sum_{s=1}^{S} \max\{\qn\qsn \log\frac{\cn\sqrt{\qnn}}{\taun}, \qn\qsn\log (\qn\qsn), \qnn\qsnn \} }\notag\\
&\geq e^{-C'' \max(\qn^{2}  \log\frac{\cn\sqrt{\qnn}}{\taun}, \qn^{2} \log \qn, \qnn^{2} ) }.
\label{eq:Asnpart}
\end{align}
Using \ref{condd2} we get $\max(\qn^{2} \log \qn, \qnn^{2} )=o(n\taun^{2})$, and \ref{condd1} implies that  $\qn^{2}  \log\frac{\cn\sqrt{\qnn}}{\taun}=o\{\cn^{2}\sn\qnn\log(\dn\qn) \}$.
Thus the quantity in \eqref{eq:Asnpart} exceeds $e^{-Cn\taun^{2}}$ for some $C>0$.

\item [{The $\bLambdan$ term in \eqref{eq:productprior}:}]
Using \citet[Lemma 7.1]{pati2014} we obtain
\begin{equation}
\pin\left( \fnorm{\bLambdan-\tLambdann}< \frac{C\taun} {\qnn\sqrt{\cn}} \right) 
\geq e^{- C \max\left\{\fnorm{\bLambdann}^{2}, \sn\qnn\log\frac{\sn\qnn}{\frac{C\taun} {\qnn\sqrt{\cn}} }, \log(\dn\qn)   \right\}  }.\label{eq:Lambdaaapart}
\end{equation}

Note that from \ref{ass3} $\fnorm{\bLambdann}^{2}\leq \qnn\specnorm{\bLambdann}^{2}=o(\cn)$ and hence $\max\left\{\fnorm{\bLambdann}^{2}, \log(\dn\qn)   \right\}=o(n\taun^{2})$.
Additionally using \ref{condd1}, we obtain that 
$\sn\qnn\log\frac{\sn\qnn}{\frac{C\taun} {\qnn\sqrt{\cn}} }=\half\sn\qnn\log \frac{n\sn\qnn^{3} }{\cn\log(\dn\qn) }=o\{\cn^{2}\sn\qnn\log(\dn\qn) \}$.
Hence the RHS of 
\eqref{eq:Lambdaaapart} exceeds $e^{-Cn\taun^{2}}$ for some $C>0$.

\item[The $\bDeltan$ term in \eqref{eq:productprior}:] 
Note that
for a differentiable function $f(\cdot)$ in a compact interval $[a,b]$, using the mean value theorem we have $\abs{f(w)-f(y)}= f'(x) (b-a)$ where $a<w<x<y<b$ implying that $\abs{f(w)-f(y)} \leq (b-a)\times \sup_{x \in (a,b)} \abs{f'(x)}$.
Thus we have $\nu \leq \min_{j} \delta_{jn}^{2}\leq \max_{j} \delta_{jn}^{2} \leq \cn+M \Rightarrow \abs{  \delta_{jn}^{2} - \delta_{0jn}^{2}}\leq (\cn+M)\log(\cn+M) \abs{\log  \delta_{jn}^{2} -\log \delta_{0jn}^{2}} \leq C\cn\log\cn \abs{\log  \delta_{jn}^{2} -\log \delta_{0jn}^{2}}$ where $C,M>0$ are positive large enough constants.
Hence
\begin{align}
&\pin\left\{\fnorm{\bDeltan - \bDeltann} \leq C\taun, \smin{\bDeltan}\geq \nu \right\}\notag\\
\geq &\pin\left(\fnorm{\bDeltan - \bDeltann}\leq C\taun, \min_{1\leq j\leq \dn} \delta_{jn}^{2} \geq \nu, \max_{1\leq j\leq \dn} \delta_{jn}^{2} \leq \cn+M \right)\notag\\
\geq&\pin \left(\norm{\btdeltan-\btdeltann}\leq\frac{C'\taun}{\cn\log\cn}, \min_{1\leq j\leq\dn} \delta_{jn}^{2} \geq \nu, \max_{1\leq j\leq \dn} \delta_{jn}^{2} \leq \cn +M \right) ,
\label{eq:delta_eq}
\end{align}
where $\btdeltan=(\log\delta_{1n}^{2},\dots,\log\delta_{\dn n}^{2})\trans$ and $\btdeltann=(\log\delta_{01n}^{2},\dots,\log\delta_{0\dn n}^{2})\trans$.
Recall that $\log \delta_{jn}^{2}\simiid \mn(\mu_{\delta},\sigman^{2})$.
Adding a constant to the elements of $\btdeltann$ does not violate any assumptions 
and therefore without loss of generality the parameter $\mu_{\delta}$ can be assumed to be zero.

Note that $\norm{\btdeltan-\btdeltann}\leq\frac{C'\taun}{\cn\log\cn}\Rightarrow \max_{j}\abs{ \log \delta_{jn}^{2} - \log \delta_{0jn}^{2}}\leq  \frac{C'\taun}{\cn\log\cn}$.
Since $\taun=o(1)$ the last display in conjunction with \ref{ass4}
further implies that
$\max_{j}\delta_{jn}^{2}< \max_{j}\delta_{0jn}^{2}+C''\taun<\cn+M $ and
$\min_{j}\delta_{jn}^{2}> \min_{j}\delta_{0jn}^{2}-C''\taun> \nu$.
Hence  the set 
$\left\{\btdeltan: \norm{\btdeltan-\btdeltann}\leq\frac{C'\taun}{\cn\log\cn} \right\}$
is a subset of $\left\{ \min_{1\leq j\leq\dn} \delta_{jn}^{2} \geq \nu, \max_{1\leq j\leq \dn} \delta_{jn}^{2} \leq \cn +M \right\}$.
Therefore using Lemma \ref{sm lemma:gaussianbdd}
\begin{align}
&\pin\left\{\fnorm{\bDeltan - \bDeltann} \leq C\taun, \smin{\bDeltan}\geq \nu \right\}
\geq\pin \left(\norm{\btdeltan-\btdeltann}\leq\frac{C'\taun}{\cn\log\cn} \right)\notag\\
\geq &  \exp {-C \max\left( \dn\log\frac{\sqrt{\sigman}\cn\log\cn}{\taun}, \dn\log \dn, \frac{1}{\sigman^{2}}\norm{\btdeltann}^{2} \right)} . \label{eq:eqDeltannnnn}
\end{align}
Using \ref{ass3} we have $\norm{\btdeltann}^{2}=o\{\dn(\log\cn)^{2}\}$. Hence $\max\left( \dn\log \dn, \frac{1}{\sigman^{2}}\norm{\btdeltann}^{2} \right)=o\{\cn^{2}\sn\qnn\log(\dn\qn) \}$ using \ref{ass5}.
Additionally from \ref{condd1} $\dn \frac{ n\sigman (\log\cn)^{2} } {\sn\qnn \log(\dn\qn)}=o\{\cn^{2}\sn\qnn\log(\dn\qn) \}$.
Thus the RHS of 
\eqref{eq:eqDeltannnnn} exceeds $e^{-Cn\taun^{2}}$ for some $C>0$.
\end{description}
As we have shown that each of the product terms in \eqref{eq:productprior} individually exceeds $e^{-Cn\taun^{2}}$, we conclude the proof.
\end{proof}

\begin{theorem}[Test function]
\label{th:testfn}
Recalling the definitions from the proof of Theorem \ref{thm: posterior concentration},
we let the set $G_{j,n}=\{\bThetan \in\Pnon: j\taun \leq \en(\bThetan,\bThetann) \leq 2j\taun \}$ denote an annulus of inner and outer radii  $n\taun^{2}\cn^{3}(\cn^{2}\epsilonn+j\taun) $ and $n\taun^{2}\cn^{3}(\cn^{2}\epsilonn+2j\taun)  $, respectively,
in spectral norm around $\bSigmann$ for positive integer $j=o\left[ \frac{\sqrt{n}}{\cn^{5}} \left\{\sn\qnn\log(\dn\qn)\right\}^{-\frac{3}{2}}\right]$. 
Based on observed data $\Dn$ 
consider the following hypothesis testing problem $H_{0}:\bThetan=\bThetann$  versus  $H_{1}:\bThetan\in G_{j,n}$.

Define $\bLambdasnn=\bLambdann\begin{bmatrix} \bI_{\qnn}& \bAsnn \end{bmatrix}$, 
$\tilqsnn=\qnn+\qsnn$,
$\bDeletasnn= \bLambdasnn\trans\bDeltann\inv\bLambdasnn$,\\
$\bSigetasnn=\left(\bI_{\tilqsnn} + \bDeletasnn \right)^{-\half} \bDeletasnn \left(\bI_{\tilqsnn} + \bDeletasnn \right)^{-\half}$ and
$\bUpsilon_{s,in}= (\bI_{\tilqsnn} + \bDeletasnn )\inv \bLambdasnn\trans \bDeltann\inv \bY_{s,i}$.
Further define
$\bXisn=\bSigetasnn^{-\frac{1}{2}} \left(\sum_{i=1}^{\ns}  \bUpsilon_{s,in} \bUpsilon_{s,in}\trans\right) \bSigetasnn^{-\half}$ and the indicator functions $\varphisn= \mathbbm{1}\left( \specnorm{\frac{1}{\ns}\bXisn  - \bI_{\tilqsnn}} > \epsilonn  \right)$ for $s=1,\dots,S$.
We define the test function: $\varphin= 1-\prod_{s=1}^{S} (1-\varphisn)$.
Then for some absolute constant $K>0$
\begin{equation*}
\eE_{H_{0}}(\varphin) \to 0; \qquad \sup_{\bThetan\in G_{j,n}} \eE_{\bThetan} (1-\varphin)\leq \exp(- Knj^{2}\taun^{2} ).
\end{equation*}
\end{theorem}
\begin{proof}
\begin{description}[leftmargin=7pt]
\item[Type-I error:]
Note that under $H_{0}$
\begin{align}
&\var_{H_{0}}(\bUpsilon_{s,in}) \notag\\
=& \left(\bI_{\tilqsnn} + \bDeletasnn \right)\inv \bLambdasnn\trans \bDeltann\inv (\bLambdasnn\bLambdasnn\trans + \bDeltann) \bDeltann\inv \bLambdasnn  \left(\bI_{\tilqsnn} + \bDeletasnn \right)\inv\label{eq:sig_eta}\\
=& \left(\bI_{\tilqsnn} + \bDeletasnn \right)\inv  (\bLambdasnn\trans \bDeltann\inv\bLambdasnn\bLambdasnn\trans\bDeltann\inv \bLambdasnn + \bLambdasnn\trans  \bDeltann\inv \bLambdasnn)  \left(\bI_{\tilqsnn} + \bDeletasnn \right)\inv\notag\\
=& \left(\bI_{\tilqsnn} + \bDeletasnn \right)\inv (\bDeletasnn^{2}+\bDeletasnn)
\left(\bI_{\tilqsnn} + \bDeletasnn \right)\inv\notag\\
=&\left(\bI_{\tilqsnn} + \bDeletasnn \right)^{-\frac{1}{2}} \bDeletasnn \left(\bI_{\tilqsnn} + \bDeletasnn \right)^{-\frac{1}{2}}=\bSigetasnn,\notag
\end{align}
implying that $\bSigetasnn^{-\half} \bUpsilon_{s,in} \simiid\mn_{q}(\bzero,\bI_{\tilqsnn})\Rightarrow \frac{1}{\ns}\eE_{H_{0}} (\bXisn)=\bI_{\tilqsnn}$  for all $s=1,\dots,S$.		
From \citetlatex[Corollary 5.50]{vershynin_book},  $\eE_{H_{0}} (\varphisn) \leq 2\exp(-\wt{C}\tilqsnn \tn^{2} )$ for a universal constant $\wt{C}>0$ and any sequence $\tn$ satisfying $\tilqsnn \tn^{2}\leq \ns\epsilonn^{2}$. 
From \ref{ass1}, \ref{ass2} and \ref{condd2}
$\tilqsnn\to\infty$ and $\ns\epsilonn^{2}/\tilqsnn>\cn^{2}\qnn\sn\log(\dn\qn) \to \infty$ as $n\to\infty$.
Hence $\tn$ can be constructed such that $\tn\succ 1$ and hence $\lim_{n\to \infty} \eE_{H_{0}}(\varphisn) =0$.
Since the $\varphisn$s are independent across $s$, 
$\lim_{n\to \infty} \eE_{H_{0}}(\varphi_{n}) =1- \prod_{s=1}^{S}\lim_{n\to \infty} \{1- \eE_{H_{0}} (\varphisn) \}= 0$.

\item[Type-II error:] 
For data generating parameters $\bThetan\in G_{j,n}$, we define $\bSigetasn=\cov(\bUpsilon_{s,in})$.
Then 
\begin{equation*}
\bSigetasn=\left(\bI_{\tilqsnn} + \bDeletasnn \right)\inv \bLambdasnn\trans \bDeltann\inv (\bLambdasn\bLambdasn\trans + \bDeltan) \bDeltann\inv \bLambdasnn  \left(\bI_{\tilqsnn} + \bDeletasnn \right)\inv.
\end{equation*}
Hence for $\bz_{s,i}\simiid \mn_{\tilqsnn} (\bzero,\bI_{\tilqsnn})$ for $i=1,\dots,\ns$,
\begin{align}
&1-\varphisn
=\mathbbm{1}\left\{ \specnorm{\frac{1}{\ns}\bSigetasnn^{-\frac{1}{2}} \left(\sum_{i=1}^{\ns}  \bUpsilon_{s,in} \bUpsilon_{s,in}\trans\right) \bSigetasnn^{-\frac{1}{2}}  - \bI_{\tilqsnn}} \leq \epsilonn \right\}\notag\\
\leq& \mathbbm{1}\left\{ \specnorm{\frac{1}{\ns}\bSigetasnn^{-\frac{1}{2}} \left(\sum_{i=1}^{\ns}  \bUpsilon_{s,in} \bUpsilon_{s,in}\trans -\bSigetasn \right) \bSigetasnn^{-\frac{1}{2}}}  \geq \specnorm{ \bSigetasnn^{-\frac{1}{2}}\bSigetasn \bSigetasnn^{-\frac{1}{2}}  - \bI_{\tilqsnn}  } - \epsilonn  \right\}\notag\\
\leq& \mathbbm{1}\left(  \specnorm{\bSigetasnn\inv\bSigetasn} \specnorm{\frac{1}{\ns}  \sum_{i=1}^{\ns}  \bz_{s,i} \bz_{s,i}\trans -\bI_{\tilqsnn}  }  \geq  \specnorm{ \bSigetasnn^{-\frac{1}{2}}\bSigetasn \bSigetasnn^{-\frac{1}{2}}  - \bI_{\tilqsnn}  } - \epsilonn \right).\label{eq:test_power}
\end{align}
Note that, 
\begin{align*}
&{ \bSigetasnn^{-\half}\bSigetasn \bSigetasnn^{-\half}  - \bI_{\tilqsnn}  }={ \bSigetasnn^{-\half} (\bSigetasn -\bSigetasnn) \bSigetasnn^{-\half}}\text{ and}\\  
&\bSigetasn -\bSigetasnn= \left(\bI_{\tilqsnn} + \bDeletasnn \right)\inv {\bLambdasnn\trans}{\bDeltann\inv} (\bSigmasn-\bSigmasnn) {\bDeltann\inv}\bLambdasnn \left(\bI_{\tilqsnn} + \bDeletasnn \right)\inv\\
&\bSigetasnn^{-\half}=\left(\bI_{\tilqsnn} + \bDeletasnn \right)^{\half}			\left\{ \bLambdasnn\trans\bDeltann\inv (\bLambdasnn\bLambdasnn\trans+\bDeltann) \bDeltann\inv \bLambdasnn \right\}^{-\half}
\left(\bI_{\tilqsnn} + \bDeletasnn \right)^{\half}.			
\end{align*}
We obtain the last identity in the above display from \eqref{eq:sig_eta}. 
Therefore,
\begin{align*}
\bSigetasnn^{-\half}\bSigetasn \bSigetasnn^{-\half}  - \bI_{\tilqsnn}=\left(\bI_{\tilqsnn} + \bDeletasnn \right)^{\half}
\left\{\bLambdasnn\trans\bDeltann\inv (\bLambdasnn\bLambdasnn\trans+\bDeltann) \bDeltann\inv \bLambdasnn \right\}^{-\half}\\
\times\left(\bI_{\tilqsnn} + \bDeletasnn \right)^{-\half} {\bLambdasnn\trans} {\bDeltann\inv} (\bSigmasn-\bSigmasnn) {\bDeltann\inv}{\bLambdasnn} \left(\bI_{\tilqsnn} + \bDeletasnn \right)^{-\half}\\
\times\left\{\bLambdasnn\trans\bDeltann\inv (\bLambdasnn\bLambdasnn\trans+\bDeltann) \bDeltann\inv \bLambdasnn \right\}^{-\half} \left(\bI_{\tilqsnn} + \bDeletasnn \right)^{\half}.
\end{align*}
Lemma \ref{lemma:pati} implies that
\begin{multline}
\specnorm{\bSigetasnn^{-\half}\bSigetasn \bSigetasnn^{-\half}  - \bI_{\tilqsnn}} \geq \specnorm{\bSigmasn-\bSigmasnn}\times s_{\min}^{2} \left[ \left(\bI_{\tilqsnn} + \bDeletasnn \right)^{\half} \times \right. \\
\left. \left\{\bLambdasnn\trans\bDeltann\inv (\bLambdasnn\bLambdasnn\trans+\bDeltann) \bDeltann\inv \bLambdasnn \right\}^{-\half}   
\times\left(\bI_{\tilqsnn} + \bDeletasnn \right)^{-\half} {\bLambdasnn\trans} {\bDeltann\inv}\right]. \label{eq:rhs_complex}
\end{multline}

Let us define $\bchisnn=\bDeltann^{-\half}\bLambdasnn$.
Then $\bDeletasnn=\bchisnn\trans\bchisnn$ and the term inside the $s_{\min}^{2}$ in \eqref{eq:rhs_complex} simplifies to
\begin{multline}
\left(\bI_{\tilqsnn} + \bchisnn\trans\bchisnn \right)^{\half} \times  
\left\{\bchisnn\trans  (\bchisnn\bchisnn\trans+\bI_{\dn}) \bchisnn \right\}^{-\half}\\
\times\left(\bI_{\tilqsnn} + \bchisnn\trans\bchisnn \right)^{-\half} {\bchisnn\trans} {\bDeltann^{-\half}}.\label{eq:rhs_complex_chi}
\end{multline}
Letting $\chisrnn,r=1,\dots,\tilqsnn$ denote the eigenvalues of $\bchisnn\trans\bchisnn$ and using Lemma \ref{lemma:pati},
we have
\begin{multline}
s_{\min}^{2}\left[\left(\bI_{\tilqsnn} + \bchisnn\trans\bchisnn \right)^{\half} \times  
\left\{\bchisnn\trans  (\bchisnn\bchisnn\trans+\bI_{\dn}) \bchisnn \right\}^{-\half}\right.\\
\left.\times\left(\bI_{\tilqsnn} + \bchisnn\trans\bchisnn \right)^{-\half} {\bchisnn\trans} {\bDeltann^{-\half}}\right]\geq \smin{\bDeltann\inv} \times\\
\left[\min_{r=1,\dots,\tilqsnn}\left\{ (1+\chisrnn)^{\half} \left( {\chisrnn}{(1+\chisrnn)} \right)^{-\half} (1+\chisrnn)^{-\half}\chisrnn^{\half} \right\}\right]^{2}\\
\geq \frac{1}{\specnorm{\bDeltann} } \times  \min_{r=1,\dots,\tilqsnn} \frac{1}{1+\chisrnn}\geq \frac{1}{\specnorm{\bDeltann} } \times \frac{1}{1+\specnorm{ \bLambdasnn\trans \bDeltann\inv \bLambdasnn} }.
\label{eq:rhs_complex_chi}
\end{multline}
From \ref{ass3}-\ref{ass4} $\specnorm{\bDeltann}=o(\cn)$, $\specnorm{\bLambdasnn}^{2}=o(\cn)$.
Now combining \eqref{eq:rhs_complex}-\eqref{eq:rhs_complex_chi}, we get
\begin{align}
\specnorm{\bSigetasnn^{-\half}\bSigetasn \bSigetasnn^{-\half}  - \bI_{\tilqsnn}} &\geq
\frac{\specnorm{\bSigmasn-\bSigmasnn}}{\specnorm{\bDeltann} (1+\specnorm{ \bLambdasnn\trans \bDeltann\inv 
\bLambdasnn}) }\notag\\
&\geq \frac{\specnorm{\bSigman-\bSigmann}} {\specnorm{\bDeltann} \left(1+ \frac{\specnorm{\bLambdasnn}^{2} }{\delmin^{2}} \right) }\geq \frac{\specnorm{\bSigman-\bSigmann}}{\cn^{2}}. \label{eq:rhs_power}
\end{align}
Similarly, we have
\begin{align}
&\specnorm{\bSigetasnn\inv\bSigetasn   } \leq \specnorm{\bSigmasn }\times  \left\| \left(\bI_{\tilqsnn} + \bDeletasnn \right)^{\half} \times \right. \notag\\
&\quad\left. \left\{\bLambdasnn\trans\bDeltann^{-\half} (\bLambdasnn\bLambdasnn+\bDeltann) \bDeltann^{-\half}\bLambdasnn \right\}^{-\half}
\times\left(\bI_{\tilqsnn} + \bDeletasnn \right)^{-\half} {\bLambdasnn\trans} {\bDeltann\inv}\right\|_{2}^{2}\notag\\
\leq& \specnorm{\bSigmasn }\times 
\left[\max_{1\leq r\leq\tilqsnn}\left\{ (1+\chisrnn)^{\half} \left( {\chisrnn}{(1+\chisrnn)} \right)^{-\half} (1+\chisrnn)^{-\half}\chisrnn^{\half} \right\}\right]^{2}\label{eq:used_later}\\
\leq & \specnorm{\bSigman}+  \specnorm{\bLambdan \bAsn\bAsn\trans  \bLambdan\trans}  \leq \specnorm{\bSigman}\left(1+ \specnorm{\bAsn}^{2} \right) .
\label{eq:lhs_complex}
\end{align}
Note that $\specnorm{\bSigman}\leq \specnorm{\bSigman-\bSigmann}+\cn$.
Combining \eqref{eq:test_power}, \eqref{eq:rhs_power} and \eqref{eq:lhs_complex} we get
\begin{align}
&1-\varphisn\notag\\
\leq &\mathbbm{1}\left\{  \specnorm{\bSigman}\left(1+ \specnorm{\bAsn}^{2} \right) \specnorm{\frac{1}{\ns}  \sum_{i=1}^{\ns}  \bz_{s,i} \bz_{s,i}\trans -\bI_{\tilqsnn}  }  \geq  \frac{\specnorm{\bSigman-\bSigmann}}{\cn^{2}} - \epsilonn \right\}\notag\\
\leq &\mathbbm{1}\left\{   \specnorm{\frac{1}{\ns}  \sum_{i=1}^{\ns}  \bz_{s,i} \bz_{s,i}\trans -\bI_{\tilqsnn}  }  \geq  \frac{\specnorm{\bSigman-\bSigmann}  - \epsilonn\cn^{2}} {\cn^{2}\left(\specnorm{\bSigman-\bSigmann}+\cn\right) \left(1+ \specnorm{\bAsn}^{2} \right) }  \right\}, \label{sm eq:testfnnn}
\end{align}
where $\bz_{s,1:\ns}\simiid \mn_{\tilqsnn}(\bzero, \bI_{\tilqsnn} )$.
Using Lemma \ref{sm lemma:testfn_aux}, we see that the RHS of \eqref{sm eq:testfnnn} is bounded below by $\left( \wt{C}\sqrt{\tilqsnn}+ j\sqrt{n}\taun \right)\frac{1 }{\sqrt{\ns}} $ for some absolute constant $\wt{C}>0$ if $\bThetan\in G_{j,n}$.
Further applying \citetlatex[Equation 5.26]{vershynin_book}, for all $\bThetan \in G_{j,n}$ we get
\begin{equation*}
\eE_{\bThetan} (1-\varphisn )\leq \Pr \left\{   \specnorm{\frac{1}{\ns}  \sum_{i=1}^{\ns}  \bz_{s,i} \bz_{s,i}\trans -\bI_{\tilqsnn}  }  \geq  \left( \wt{C}\sqrt{\tilqsnn}+ j\sqrt{n}\taun \right)\frac{1} {\sqrt{\ns}} \right\}\leq  e^{-Kj^{2}\ns\taun^{2} }.
\end{equation*}
Finally, $\sup_{\bThetan \in G_{j,n}} \eE_{\bThetan} (1-\varphin)=\sup_{\bThetan \in G_{j,n}}\prod_{s=1}^{S}   \eE_{\bThetan} (1-\varphisn)\leq e^ {- K \left(\sum_{s=1}^{S} \ns\right) j^{2}\tau_{n}^{2}}$.
\end{description}
Hence the proof.
\end{proof}

\begin{theorem}[Remaining probability mass]
\label{th:remainingmass}
$\lim_{n\to\infty} \eE_{\Pnn} \pin \left(\Pntw \mid \Dn \right)= 0$.
\end{theorem}
\begin{proof}
For densities $\bp_{0n}$ and $\bp_{n}$ corresponding to $\Pnn$ and $\Pn$, respectively, we define the average KL variation $\klv{\Pnn}{\Pn}=\int \left\{\log\frac{\bp_{0n}}{\bp_{n}}-\kl{\Pnn}{\Pn}\right\}^{2}\de \bp_{0n}$
and the set $B_{n,2}(\Pnn, \tau)=\{\Pn: \kl{\Pnn}{\Pn}\leq n\tau^{2},~ 	\klv{\Pnn}{\Pn} \leq n\tau^{2}\}$.

From \citetlatex[Theorem 3.1]{banerjee2015bayesian} we find that the KL variation between $\mn_{\dn} (\bzero,\bSigmasnn)$ and $\mn_{\dn} (\bzero,\bSigmasn)$ is 
$\half\sum_{j=1}^{\dn} (1-\psi_{j})^{2}$ where $\psi_{j}$'s are eigenvalues of $\bSigmasn^{-\half}\bSigmasnn \bSigmasn^{-\half}$.
Hence $\klv{\Pnn}{\Pn}= \sum_{s=1}^{S} \ns\fnorm{\bSigmasn\inv (\bSigmasnn-\bSigmasn)^{2} \bSigmasn\inv}^{2} \leq \sum_{s=1}^{S}  \frac{\ns\fnorm{\bSigmasnn-\bSigmasn}^{2}} {2 s_{\min}^{2}(\bSigmasn) }$.
Therefore, the same three conditions in Lemma \ref{sm lemma:kl_bound} imply that 
$2\klv{\Pnn}{\Pn}\leq n\taun^{2} $ and accordingly
$\pin\left\{ B_{n,2}(\Pnn, \taun) \right\} >e^{-\wt{C} n\taun^{2}}$ for some absolute constant $\wt{C}>0$.

Note that the elements of  $\bAsn$ are iid $\mn(0, \baA)$ across all $s$.
Thus using \citetlatex[Theorem 5.39]{vershynin_book} we get that
$\pin( \specnorm{\bAsn} \geq C\sqrt{n\taun^{2} }  ) \leq e^{-C'n\taun^{2}} $ with $C'$ depending only on the constants $\baA$ and $C$.
Therefore $\pin(\Pntw)\leq \sum_{s=1}^{S} \pin( \specnorm{\bAsn} \geq C\sqrt{n\taun^{2} }  )\leq e^{-C''n\taun^{2}}$ where $C''$ is a constant that depends only on the constants $\baA$, $S$ and $C$.
Hence, $\frac{\pin(\Pntw)}{\pin\left\{B_{n,2}(\Pnn, \taun)\right\}} \leq e^{-2(C''-\wt{C})n\taun^{2} }=o(e^{-2n\taun^{2}})$.
The last display holds if the constant $C$ in the definition of $\Pntw$ is large enough, 
however, it is independent of $n$.
Using \citetlatex[Theorem 8.20]{ghosal_book} and subsequent application of DCT we conclude the proof.
\end{proof}

\begin{theorem}
\label{sm thm:margibal_var}
Under the same conditions of Theorem \ref{thm: posterior concentration},
$\pin(\bThetan\in \Pnon: \specnorm{\bSigmasn-\bSigmasnn}>M\varepsilonn\mid \Dn)\to 0$ in $\Pnn$-probability.
\end{theorem}
\begin{proof}

Recall $\epsilonn\asymp \cn n\taun^{3}$ and $n\taun^{2}=\cn^{2}\sn\qnn\log(\dn\qn)$ defined in the proof of Theorem \ref{thm: posterior concentration}.
Using these we now state a variation of Lemma \ref{th:ghosal}. stated earlier.
\begin{lemma}
\label{lemma:marg_ghosal}
Let	 $\ten$ be a semimetric on $\PPn$ such that
$\ten(\bThetan,\bThetann)= \frac{1}{\cn^{3} } (\specnorm{\bSigmasn-\bSigmasnn}-\cn^{2}\epsilonn-\frac{\cn^{3}\tilqsnn}{\sqrt{\ns}} )$. 
Then $\pin \left\{\bThetan\in\Pnon:  \ten\left(\bThetan, \bThetann\right) >C\taun\mid \Dn \right\} \to 0$ in $\Pnn$-probability for a sufficiently large $M$ and a constant $C>0$,
if for every sufficiently large $j\in \nN$  the following conditions hold 
\begin{enumerate}[label=({{\Roman*}'})]
\item \label{condp1} For some $C>0$, $\pin\left\{B_{n,0}(\bThetann,\taun)\right\}\geq e^{-Cn\taun^{2}}$.
\item \label{condp2} Define the set $\wt{G}_{j,n}=\left\{\bThetan\in \Pnon: j\taun< \ten \left( \bThetan,\bThetann\right)\leq 2j\taun \right\}$. 
There exists tests $ \varphisn$ such that for some $K>0$,
$\lim_{n\to \infty} \eE_{\bThetann} \varphisn= 0$ and $\sup_{\bThetan\in \wt{G}_{j,n}} \eE_{\bThetan} (1-\varphisn)\leq \exp(- Knj^{2}\taun^{2} )$.			
\end{enumerate}
\end{lemma}
Proof of the above lemma follows same line of arguments as Lemma \ref{th:ghosal}.
Additionally  condition \ref{condp1} is identical to \ref{cond1} from Lemma \ref{th:ghosal}.
We show the existence of a sequence of test functions satisfying condition \ref{condp2} in Lemma \ref{th:testfn marginal}.
This implies that 
\begin{equation}
\pin \{\bThetan\in\Pnon:  \ten\left(\bThetan, \bThetann\right) >C\taun\mid \Dn \}\to 0 \text{ in $\Pnn$-probability} .\label{eq:prob_marg}
\end{equation}

Note that $\ten\left(\bThetan, \bThetann\right) >C\taun \Leftrightarrow \specnorm{\bSigmasn-\bSigmasnn} >\cn^{2}\epsilonn+\frac{\cn^{3}\tilqsnn}{\sqrt{\ns}} + C\cn^{3}\taun$. 
From respective definitions of $\epsilonn$ and $\varepsilonn$, $\cn^{2}\epsilonn\asymp \varepsilonn$ and from \ref{condd2} $\varepsilonn \succ \max\left(\frac{\cn^{3}\tilqsnn}{\sqrt{\ns}},  \cn^{3}\taun\right)$.
Thus \eqref{eq:prob_marg} implies that for some large enough $M$,
$\pin (\bThetan\in\Pnon:\specnorm{\bSigmasn-\bSigmasnn} >M\varepsilonn\mid \Dn )\to 0$ in $\Pnn$-probability. 
\end{proof}

\begin{lemma}
\label{th:testfn marginal}
Recalling the definitions from Lemma \colredb{\ref{lemma:marg_ghosal}},
we let the set $\wt{G}_{j,n}=\{\bThetan \in\Pnon: j\taun \leq \ten(\bThetan,\bThetann) \leq 2j\taun \}$ 
for positive integer $j=o\left[ {\sqrt{n}} \left\{ 
\cn^{2} \sn\qnn\log(\dn\qn)\right\}^{-\half}\right]$. 
Based on observed data $\Dn$ 
consider the following hypothesis testing problem $H_{0}:\bThetan=\bThetann$  versus  $H_{1}:\bThetan\in \wt{G}_{j,n}$.
We use $\varphisn= \mathbbm{1}\left( \specnorm{\frac{1}{\ns}\bXisn  - \bI_{\tilqsnn}} > \epsilonn  \right)$ from Theorem \ref{th:testfn} as the test function.
Then for some absolute constant $K>0$
\begin{equation*}
\eE_{H_{0}}(\varphisn) \to 0; \qquad \sup_{\bThetan\in \wt{G}_{j,n}} \eE_{\bThetan} (1-\varphisn)\leq \exp(- Knj^{2}\taun^{2} ).
\end{equation*}
\end{lemma}
\begin{proof}
\begin{description}[leftmargin=7pt]
\item[Type-I error:] Please refer to corresponding section in the proof of Theorem \ref{th:testfn} where it has already been established that $\lim_{n\to \infty} \eE_{H_{0}}(\varphisn) =0$.

\item[Type-II error:] From equation \eqref{eq:used_later}, 
\begin{equation}
\specnorm{\bSigetasnn\inv\bSigetasn   } \leq \specnorm{\bSigmasn }\leq \specnorm{\bSigmasn- \bSigmasnn }+\specnorm{ \bSigmasn }<\specnorm{\bSigmasn- \bSigmasnn }+\cn.    \label{eq:ineq_marg_typ2a}
\end{equation}    
We obtain the last inequality in the previous display using \ref{ass3} and \ref{ass4}.
Using \eqref{eq:rhs_power}
\begin{equation}
\specnorm{\bSigetasnn^{-\half}\bSigetasn \bSigetasnn^{-\half}  - \bI_{\tilqsnn}} \geq
\frac{\specnorm{\bSigmasn-\bSigmasnn}}{\specnorm{\bDeltann} (1+\specnorm{ \bLambdasnn\trans \bDeltann\inv 
\bLambdasnn}) } \geq \frac{\specnorm{\bSigmasn-\bSigmasnn}} {\cn^{2}}.\label{eq:ineq_marg_typ2b}
\end{equation}

Combining \eqref{eq:test_power}, \eqref{eq:ineq_marg_typ2a} and \eqref{eq:ineq_marg_typ2b} we get
\begin{align}
&1-\varphisn\notag\\
\leq &\mathbbm{1}\left\{  \left(\specnorm{\bSigmasn- \bSigmasnn }+\cn\right) \specnorm{\frac{1}{\ns}  \sum_{i=1}^{\ns}  \bz_{s,i} \bz_{s,i}\trans -\bI_{\tilqsnn}  }  \geq  \frac{\specnorm{\bSigmasn-\bSigmasnn}}{\cn^{2}} - \epsilonn \right\}\notag\\
\leq &\mathbbm{1}\left\{   \specnorm{\frac{1}{\ns}  \sum_{i=1}^{\ns}  \bz_{s,i} \bz_{s,i}\trans -\bI_{\tilqsnn}  }  \geq  \frac{\specnorm{\bSigmasn-\bSigmasnn}  - \epsilonn\cn^{2}} {\cn^{2}\left(\specnorm{\bSigmasn-\bSigmasnn}+\cn\right)  }  \right\}. 
\end{align}
Using $j=o\left[ {\sqrt{n}} \left\{ 
\cn^{2} \sn\qnn\log(\dn\qn)\right\}^{-\half}\right]$ and 
similar arguments as in the proof of Lemma \ref{sm lemma:testfn_aux},
it can be straightforwardly seen that $\frac{\specnorm{\bSigmasn-\bSigmasnn}  - \epsilonn\cn^{2}} {\cn^{2}\left(\specnorm{\bSigmasn-\bSigmasnn}+\cn\right)  }>\left( \wt{C}\sqrt{\tilqsnn}+ j\sqrt{n}\taun \right)\frac{1 }{\sqrt{\ns}}$
for some absolute constant $\wt{C}>0$ if $\bThetan\in \wt{G}_{j,n}$.
Further applying \citetlatex[Equation 5.26]{vershynin_book}, for all $\bThetan \in \wt{G}_{j,n}$ we get
\begin{equation*}
\eE_{\bThetan} (1-\varphisn )\leq \Pr \left\{   \specnorm{\frac{1}{\ns}  \sum_{i=1}^{\ns}  \bz_{s,i} \bz_{s,i}\trans -\bI_{\tilqsnn}  }  \geq  \left( \wt{C}\sqrt{\tilqsnn}+ j\sqrt{n}\taun \right)\frac{1} {\sqrt{\ns}} \right\}\leq  e^{-Kj^{2}\ns\taun^{2} }.
\end{equation*}
Using \ref{ass1} we conclude $\sup_{\bThetan \in \wt{G}_{j,n}} \eE_{\bThetan} (1-\varphisn )\leq   e^{-K'j^{2}n\taun^{2} }$ for some $K'>0$.
\end{description}    
\end{proof}

\begin{proof}[Proof of Lemma \ref{lem:numerical_complexity}]
Each HMC step comprises calculating the gradients of the $\log$-posterior followed by evaluating the density at the proposed value.
We separately compute the numerical complexities of each.
\vspace*{-2ex}
\paragraph*{Gradient computation:} As steps \ref{sum_del} and \ref{sum_lam} are simply parallelized sums of objects with size at most $qd$, it suffices to focus attention on steps \ref{mat_inv}-\ref{step5}. 
These are parallelized across studies $s=1,\dots,S$, so we analyze the cost for each study as follows.
\begin{description}[leftmargin=7pt,topsep=.5ex, partopsep=0pt,itemsep=-1pt ]
\item[Step \ref{mat_inv}:] Computing the Cholesky factor $\bCs^{\half}$ is order $\O(q^{3})$ since $\bCs$ has a low-rank plus diagonal decomposition, and  Cholesky decomposition of a $q\times q$ matrix is  $\O(q^{3})$ in the worst case.
Next, the matrix product $\wt{\bLambda}_{s}=\bLambda\bCs^{\half}$ is $\O(q^{2}d)$.
Finally, owing to the low-rank and diagonal decomposition structure, the inversion $\bSigma_{s}\inv$ has complexity $\mathcal{O}(q^{2}d)$.
Therefore the overall complexity of this step is $\mathcal{O}\{\max(q^{2}d, q^{3}) \}=\O(q^{2}d)$ since   in our setting $q\ll d$.
\item[Step \ref{step:bigmats}:] Noting that $\bSigma_{s}\inv=(\wt{\bLambda}_{s} \wt{\bLambda}_{s} \trans+\bDelta)^{-1}=\bDelta\inv$ $- \bDelta\inv \wt{\bLambda}_{s} (\bI_{q}+\wt{\bLambda}_{s}\trans\bDelta^{-1}\wt{\bLambda}_{s})	^{-1} \wt{\bLambda}_{s}\trans\bDelta^{-1}$ reveals that computing $\bGs$ is $\mathcal{O}(qd^{2})$.

\item[Steps \ref{step:gradlambda} and \ref{step5}:] Because the factor $\bGs$ is cached from the preceding step, obtaining the derivative (following the flowchart in Figure \ref{fig:distributed_computing}) only requires simple matrix multiplies, with complexity no more than $\mathcal{O}(qd^{2})$.    
\end{description}
Hence, the dominant complexity of steps \ref{mat_inv}-\ref{step5} is  $\mathcal{O}(qd^{2})$, so that the combined computational complexity of all leapfrog steps is $\O(Lqd^{2})$.
\vspace*{-2ex}
\paragraph*{$\log$-posterior computation:} From equation \eqref{eq:full_posterior} it can be seen that evaluating $\L$ at a proposal is the expensive step involving high-dimensional matrix operations while the prior densities are quite simple to evaluate.
Therefore, we focus on the complexity of calculating $\L$:
owing to the structure of $\bSigma_{s}$, $\abs{\bSigma_{s}}$ can be computed  in $\O(q^{2}d)$.
Using the already cached $\bSigma_{s}\inv \bW_{s}$ from step \ref{step:bigmats}, $\trace(\bSigma_{s}\inv \bW_{s})$ can be computed in $\O(d)$.
Both of the aforementioned operations are done in parallel across $s$, so 
evaluating the $\log$-posterior has complexity no more than $\O(q^{2}d)$.

Hence, the runtime complexity of each HMC step is $\O(Lqd^{2})$.    
\end{proof}

\subsection{Technical Lemmas and Complete Proofs}
\label{sm sec:more_proofs}

\begin{proof}[\bf Proof of Lemma \ref{th:ghosal}]
Lemma \ref{th:ghosal} closely resembles Theorem 8.22 from \citetlatex{ghosal_book}. 
In the discussion following equation (8.22) in the book, the authors note that for the iid case, simpler theorems are obtained by using an absolute lower bound on the prior
mass and by replacing the local entropy by the global entropy. 
In particular, \ref{cond1} implies
Theorem 8.19(i) in the book. 
Similarly, \ref{cond2} is the same condition in Theorem 8.20, and so the result follows.
\end{proof}

\begin{lemma}
\label{lemma:pati}
For any two matrices $\bA$ and $\bB$,
\begin{enumerate}[label=({{\roman*}})]
\item $\smin{\bA} \fnorm{\bB}\leq \fnorm{\bA \bB} \leq \specnorm{\bA} \fnorm{\bB}$.
\item $\smin{\bA} \specnorm{\bB}\leq \specnorm{\bA \bB} \leq \specnorm{\bA} \specnorm{\bB}$.
\item \label{lemma:pati3} $\smin{\bA} \smin{\bB}\leq \smin{\bA \bB} \leq \specnorm{\bA} \smin{\bB}$.
\end{enumerate}
\end{lemma}
\begin{proof}
We borrow Lemma \ref{lemma:pati} from \citelatex{pati2014} (Lemma 1.1. from the supplement).
\end{proof}

\begin{lemma}
\label{sm lemma:gaussianbdd}
Let $\bx=(x_{1},\dots,x_{n})\trans$ be a random vector with $x_{i}\simiid\mn(0,\sigma^{2})$ and $\bx_{0}=(x_{01},\dots,x_{0n})\trans$ be a fixed vector.
Then for some absolute constant $C >0$, $\Pr(\norm{\bx-\bx_{0} }\leq \tau) \geq e^{-C \max\{n\log\frac{\sqrt{\sigma} }{\tau}, n\log n, \frac{1}{\sigma^{2}}(\norm{\bx_{0}}+\tau)^{2} \} }$.
\end{lemma}
\begin{proof}
Let us first define $v_{n}(r)$ to be the $n$-dimensional Euclidean ball of radius $r$ centered at zero and $\abs{v_{n}(r)}$ denote its volume.
For the sake of brevity, denote $v_{n}=\abs{v_{n}(1)}$, so that $\abs{v_{n}(r)}=r^{n}v_{n}$. 
Using \citetlatex[Lemma 5.3]{castillo2012needles}, $v_{n}\asymp (2\pi e)^{n/2}n^{-(n+1)/2}$.

Note that $\Pr(\norm{\bx-\bx_{0} }\leq \tau)
\geq \abs{v_{n}(\tau)}   \inf_{\bz: \norm{\bz-\bx_{0} }\leq \tau } \mn_{n}(\bz; \bzero, \sigma^{2}\bI_{n}) $ where $\mn_{n}(\bz; \bmu, \bSigma)$ denotes the density of a $n$ dimensional multivariate normal distribution with mean $\bmu$ and covariance matrix $\bSigma$ at the point $\bz$.
Now the greatest distance between the origin of the $n$-dimensional Euclidean space and the $n$-dimensional sphere of radius $\tau$ centered at $\bx_{0}$ is bounded above by $\norm{\bx_{0}}+\tau$.
Since the density of $\mn_{n}(\bzero, \sigma^{2}\bI_{n} )$ at a point $\bz$ monotonically decreases with $\norm{\bz}$, we have $\inf_{\bz: \norm{\bz-\bx_{0} }\leq \tau } \mn_{n}(\bz; \bzero, \sigma^{2}\bI_{n}) \geq \frac{1}{(\sqrt{2\pi }\sigma )^{n} } e^{- \frac{1}{2\sigma^{2}} (\norm{\bx_{0}}+\tau)^{2} }$.
Hence for absolute constants $C, C', C''>0$
\begin{align*}
&\Pr(\norm{\bx-\bx_{0} }\leq \tau)
\geq\abs{v_{n}(\tau)}   \inf_{\bz: \norm{\bz-\bx_{0} }\leq \tau } \mn_{n}(\bz; \bzero, \sigma^{2}\bI_{n})\\
\geq &C \tau^{n} (2\pi e)^{n/2}n^{-{(n+1)}/2} \times \frac{1}{(\sqrt{2\pi }\sigma )^{n} } e^{- \frac{1}{2\sigma^{2}}(\norm{\bx_{0}}+\tau)^{2} } \geq e^{n\log\frac{\tau}{\sqrt{\sigma}} -C' n(\log n +C'') - \frac{1}{2\sigma^{2}}(\norm{\bx_{0}}+\tau)^{2} } \label{eq:normalbdd}
\end{align*}
which concludes the proof.
\end{proof}

\begin{lemma}
\begin{equation*}
2\kl{\Pnn}{\Pn}= \sum_{s=1}^{S}\ns \left\{\trace(\bSigmasn\inv\bSigmasnn -\bI_{\dn}) -\log \abs{\bSigmasn\inv\bSigmasnn }  \right\} \leq  \sum_{s=1}^{S}  \frac{\ns\fnorm{\bSigmasn - \bSigmasnn}^{2} } {\delmin^{2} \smin{\bSigmasn}}.
\end{equation*}
\label{lemma:kl_bound}
\end{lemma}
\begin{proof}
Recall that $\Pn$ and $\Pnn$ denote the joint distributions of $\Dn$ 
under $\bThetan$ and the true value $\bThetann$, respectively. 
Let $f_{\pP}(\bY_{s,i})$ denote the marginal distribution of $\bY_{s,i}$ under a measure $\pP$. 
Since $f_{\Pn}(\bY_{s,i})\equiv \mn_{\dn}(\bzero, \bSigmasn) $ and
$f_{\Pnn}(\bY_{s,i})\equiv \mn_{\dn}(\bzero, \bSigmasnn) $ for all $i$, and the $\bY_{s,i}$s are independent, we have 
\begin{align}
\kl{\Pnn}{\Pn}&=\sum_{s=1}^{S} \sum_{i=1}^{\ns} \kl{f_{\Pnn}(\bY_{s,i})} {f_{\Pn}(\bY_{s,i})}\notag\\
&= \sum_{s=1}^{S} \frac{\ns}{2} \left\{\trace(\bSigmasn\inv\bSigmasnn -\bI_{\dn}) -\log \abs{\bSigmasn\inv\bSigmasnn}  \right\}.  \label{eq:kl_sum}
\end{align}

Let 
$\bH= \bSigmasn ^{-\half} \bSigmasnn \bSigmasn ^{-\half}$ and $\mathbb{KL}_{s}=-\log \abs{\bSigma_{sn}^{-1}\bSigma_{0sn}} +\trace(\bSigmasnn \bSigmasn^{-1}-\bI_{\dn})$. 
Letting $\psi_{1},\dots,\psi_{\dn}$ be the eigenvalues of $\bH$, we note that $\mathbb{KL}_{s}
=\sum_{j=1}^{\dn} \left\{(\psi_{j}-1)-\log\psi_{j} \right\}$.
Observe that for any $x>-1$, $\log(1+x)\geq \frac{x}{1+x}$;
additionally
$\psi_{j}>0$ for all $j=1,\dots,\dn$ since these are eigenvalues of the positive definite matrix $\bH$.
Hence $\log\psi_{j}\geq 1-\frac{1}{\psi_{j}}$ implying that $(\psi_{j}-1)-\log\psi_{j}\leq \frac{(\psi_{j}-1)^{2} }{\psi_{j} }$.
Therefore,
\begin{equation}
\textstyle{\mathbb{KL}_{s}\leq \sum_{j=1}^{\dn} \frac{(\psi_{j}-1)^{2} }{\psi_{j} }=\fnorm{(\bH-\bI_{\dn})^{2}\bH^{-1} }=\fnorm{(\bSigmasnn - \bSigmasn)^{2} \bSigmasn^{-1} \bH^{-1} \bSigmasn^{-1}}}.\label{eq:klbound01}
\end{equation}
Now $\bSigmasn \bH \bSigmasn= \bSigmasn^{\half} \bSigmasnn \bSigmasn^{\half}$.
Hence, from \eqref{eq:klbound01} and using Lemma \ref{lemma:pati} we arrive at
\begin{equation*}
\mathbb{KL}_{s}\leq \fnorm{\bSigmasn - \bSigmasnn}^{2} \specnorm{(\bSigmasn^{\half} \bSigmasnn \bSigmasn^{\half})^{-1} }\leq \frac{\fnorm{\bSigmasn - \bSigmasnn}^{2} } { \smin{\bSigmasnn} \smin{\bSigmasn} }\leq \frac{\fnorm{\bSigmasn - \bSigmasnn}^{2} } { \delmin^{2} \smin{\bSigmasn}}.
\end{equation*}
Combining the above display with \eqref{eq:kl_sum} we conclude the proof.
\end{proof}

\begin{lemma}
\label{sm lemma:kl_bound}
Assume that 
$\fnorm{\bLambdan-\tLambdann}<\frac{C\taun} {\qnn\sqrt{\cn}}$,
$\max_{s} \fnorm{\bAsn-\tAsnn}<\frac{C\taun}{\cn\sqrt{\qnn}}$, \\
$\fnorm{\bDeltan - \bDeltann} \leq C\taun$ and $\smin{\bDeltan}>\nu$ where $C=\nu\delmin^{2}/5$ and $0<\nu<\delmin^{2}$.
Then $2\kl{\Pnn}{\Pn}\leq n\taun^{2}$.
\end{lemma}
\begin{proof}
Using Lemma \ref{lemma:kl_bound} $\kl{\Pnn}{\Pn} \leq  \sum_{s=1}^{S}  \frac{\ns\fnorm{\bSigmasn - \bSigmasnn}^{2} } {2\delmin^{2} \smin{\bSigmasn}}$.
Now $\fnorm{\bSigmasn - \bSigmasnn} \leq 
\fnorm{\bLambdan\bLambdan\trans -\bLambdann\bLambdann\trans}+
\fnorm{\bLambdasn\bLambdasn\trans -\bLambdasnn\bLambdasnn\trans}+ \fnorm{\bDeltann - \bDeltan}$ where $\bLambdasn:=\bLambdan\bAsn$ and $\bLambdasnn:=\bLambdann\bAsnn$.
Recalling the notations defined in the beginning of Section \ref{sec: sm proofs} we have
\begin{align}
&\fnorm{\bLambdasn\bLambdasn\trans -\bLambdasnn\bLambdasnn\trans} = \fnorm{\bLambdasn\bLambdasn\trans -\tLambdasnn\tLambdasnn\trans}\notag\\
=&
\fnorm{(\bLambdasn-\tLambdasnn) \bLambdasn\trans +\tLambdasnn(\bLambdasn-\tLambdasnn)\trans}\notag\\
\leq& \fnorm{\bLambdasn-\tLambdasnn} \times \left( \specnorm{\bLambdasn} +\specnorm{\tLambdasnn} \right). \label{eq:lam_bound}
\end{align}
Then using Lemma \ref{lemma:pati} 
\begin{align}
&\fnorm{\bLambdasn-\tLambdasnn} =\fnorm{\bLambdan\bAsn-\tLambdann\tAsnn}=
\fnorm{\bLambdan(\bAsn-\tAsnn)+(\bLambdan-\tLambdann)\tAsnn}\notag\\
\leq& \specnorm{\bLambdan} \fnorm{\bAsn-\tAsnn}+\fnorm{\bLambdan-\tLambdann} \specnorm{\tAsnn}.
\label{eq:lam_bound2}
\end{align}
Note that $\fnorm{\bLambdan-\tLambdann}<\frac{C\taun}{\qnn\sqrt{\cn} }\Rightarrow \specnorm{\bLambdan}\leq \fnorm{\bLambdan}\leq \fnorm{\tLambdann} + \frac{C\taun}{\qnn\sqrt{\cn} }<\sqrt{\cn}$.
The last display holds since $\ref{ass3}\Rightarrow \specnorm{\bLambdann}=o(\sqrt{\cn/\qnn}) \Rightarrow \fnorm{\bLambdann}=o(\sqrt{\cn})$.
Combining the condition $\max_{1\leq s \leq S} \fnorm{\bAsn-\tAsnn}\leq \frac{C\taun}{\cn\sqrt{\qnn} }$ and \eqref{eq:lam_bound2} we get
\begin{equation}
\max_{1\leq s \leq S}\fnorm{\bLambdasn-\tLambdasnn} 
\leq \sqrt{\cn}\times \frac{C\taun}{\cn\sqrt{\qnn}}+ \frac{C\taun}{\qnn\sqrt{\cn}}\times \sqrt{\qnn}<  \frac{2C\taun}{\sqrt{\cn\qnn}}.\label{eq:eqqq} 
\end{equation}

Now $\fnorm{\bLambdasn-\tLambdasnn} <  \frac{C\taun}{\sqrt{\cn\qnn}} \Rightarrow \specnorm{\bLambdasn}\leq \fnorm{\bLambdasn}< \fnorm{\tLambdasnn}+\frac{C\taun}{\sqrt{\cn\qnn}}$.
Note that $\fnorm{\tLambdasnn}\leq \sqrt{\qsnn}\times \specnorm{\tLambdasnn}\leq \sqrt{\qsnn}\times \specnorm{\bLambdann}\specnorm{\bAsnn}< \sqrt{\cn\qnn}$ implying that
$\specnorm{\bLambdasn} +\specnorm{\tLambdasnn}< 2\sqrt{\cn\qnn}$.
Combining the last display with \eqref{eq:lam_bound} and \eqref{eq:eqqq} we obtain 
$\fnorm{\bLambdasn\bLambdasn\trans -\bLambdasnn\bLambdasnn\trans}\leq 2C\taun$.
Similarly $\fnorm{\bLambdan\bLambdan\trans -\bLambdann\bLambdann\trans}\leq \fnorm{\bLambdan -\tLambdann} (\fnorm{\bLambdan} +\fnorm{\tLambdann} )\leq 2C\taun$.
Using the additional condition $\fnorm{\bDeltan - \bDeltann} \leq C\taun$  and  $\smin{\bDeltan}>\nu$, we get $\fnorm{\bSigmasn - \bSigmasnn}  \leq 5C\taun$.
Furthermore $\smin{\bSigmasn}\geq \smin{\bDeltan}$.
Hence, $2\kl{\Pnn}{\Pn}\leq \sum_{s=1}^{S}  \frac{\ns\fnorm{\bSigmasn - \bSigmasnn}^{2} } { \delmin^{2} \smin{\bSigmasn}}  \leq n\taun^{2}$.
\end{proof}

\begin{lemma}
\label{sm lemma:testfn_aux}
Recall the notations in Theorem \ref{th:testfn}.	
If 
$\bThetan \in G_{j,n}$, then $\frac{\specnorm{\bSigman-\bSigmann}  - \epsilonn\cn^{2}}  {\cn^{2} \left(\specnorm{\bSigman-\bSigmann}+\cn\right) \left(1+ \specnorm{\bAsn}^{2} \right) }> \left( \wt{C}\sqrt{\tilqsnn}+ j\sqrt{n}\taun \right)\frac{\cn^{2} }{\sqrt{\ns}} $ for some absolute constant $\wt{C}>0$.
\end{lemma}
\begin{proof}	
For all $\bThetan \in\Pnon$, $1+ \specnorm{\bAsn}^{2}\leq C n\taun^{2}$ for some $C>0$ for all $s=1,\dots,S$.
For brevity of notations, we define 
$\omega=Cn \taun^{2} $,
$x=\specnorm{\bSigman-\bSigmann}$ and $a=\epsilonn\cn^{2}$.
Hence observe that for $\bThetan \in G_{j,n}$
\begin{align}
&\frac{\specnorm{\bSigman-\bSigmann}  - \epsilonn\cn^{2}}  {\cn^{2} {\left(\specnorm{\bSigman-\bSigmann}+\cn\right) \left(1+ \specnorm{\bAsn}^{2} \right) }}> \left( \wt{C}\sqrt{\tilqsnn}+ j\sqrt{n}\taun \right)\frac{1 }{\sqrt{\ns}} \label{sm eq:rhs_bdd}\\
\Leftrightarrow& \frac{x-a}{(x+\cn)w}>\left( \wt{C}\sqrt{\tilqsnn}+ j\sqrt{n}\taun \right) \frac{\cn^{2} }{\sqrt{\ns}} \Leftarrow x> \frac{a+\omega\left( \wt{C}\sqrt{\tilqsnn}+ j\sqrt{n}\taun \right)\frac{\cn^{3} }{\sqrt{\ns}}} {1-\omega\left( \wt{C}\sqrt{\tilqsnn}+ j\sqrt{n}\taun \right)\frac{\cn^{2} }{\sqrt{\ns}} }.\notag
\end{align}
The specifics on the postulated model from Section \ref{sec:theory} 
imply
$\tilqsnn\leq \qn+\qsn=o(n\taun^{2}) $.
As $j=o\left[ \frac{\sqrt{n}}{\cn^{5}} \left\{\sn\qnn\log(\dn\qn)\right\}^{-\frac{3}{2}}\right]$,
$ \omega\left( \wt{C}\sqrt{\tilqsnn}+ j\sqrt{n}\taun \right)\frac{\cn^{2} }{\sqrt{\ns}}\precsim \frac{\cn^{5}}{\sqrt{n}} { \left\{\sn\qnn\log(\dn\qn)\right\}^{\frac{3}{2}}} =o(1)$ which can be seen from \ref{condd2}.   
Therefore, $\frac{a+\omega\left( \wt{C}\sqrt{\tilqsnn}+ j\sqrt{n}\taun \right)\frac{\cn^{3} }{\sqrt{\ns}}} {1-\omega\left( \wt{C}\sqrt{\tilqsnn}+ j\sqrt{n}\taun \right)\frac{\cn^{2} }{\sqrt{\ns}} }<a+ j\times C\cn^{3} \sqrt{ \frac{(n\taun^{2})^{3} } {\ns}}$ for some absolute constant $C>0$ and large enough $n$.
Hence, \eqref{sm eq:rhs_bdd} holds if $x>a+ j\times C\cn^{3} \sqrt{ \frac{(n\taun^{2})^{3} } {\ns}}$
or equivalently if $j\taun \leq \en(\bThetan,\bThetann)$, establishing the result.
\end{proof}

\newpage  

\section{Hamiltonian Monte Carlo Algorithm}
\label{subsec:hmc_brief}
Let $\Pi(\bTheta)$ be the target density where $\Pi(\cdot)$ is differentiable with respect to $\bTheta \in \rR^{D}$.
In HMC,  a dynamical system is considered in which auxiliary ``momentum" variables $\bp\in \rR^{D}$ are introduced
and the uncertain parameters $\bTheta$ in the target distribution are treated as the variables for the displacement. 
The total energy (Hamiltonian function) of the dynamical system is defined by $H(\bTheta, \bp)= V(\bTheta) + K(\bp)/2$, where its potential energy $V(\bTheta) = -\log \Pi(\bTheta)$ and its kinetic energy $K(\bp)=\bp\trans \bM^{-1} \bp$ depends only on $\bp$ and some chosen positive definite ``mass" matrix $\bM \in \rR^{D \times D}$.
The Hamiltonian dynamics will preserve the distribution $e^{-H(\bTheta, \bp)}$ and the invariant distribution will have $\Pi(\cdot)$ to be the marginal distribution of $\bTheta$.
Using Hamilton's equations, the evolution of $\bTheta$, $\bp$ through ``time'' $t$ is given by 
\vskip-2ex
\begin{equation}
\dv{ \bp} { t}= -\pdv{ H} { \bTheta}= -\nabla V(\bTheta), \qquad 
\dv{ \bTheta} { t}= -\pdv{ H} { \bp}= -\bM^{-1} \bp.
\label{eq:hailtonial_eq}
\end{equation}
\vskip-1.5ex
If we start with $\bTheta(0)$ and draw a sample $\bp(0)$ from $\mn(\bzero,\bM)$,
the final values $\bTheta(t)$, $\bp(t)$ will provide an independent sample $\bTheta(t)$ from $\Pi$. 
The leapfrog algorithm \citeplatex{DUANE_hmc} is popularly used to approximately solve the differential equations in \eqref{eq:hailtonial_eq}.
For time step $\delta t$, we have
\vskip-7ex
\begin{align}
&\bTheta(t+ \delta t)= 	\bTheta(t) + \delta t \bM\inv \left[ \bp(t)-\frac{\delta t}{2} \nabla V\left\{\bTheta(t)\right\} \right],\label{eq:leap1}\\
&\bp(t+ \delta t)= 	\bp(t) - \frac{\delta t }{2}  \left[  \nabla V\left\{\bTheta(t)\right\} + \nabla V\left\{\bTheta(t + \delta t)\right\}\right].\label{eq:leap2}
\end{align}
\vskip-2ex
\noindent The complete HMC sampler is summarized as follows for some $\bM$, $\delta t$, and $L$ leapfrog steps.

\begin{algorithm}[!h]
\footnotesize
\caption{Hamiltonian Monte Carlo}
\label{algo:generic_hmc}

\SetAlgoLined
\DontPrintSemicolon
Initialize $\bTheta$ and simulate $\bp(0)\sim \mn(\bzero, \bM)$.\;
\For{$i=1,\dots,N$}{In iteration $i$, let the most recent sample be $(\bTheta_{i-1} ,\bp_{i-1})$, then
do the following to simulate a new sample $(\bTheta_{i} ,\bp_{i})$:\;
\label{hmcstep}
\begin{enumerate}[label=({\roman*}),leftmargin=10pt]
\item Randomly draw a new momentum vector $\bp'$ from $\mn(\bzero, \bM)$.
\item Initiate the leapfrog algorithm with $\{\bTheta(0), \bp(0)\}=(\bTheta_{i-1},\bp')$ and run the algorithm \eqref{eq:leap1}-\eqref{eq:leap2}\\ for $L$ time steps to obtain a new candidate sample $(\bTheta'', \bp'')=(\bTheta(t+ L\delta t),\bp(t+ L\delta t))$.
\item Set $(\bTheta_{i} ,\bp_{i})=(\bTheta'', \bp'')$ with probability $\min \left[1, e^{-\{ H\left(\bTheta'', \bp''\right)- H\left(\bTheta_{i-1},\bp'\right)\}} \right]$.
\end{enumerate}\vskip-3ex}
\end{algorithm}
\noindent Thus we obtain the MCMC samples $\bTheta_{1},\dots, \bTheta_{N}$.

\subsection{Choice of Tuning Parameters}
\label{sm subsec:HMC_tuning}
Practical performances of HMC samplers can be sensitive to the choice of hyperparameters $\delta t$, $L$ and $\bM$.
To reduce sensitivity on particular choices of the hyperparameters, in each MCMC iteration we randomly sample $\delta t \sim \Unif(0, 0.01)$, $L \sim \Poisson (5)$ truncated in the range $[1,10]$ and following usual practices set $\bM$ to be the identity matrix.

\newpage

\section{Extended Simulation Studies}
\label{sm sec:simstudy}
\subsection{True Simulation Settings}
\label{sm subsec:sim_truths}

\paragraph*{Details on shared $\bLambda$:}
\begin{description}[leftmargin=7pt,topsep=3pt, partopsep=0pt,itemsep=2pt]
\item[Scenario 1:] We fill a consecutive 25\% of the elements in each column of $\bLambda$ with independent samples from $\Unif(-2,2)$ and set the rest as zero.
We choose the starting position of the consecutive non-zero elements randomly.    
To avoid having rows with all zero elements, we fill a randomly chosen 5 elements from each null row with iid samples from $\Unif(-2,2)$.
A typical example for a $200\times 20$ $\bLambda$ is shown in the leftmost panel of Figure \ref{sm fig:typical_lambda}.

\item[Scenario 2:] We separately fill the first and second consecutive $d/2$ values in each column of $\bLambda$.
For each half we follow the same strategy as scenario 1 and eliminate having entirely null rows using the aforementioned scheme.    
A typical example for a $200\times 20$ $\bLambda$ is shown in the middle panel of Figure \ref{sm fig:typical_lambda}.
From the correlation structures shown under the ``\textit{True}" panel of Figure \ref{fig:cormats} in the main paper, it can be seen that the above two scenarios induce very different correlation structures in the marginal distribution. 

\item[Scenario 3:] We randomly choose and fill 25\% of elements in each column of $\bLambda$ with independent samples from $\Unif(-2,2)$ and set the rest as zero.
We eliminate null rows using the same strategy of scenario 1.
A typical example for a $200\times 20$ $\bLambda$ is shown in the rightmost panel of Figure \ref{sm fig:typical_lambda}.
\end{description}
\begin{figure}[h]
\centering
\includegraphics[width=\linewidth]{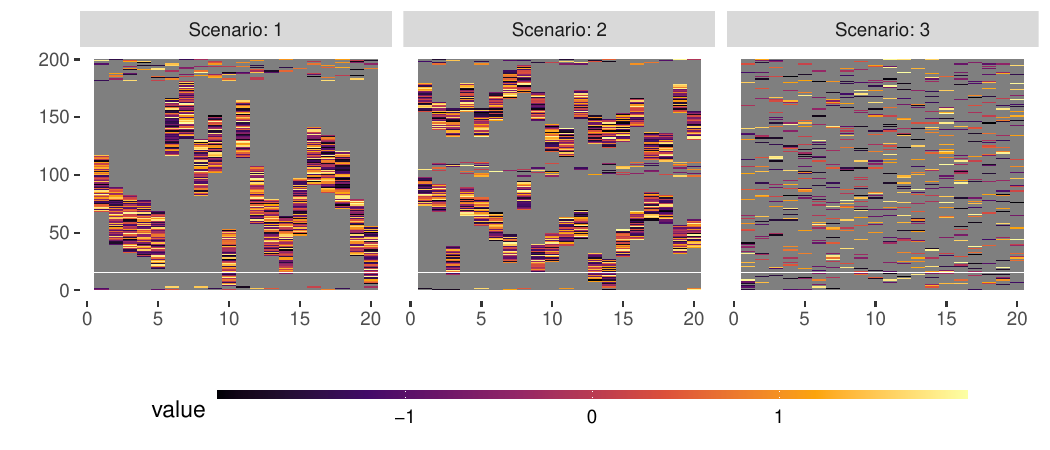}
\vskip-1ex
\caption{Simulation truths of $\bLambda$ across different scenarios under consideration.
The gray color encodes exact 0 values.}
\label{sm fig:typical_lambda}
\end{figure}

\vspace*{-1ex}
\paragraph*{Details on study-specific loadings:} 
\begin{description}[leftmargin=7.5pt,itemsep=-2pt]
\item[Slight misspecification:] We set the study-specific loading $\bPhi_{s}=\bLambda\bAs+\bE_{s}$ where we generate $\bAs=((a_{s,j,h}))$ as $a_{s,j,h}\simiid \mn(0,0.25^{2})$ and $\bE_{s}=((e_{s,j,h}))$ as $e_{s,j,h}\simiid \mn(0, 0.10^{2})$.
{Note that the SUFA model in \eqref{eq:ourmodel} assumes that study-specific loading matrices exactly admit the form $\bLambda\bAs$, i.e., they are in $\cspace{\bLambda}$.
Since we \textit{slightly} deviate from that assumption in this simulation setting, we refer to this as the  ``slight misspecification" case.}

\item[Complete misspecification:] 
{We generate $\bPhi_{s}$'s from the null-space of $\bLambda$ for all $s$ using the \texttt{MASS} package in \texttt{R}.
We additionally take $\bPhi_{s}\trans \bPhi_{s'}=\bzero$ for all $s\neq s'$ to ensure that the column spaces of study-specific loading matrices are non-overlapping.
We refer to this as the  ``complete misspecification" case since $\cspace{\bLambda} \cap \cspace{\bPhi_{s}}=\phi$ for all $s$.}
\end{description}

Note that the shared and study-specific signals are characterized by the magnitudes of the elements in the respective factor loading matrices. 
To investigate the signal ratios in shared versus study-specific covariance structures in the simulation truths, 
for each loading matrix we compute the \textit{mean of squares} for each column separately and
compare the medians of these mean squares across shared and study-specific
loading matrices. 
In Figure \ref{sm fig:colmeans_loadings} we plot the histogram
of these medians across independent replicates.
The figure indicates that the elements in the shared and study-specific loading matrices are comparable in magnitude.
This implies that despite having similarities via the shared loading matrix $\bLambda$, the marginal covariance structures across studies have substantial differences.
\begin{figure}[H]
    \centering
    \includegraphics[width=\linewidth]{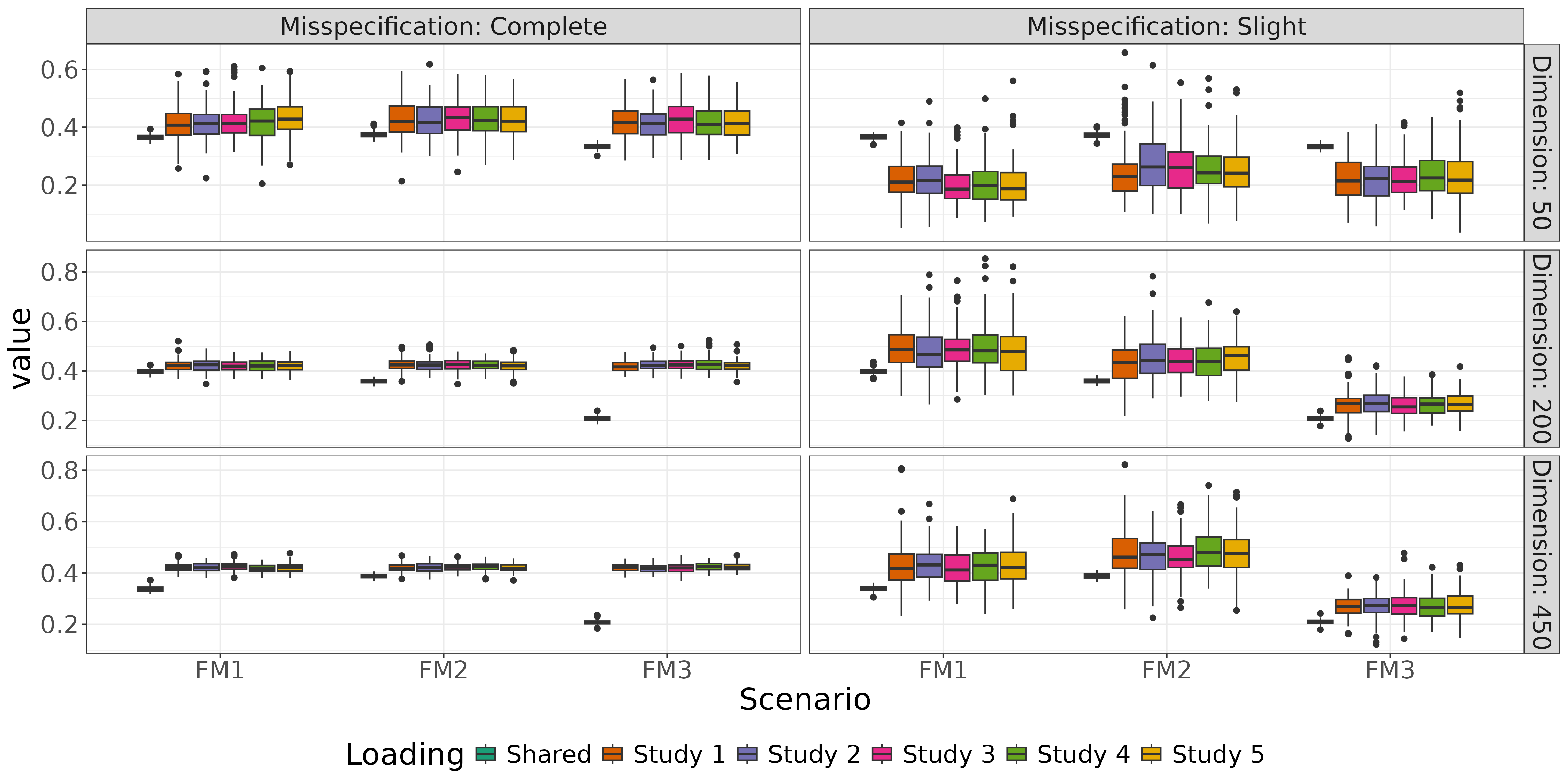}
    \caption{Investigation of the signal-to-noise ratios in shared versus study-specific loadings in the simulation truths}
    \label{sm fig:colmeans_loadings}
\end{figure}

\paragraph*{Details on idiosyncratic variance $\bDelta$:} We set the diagonals of $\bDelta$ as $0.50$.

\subsection{Recovering Shared Loading Matrix}
\label{sm sec:additional_sim}
Completing the discussion in Section \ref{sec:simstudies} in the main paper, here we include details on an extended simulation study showcasing recovery of the shared factor loading matrix $\bLambda$.
We apply the post-processing strategy discussed in Section \ref{subsec:rot_ambig} on the MCMC samples to align them with respect to a common orthogonal rotation and then use the mean of the post-processed MCMC samples as a point estimate.
The uncertainty quantification via the posterior samples guide inference on whether entries of $\bLambda=((\lambda_{j,h}))$ are $0$ or not;
following common practice in Bayesian inferences on sparsity patterns in matrices \citeplatex{KsheeraSagar2021precision}, we set $\lambda_{j,h}=0$ if its 95\% posterior credible interval includes 0.
\begin{figure}[H]
\centering
\includegraphics[width=.9\linewidth]{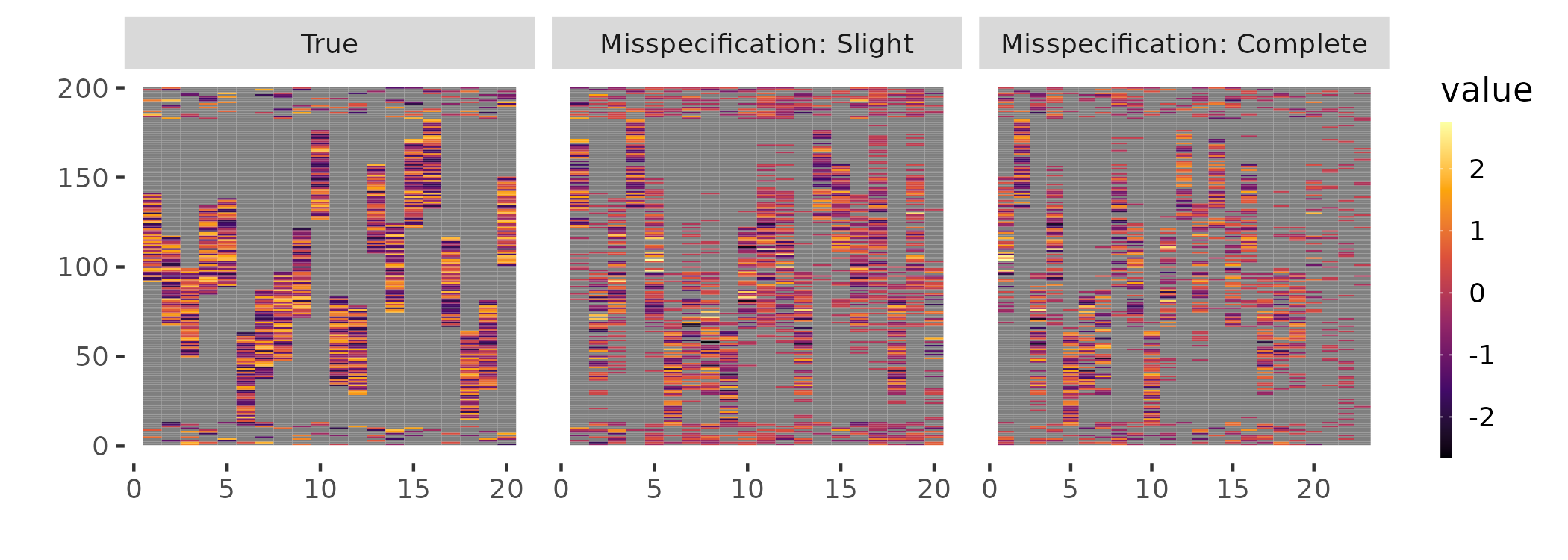}
\vskip-3.25ex
\includegraphics[width=.9\linewidth]{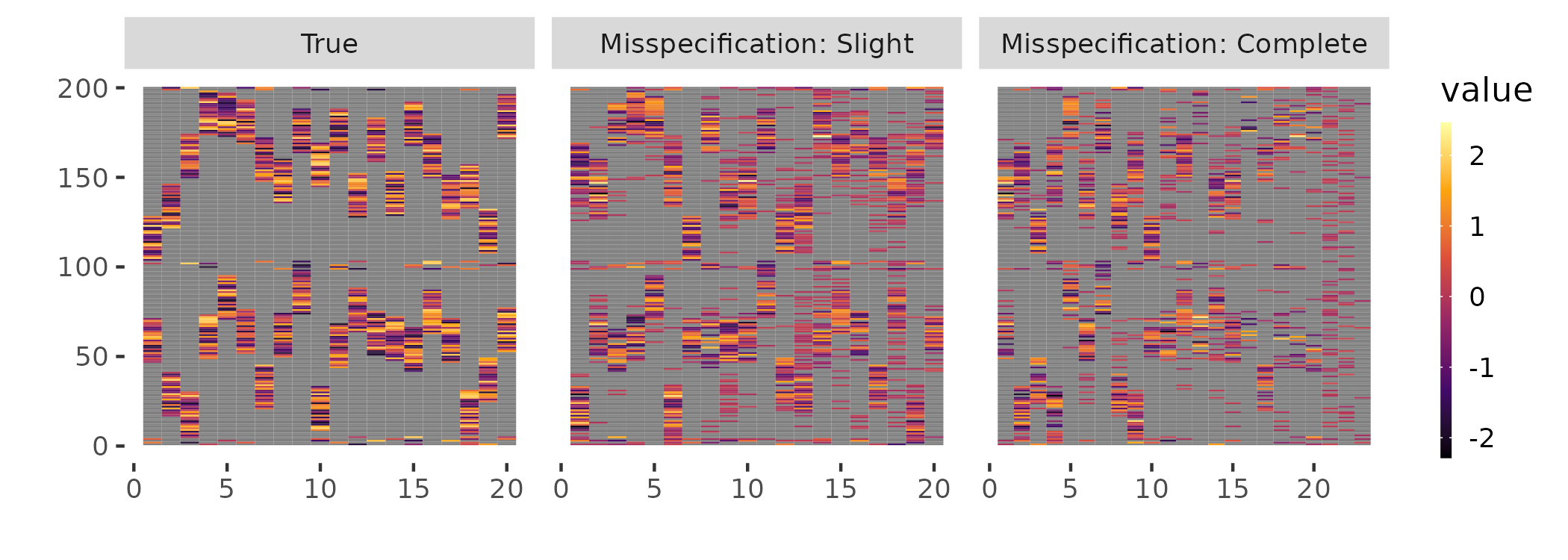}
\vskip-3.25ex
\includegraphics[width=.9\linewidth]{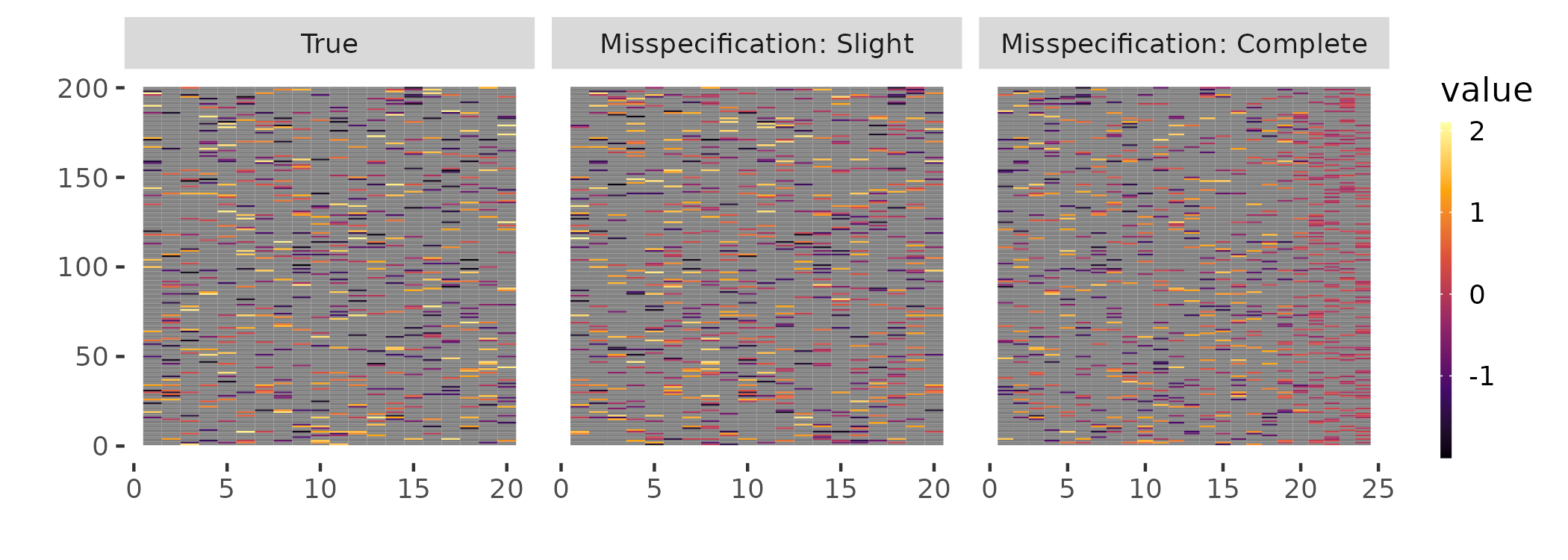}
\vskip-3.75ex
\caption{The true and estimated shared factor loading matrices: The rows correspond to the three different scenarios, viz., FM1, FM2 and FM3 under consideration.
In each row, the true factor loading matrix is shown in the leftmost plot; the middle and right plots are the point estimates in the slight and complete misspecified cases, respectively.}
\label{sm fig:loadmats}
\end{figure}

{Figure \ref{sm fig:loadmats} plots the shared factor loading matrices $\bLambda$ recovered by SUFA. 
In each setting, as a representative of the 100 independent replicates, we choose the one with the median Frobenius norm from the truth with respect to the shared covariance matrix $\bSigma=\bLambda\bLambda\trans+\bDelta$.
We use the approach elaborated in Section \ref{subsec:rot_ambig} of the main paper on the MCMC samples of $\bLambda$ to obtain a point estimate.
For sake of brevity, we only show results for $d=200$ in Figure \ref{sm fig:loadmats}.
The heatplots under the ``\textit{True}" panel of Figure \ref{sm fig:loadmats} show the simulation truths of the three different correlation structures;  
the ``\textit{Slight}" and ``\textit{Complete}" \textit{misspecification} panels show the recovered $\bLambda$ by the SUFA model for the two model misspecification types under consideration.}


Although the post-processing scheme aligns the MCMC samples of $\bLambda$ with respect to a common orthogonal rotation, {from Figure \ref{sm fig:loadmats}} it is not immediately obvious whether the estimate $\wh{\bLambda}$ is well-aligned with the simulation truth.
This can be better assessed by further regressing the columns of the true $\bLambda$ on $\wh{\bLambda}$. The accuracy can then be quantified in terms of the coefficient of determination $R^{2}$; see Figure \ref{fig:adj_R} in the main text. 

\subsection{Recovering Study-specific Loading Matrices}
\label{sm subsec:additional_sim_studyspecific}
In this section we study how well the study-specific loading matrices are recovered by the different methods.
As BFR and PFA do not have study-specific factor components we compare B-MSFA and SUFA only.
From the posterior MCMC samples of the study-specific loading matrices we obtain a point estimate in the same manner described in Section \ref{sm sec:additional_sim} and
use the coefficient of determination as the performance-metric.

\begin{figure}[H]
\includegraphics[trim={0cm .1cm .1cm 0cm}, clip,width=\linewidth] {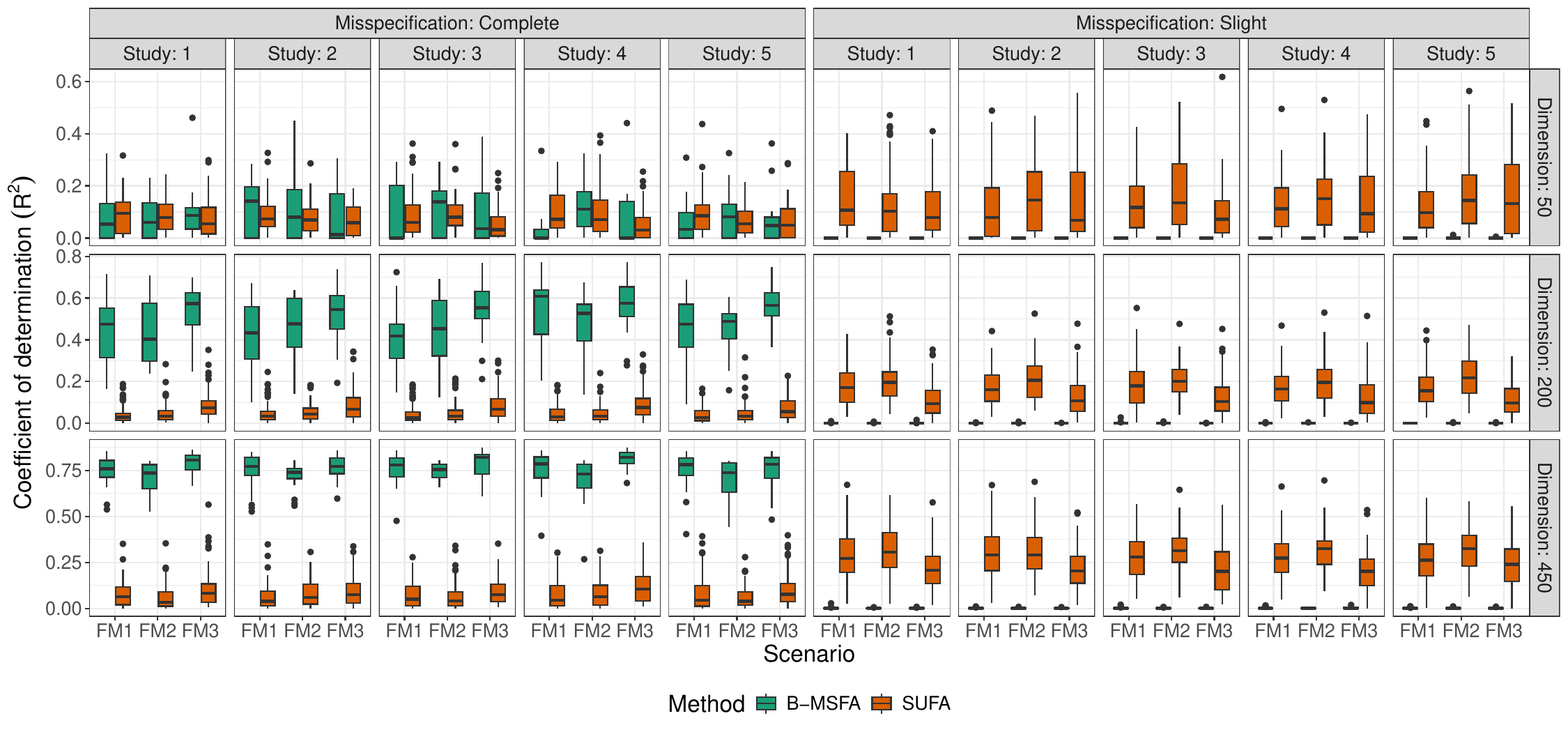}
\vskip-1ex
\caption{Coefficient of determination $R^{2}$ from regressing the columns of the true study-specific loading matrices on their estimated counterparts across all simulations settings for B-MSFA and SUFA.}
\label{fig:adj_R_studyspecific}
\end{figure}
B-MSFA recovers the study-specific loading matrices uniformly better than SUFA in the complete misspecification case.
This is expected as the simulation truths of the study-specific loading matrices are from the null space of $\bLambda$, see Section \ref{sm subsec:sim_truths} for details.
This is effectively a worst case for our proposed methodology as
SUFA assumes that the study-specific loadings are in the column space of the shared loading matrix $\bLambda$. 
Hence, SUFA can not recover the study-specific loadings in this setup.

On the contrary, there are overlaps between the column spaces of $\bLambda$ and the study-specific loading matrices in the slight misspecification case.
Performance of SUFA is more favorable in this case as B-MSFA is susceptible to the information switching issue.

\subsection{Scalability Analysis}
\label{sm subsec:scalability_sims}
In this section we study the scalability of the SUFA model with increasing sample sizes.
We generate data from the ``complete misspecification" case  in ``scenario 1", see Section \ref{sm subsec:sim_truths} for details.
We vary $d=50, 200$ and $450$.
For each $d$ we set $\ns=\Poisson(m\times \frac{d}{S})$, 
for each $s=1,\dots,S$ and repeat the experiments for $m\in \{1,10,25\}$.
We generate $\qs\sim \Poisson(\frac{q}{S})$ where $q=10$ for the $d=50$ and $q=20$ otherwise.
We fit SUFA using the HMC-within-Gibbs sampler in Algorithm \ref{algo:hmc_msfa} and report the boxplots of execution times (in minutes) across independent replicates in Figure \ref{sm fig:scalability_sims}.
We used a system with $13\th$ Gen Intel(R) Core(TM) i9-13900K CPU and 128GB RAM.
\begin{figure}[H]
\centering
\includegraphics[width=\linewidth]{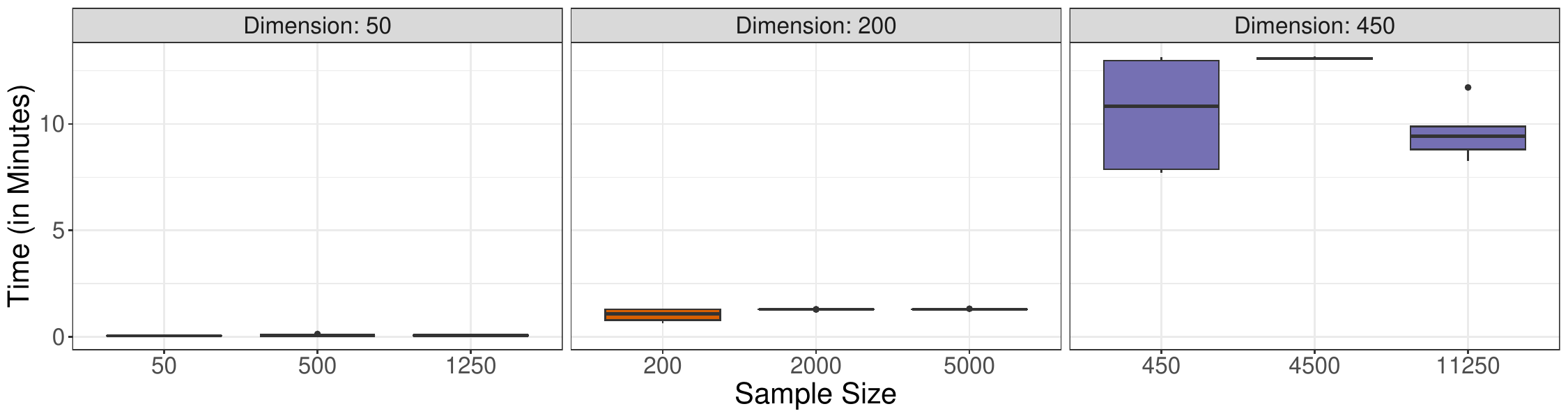}
\caption{Execution times (in minutes) of SUFA implementation with increasing sample sizes and dimensions.} 
\label{sm fig:scalability_sims}
\end{figure}
As expected, we see that the execution time is increasing with dimension. 
However, it remains stable for increasing sample sizes for fixed dimensions.
This is because our proposed sampler uses the sufficient statistics calculated before the MCMC iterations and becomes free of the sample size.
This indicates that our proposed approach is scalable to massive sample sizes.

\subsection{\changen{Varying Idiosyncratic Variances across Studies}}
\label{sm subsec:varying_resid}
Our proposed SUFA model in equation \eqref{eq:ourmodel_marginal} assumes fixed idiosyncratic variance $\bDelta=\diag(\delta_{1}^{2},\dots, \delta_{d}^{2})$ across studies unlike several approaches in the literature that allow the idiosyncratic variance component to vary across studies as well \citeplatex{devito2019multi, devito2018bayesian, alejandra2022, grabski2020bayesian}.
In this section, we study the effect of model misspecification in our SUFA model where idiosyncratic variances vary across studies.
We thus generate $\bY_{s,i}$'s from the MSFA model \eqref{eq:MSFA_model} as described in the main paper.
First, we sample $\delta_{s,j}^{2}\simiid \Unif(0.2, U_{\delta})$ where $\bDelta_{s}=\diag(\delta_{s,1}^{2},\dots,\delta_{s,d}^{2})$ is the diagonal matrix of idiosyncratic variances in study $s$.
We repeat the experiment for $U_{\delta} \in \{0.8, 1.5, 3, 5\}$.
Notably a higher value of $U_{\delta}$ implies greater misspecification from the SUFA model.
The $\bPhi_{s}$'s sampled from the null space of $\cspace{\bLambda}$ as described in the \textit{Complete misspecification} part in Section \ref{sm subsec:sim_truths}
for all $s=1,\dots,S$.
We do these simulation studies for $d=50$.
\begin{figure}[h]
    \centering
    \includegraphics[width=\linewidth]{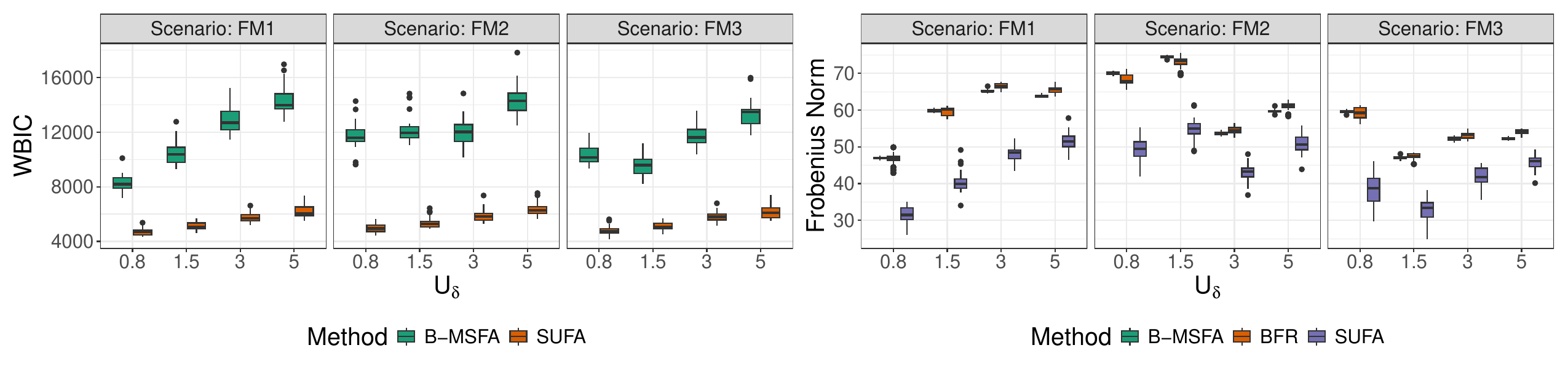}
    \caption{Comparing SUFA with B-MSFA and BFR across different simulation scenarios (FM1, FM2 and FM3): The left and right panels show the boxplots of the WBIC values and the boxplots of Frobenius norms between the true and estimated shared covariance $\bSigma$, respectively, for increasing misspecification from the SUFA model.}
    \label{sm fig:varying_resid}
\end{figure}

Figure \ref{sm fig:varying_resid} clearly shows that overall model fit and recovery of the shared covariance structure is evidently better for SUFA compared to B-MSFA and BFR.
Since the current BFR implementation provided by the authors only provides a MAP estimate and not MCMC samples, we could not estimate WBIC values for BFR.


\subsection{\changen{SUbspace Factor Analysis with Many Studies}}
\label{sm subsec:many_studies}
Note that any $d\times d$ positive definite matrix $\bSigma$ admits the decomposition   $\bSigma=\bLambda\bLambda\trans+\bDelta$ for some $d\times q $ matrix $\bLambda$ with $0\leq q \leq d$.
Letting $\bSigmas$ be the marginal covariance matrix in study $s$, we further assume that  the $\bSigmas$'s admit the low-rank plus diagonal decomposition $\bSigmas=\bLambdas\bLambdas\trans+\bDelta$ for $d\times \tilqs$ matrix $\bLambdas$ with $0\leq \tilqs \leq d$ for all $s=1,\dots,S$.
In the multi-study setting, most of the existing literature borrows information across studies by assuming $\cspace{\bLambdas}=\cspace{\bLambda}\cup \cspace{\tLambdas}$ such that $\bigcap_{s=1}^{S} \cspace{\tLambdas}$.
Heuristically,  $\bLambda$ explains the shared variation while $\tLambdas$'s explain study-specific variations in the sense  $\tLambdas$'s do not have shared column space.
Our SUFA framework assumes $\tLambdas=\bLambda\bAs$.
Theorem \ref{th:no_prior_support} ensures that under the condition $q\geq \sum_{s=1}^{S}\qs>0$, $\bigcap_{s=1}^{S} \cspace{\tLambdas}=\phi$.

When $S$ is large, it may not possible  to  meet the condition from Theorem \ref{th:no_prior_support}.
In that case, identifiability of shared versus study-specific components is not guaranteed.
However, estimation of marginal covariance matrices is possible. 
The SUFA framework allows borrowing of information across studies via the shared subspace.
\begin{figure}[h]
	\includegraphics[width=\linewidth]{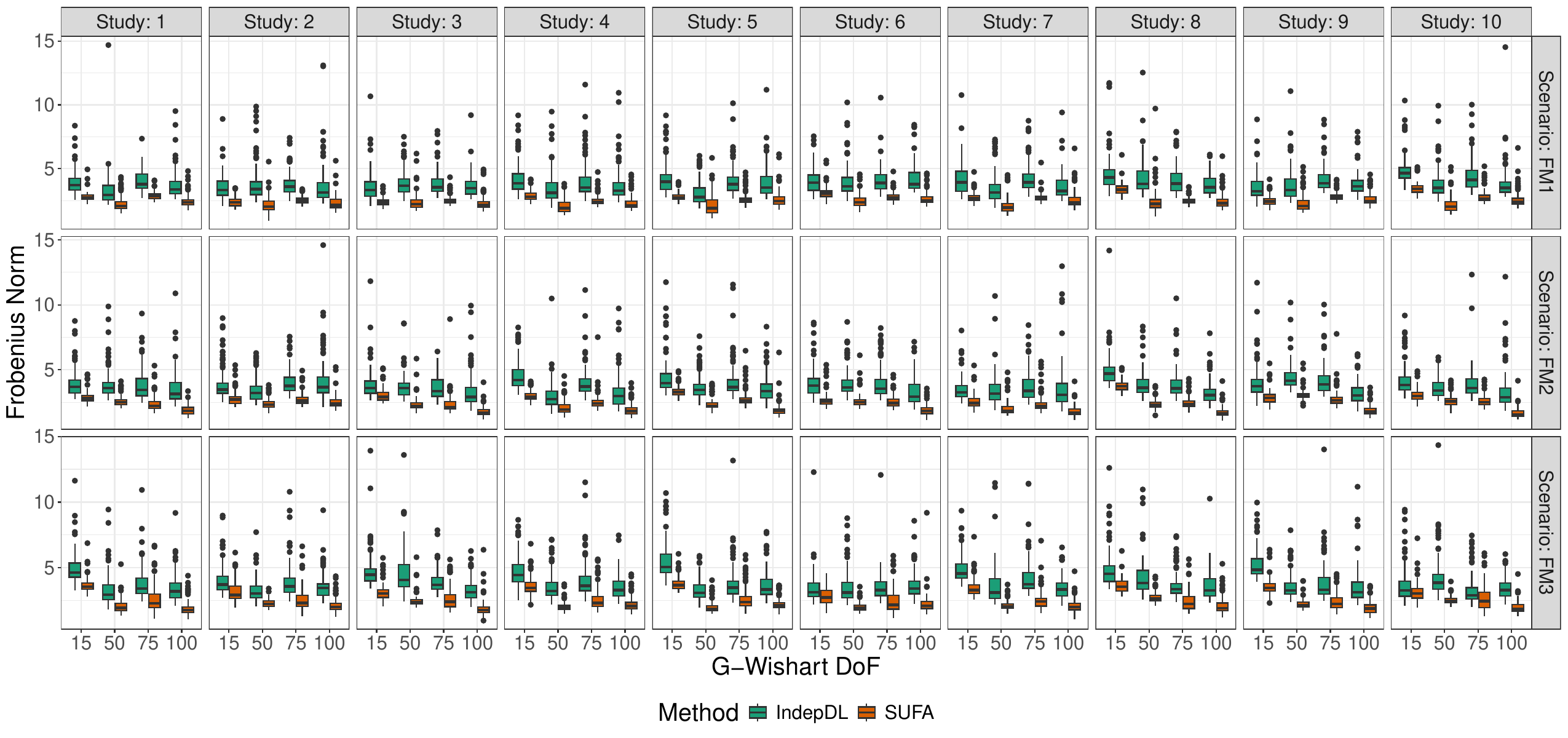}
	\caption{Comparing SUFA with a Bayesian sparse factor analysis approach in estimating marginal covariance matrices of each study: The Frobenius norm between the true and estimated marginal covariance matrices of each study is used as the metric of comparison.
		Increasing DoF implies more similarity across studies in the simulation truth.}
	\label{sm fig:many_studies}
\end{figure}

To see this, we let $\bSigma=\bLambda\bLambda\trans+\bDelta$ and sample $\bSigmas\simiid   \text{G-Wishart}_{\nu}(\bG, \frac{1}{\nu}\bSigma) $ for $s=1,\dots,S$ where  $\nu$, $\bG$, and $\frac{1}{\nu}\bSigma$ are the degrees of freedom (DoF), adjacency and scale matrices of  the G-Wishart distribution, respectively.
We set $\bG$ as the sparsity structure of $\bLambda$.
Then we sample $\bY_{s,1:\ns}\simiid \mn_{d}(\bzero,\bSigmas)$.
Thus the covariance structures of data across studies are not exactly the same but are similar.
We set $d=10$, $q=4$ and $S=10$ so that the condition of Theorem \ref{th:no_prior_support} can not be met.
To induce high-dimensional small sample size scenarios, we set $\ns=5$ in all settings.
We repeat the experiments for the three different structures $\bLambda$ shown in Figure \ref{sm fig:typical_lambda}.
Note that as $\nu$ increases, $\bSigmas$'s will concentrate around $\bSigma$. 
We thus vary $\nu\in \{15,50,75,100\}$ to regulate the extent of similarity across studies.

We use the SUFA model to jointly analyze data across studies and estimate the marginal covariance matrices.
Notably the marginal covariance matrices do not exactly satisfy the SUFA structure in equation \eqref{eq:ourmodel_marginal}.
While fitting the SUFA model on the data, we use the strategy discussed in Section \ref{subsec:specify_q} to obtain $\wh{q}$ but set $\wh{q}_{s}=2$ for all $s$ which violates the condition $\sum_{s}\qs\leq q$ of Theorem \ref{th:no_prior_support}.
We compare SUFA with a standard sparse Bayesian factor analysis model that separately estimates the marginal covariance matrices in each study.
We use the Frobenius norm between the true and estimated marginal covariance matrices of each study as the metric of comparison and in Figure \ref{sm fig:many_studies} report the boxplots of the norms across 100 repeat simulations in each setting.

Although, the identifiability between shared versus study-specific components is not guaranteed, SUFA provides more accurate estimates by borrowing information across related studies in every setting.
This is also coherent with Theorem \ref{sm thm:margibal_var} of the supplementary materials where we showed that the contraction rate $\varepsilon_{n}$ shrinks with increasing number of studies.
We also observe that as the DoF increases, the SUFA estimates tend to get more accurate.

\section{Supplementary to Section \ref{sec:application}}
\subsection{Heatmaps of the Gene Expressions}
\begin{figure}[H]
\centering
\includegraphics[width=\linewidth]{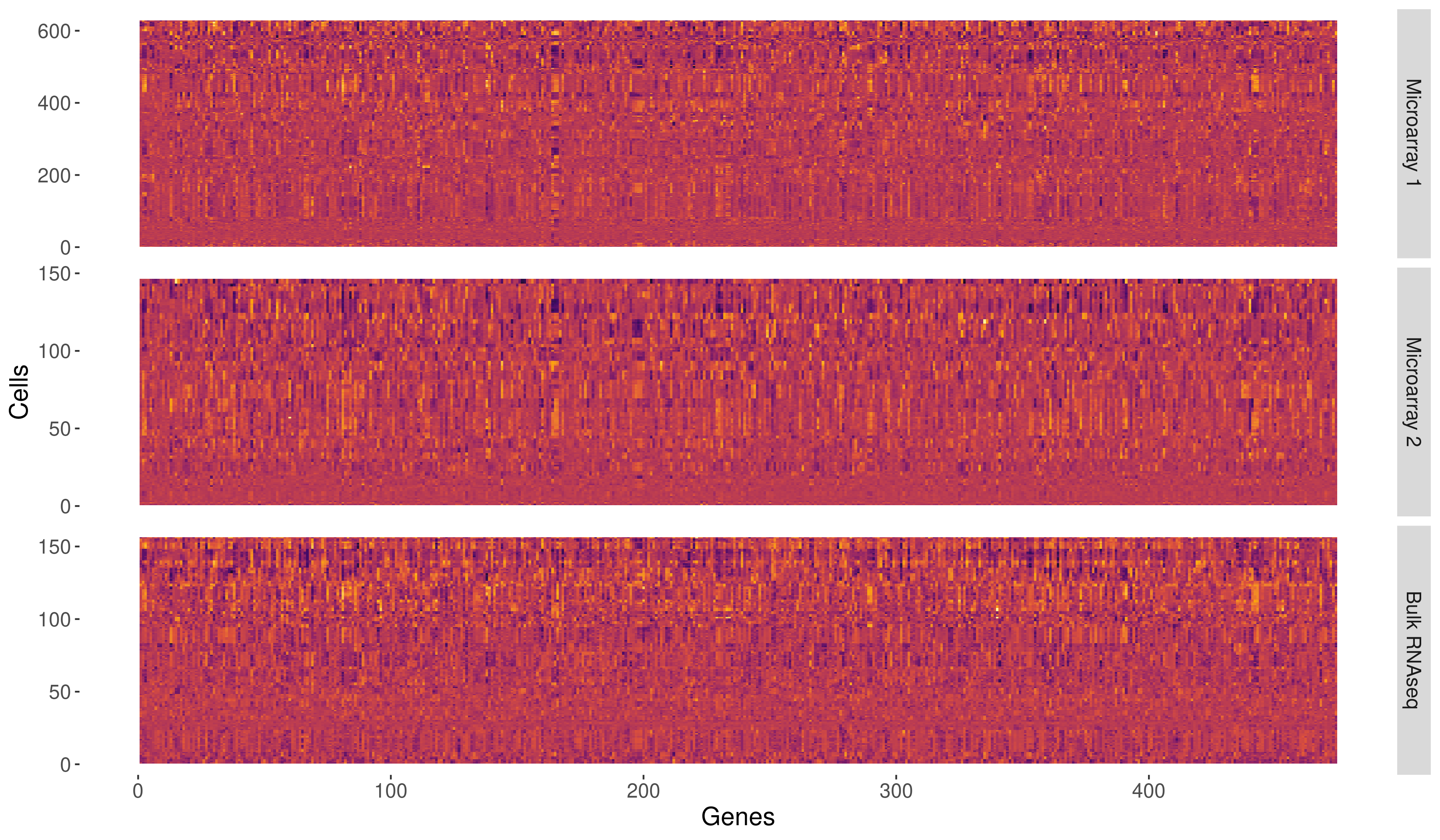}
\caption{Heatmaps of two microarray and a bulk RNAseq dataset on a common set of genes. 
The X and Y axes represent genes and cells across studies, respectively. } 
\label{fig:heatmap_originaldata}
\end{figure}

\subsection{Pre-processing the Data}
\label{sm subsec:data_preprocess}
The normalized dataset has more than $21,000$ gene expressions from $628$ and $146$ immune cells in the two microarray datasets, respectively, and more than $49,000$ genes from $156$ cells in the bulk RNAseq dataset. 
We made a $\log_{2}$ transformation of the data and filtered the top $5\%$ genes with highest variances using the \texttt{genefilter R} package \citeplatex{genefilter} from each of the datasets and considered the intersection of the filtered genes in our analysis, 
resulting in a $d=474$ dimensional problem.  
Since different cell-types exhibit very different gene expression profiles, we centered the gene-expressions separately within each cell type.

\subsection{Discussion on the Gene Network}
\label{sm subsec:networks}
In addition to corroborating qualitative results in the prior literature with a principled statistical analysis, our study also reveals new insights via integrating the data. 
Bcl2a1b and Bcl2a1d genes--two functional isoforms of the B cell leukemia 2 family member \citeplatex{bcl2} --  show strong positive correlations.
We find strong positive associations between genes Bub1, Ccna2, Ccnb2, Top2a, which corroborates findings from a study by \citetlatex{tcellleukemia} on adult human T cell leukemia. 
Interestingly, we find this group of genes strongly negatively correlated with those in the Gimap family, which are involved in lymphocyte development and play important roles in immune system homeostasis.
On the other hand, the analysis reveals that the Gimap class features positive within-group associations, supporting
recent studies that suggest  Gimap proteins may interact with each other in roles such as moving cellular cargo along the cytoskeletal network \citeplatex{gimap}.
Up-regulation of the Ccr2 gene has been found to be associated with cancer advancement, metastasis and relapse \citeplatex{ccr2}, 
whereas Ccr5 inhibitors exhibit negative association with lung metastasis of human breast cancer cell lines \citeplatex{ccr5}.
Coherent with these studies, we find strong positive correlation between Ccr2 and Ccr5.
We also observe Ccr2 to be highly positively correlated with Il18rap, supported by  genome-wide association studies identifying their susceptibility with coeliac diseases \citeplatex{il18rap}.
Additionally we see strong positive association between the structurally related genes Il18rap and Il18r1 within the Il18 receptor complex \citeplatex{il18}.
Interestingly, strong negative correlation is observed between the Il18 family of genes and the gene Cd81, which is required for multiple normal physiological functions \citeplatex{cd81}.

\subsection{Interpretation of Factor Loading Matrices}
\label{sm subsec:loading_interpretation}
Note that genes Il18r1,
Il18rap,
Klrk1,
Ccl5,
Slamf7,
Tbx21,
Ccr2,
Ccr5, etc.
located in the bottom of Figure \ref{subfig:all_loads} have negative loadings corresponding to the first latent factor in the shared space as well as in the bulkRNASeq space but negative loadings for Microarray 2.
To interpret the loading values we primarily focus on the two bottom-most genes Il18rap and Il18r1 with 
$\lambda_{g,h}$ and $\lambda_{g',h}$ being their respective shared loadings corresponding to factor $h$.
From Figure \ref{subfig:all_loads} we see that $\lambda_{g,h},\lambda_{g',h}\approx 0$ for $h>1$.
As covariances in factor analysis models are entirely characterized by the dot-product between the row-vectors of $\bLambda$,
 for the sake of argument we assume that the shared covariance between the two aforementioned genes is $\cov_{g,g'}\approx \lambda_{g,1} \lambda_{g',1}$.
For characterizing $\cov_{g,g'}$ the relative signs of $\lambda_{g,1}$ and $\lambda_{g',1}$ are meaningful, i.e.,  $\sgn(\lambda_{g,1})=\sgn(\lambda_{g',1})$ implies $\cov_{g,g'}>0$.
Accordingly we observe strong positive correlation between Il18rap and Il18r1 in the shared covariance matrix.

Now let $\lambda'_{g,1}$ and $\lambda'_{g',1}$ be the loadings of  Il18rap and Il18r1 respectively, corresponding to bulkRNASeq-specific factor $1$.
From Figure \ref{subfig:all_loads} we see that $\sgn(\lambda'_{g,1})=-\sgn(\lambda_{g,1})$ and $\sgn(\lambda'_{g',1})=-\sgn(\lambda_{g',1})$, i.e., 
opposite signs in the shared and study-specific loadings.
However, we have $\sgn(\lambda_{g,1}\lambda_{g',1})=\sgn(\lambda'_{g,1}\lambda'_{g',1})$ implying that the two genes exhibit positive correlation in the bulkRNASeq-specific covariace part as well.

Similar phenomena can also be observed between genes Fosb,
Nr4a1,
Egr1,
Dusp1,
Hspa1a, etc. located towards the top of Figure \ref{subfig:all_loads} that exhibit positive loadings corresponding to latent factor 5 in the shared space but negative loadings corresponding to the bulkRNASeq-specific latent factor 2.
However, the relative signs of the loadings within the shared space are the same implying that they are positively correlated.
The same is observed for the bulkRNASeq-specific covariance as well.
Therefore, we conclude that the signs of the loadings must be interpreted together and not individually. 

The magnitudes of the loadings indicate the absolute value of the respective variances and covariances.
For example the genes  Klrk1 and Ccl5 (third and fourth from the bottom in Figure \ref{subfig:all_loads}) have higher shared loadings compared to most other genes and therefore they are expected to exhibit higher shared covariance and variances.

\baselineskip=14pt
\bibliographystylelatex{natbib}
\bibliographylatex{{main}}
\end{document}